\documentclass[a4paper,11pt]{article}

\usepackage[english]{babel}
\usepackage[a4paper]{geometry}
\usepackage{enumerate}  
\usepackage{amsmath}
\usepackage{amsthm}
\usepackage{amssymb}
\usepackage{hyperref}   

\usepackage{amssymb}
\usepackage{color}
\usepackage[utf8]{inputenc}
\usepackage{graphicx}
\usepackage{tikz}
\usepackage{verbatim} 

\usepackage{mathabx} 
\usepackage{hyperref}

\usetikzlibrary{decorations.markings}

\def\bR{\mathbb{R}}
\def\bN{\mathbb{N}}

\def\bZ{\mathbb{Z}}

\def\cV{\mathcal{V}}
\def\cF{\mathcal{F}}
\def\cG{\mathcal{G}}
\def\cD{\mathcal{D}}
\def\cJ{\mathcal{J}}
\def\cL{\mathcal{L}}
\def\cN{\mathcal{N}}
\def\cE{\mathcal{E}}
\def\cK{\mathcal{K}}
\def\cH{\mathcal{H}}

\def\eps{\varepsilon}
\def\ph{\varphi}

\def\wt{\widetilde}
\def\wh{\widehat}
\def\wch{\widecheck}

\def \pn {\varphi }
\def \whpn{\widehat{\varphi}}

\def \an {\mathfrak{a}_{0} }

\def \whpn { \widehat{ \varphi} }

\def \cFpn {\mathcal{F}_{\bot \varphi}^{\leq N} }

\newcommand{\gl}{g_L}
\newcommand{\cgl}{\check{g}_L}

\def\be{\begin{equation}}
\def\ee{\end{equation}}

\newtheorem{theorem}{Theorem}[section]  

\newtheorem{prop}[theorem]{Proposition}
\newtheorem{lemma}[theorem]{Lemma}
\numberwithin{equation}{section}

\begin{document}

\title{Bose-Einstein Condensation with Optimal Rate for \\  Trapped Bosons in the Gross-Pitaevskii Regime}

\author{Christian Brennecke$^1$, Benjamin Schlein$^2$, Severin Schraven$^3$ \\
\\
Department of Mathematics, Harvard University, \\
One Oxford Street, Cambridge MA 02138, USA$^{1}$ \\
\\
Institute of Mathematics, University of Zurich, \\
Winterthurerstrasse 190, 8057 Zurich, Switzerland$^{2,3}$}

\maketitle

\begin{abstract}
We consider a Bose gas consisting of $N$ particles in $\bR^3$,  trapped by an external field and interacting through a two-body potential with scattering length of order $N^{-1}$. We prove that low energy states exhibit complete Bose-Einstein condensation with optimal rate, generalizing previous work in \cite{BBCS1, BBCS4}, restricted to translation invariant systems. This extends recent 
results in \cite{NNRT}, removing the smallness assumption on the size of the scattering length.
\end{abstract}

\section{Introduction and Main Results} \label{sec:intro}

We consider a system of $N \in \bN$ bosons trapped by an external potential in the Gross-Pitaevskii regime; the particles interact through a repulsive two-body potential with scattering length of order $N^{-1}$. The Hamilton operator as the form 
\begin{equation}\label{eq:defHN} H_N = \sum_{j=1}^N \left[ -\Delta_{x_j} + V_{ext} (x_j) \right] + \sum_{1\leq i<j\leq N} N^2 V(N(x_i -x_j)).
\end{equation}
and it acts on a dense subspace of the Hilbert space $L^2_s (\bR^{3N})$, the subspace of $L^2 (\bR^{3N})$ consisting of functions that are symmetric with respect to permutations of the $N$ particles. The confining potential $V_{ext} \in L^\infty_\text{loc} (\bR^3)$ diverge to infinity, as $|x| \to \infty$ (more precise conditions on $V_{ext}$ will be introduced later on). Furthermore, we assume $V\in L^3(\bR^3)$ to be pointwise non-negative, spherically symmetric and compactly supported (but our results could easily be extended to potentials decaying sufficiently fast at infinity). 

The scattering length of $V$ is defined through the zero-energy scattering equation 
\begin{equation}\label{eq:0en} \left[ - \Delta + \frac{1}{2} V (x)  \right] f (x) = 0 \end{equation}
with the boundary condition $f (x) \to 1$, as $|x| \to \infty$. For $|x|$ large enough (outside the support of $V$) we find \[ f(x) = 1- \frac{\frak{a}_0 }{|x|} \]
where the constant $\frak{a}_0 > 0$ is known as the scattering length of $V$. A simple computation shows that  
\begin{equation}\label{eq:a0}  8 \pi \frak{a}_0  =  \int_{\mathbb{R}^3} V(x) f(x) dx .  \end{equation}
Moreover, by scaling, (\ref{eq:0en}) implies that 
		\[ \left[ - \Delta + \frac{1}{2} N^2 V (Nx) \right] f (Nx) = 0 \]
so that the scattering length of $N^2 V (N.)$ is given by $\frak{a}_0 /N$.  

From \cite{LSY}, it is known that the ground state energy $E_N$ of the Hamilton operator \eqref{eq:defHN} satisfies		
\begin{equation}\label{eq:LSY} \lim_{N\to \infty} \frac{E_N}N =  \inf_{\psi\in H^1(\bR^3): \|\psi\|_2=1}\cE_{GP}(\psi), 
\end{equation}
where $\cE_{GP}$ denotes the Gross-Pitaevskii energy functional 
\begin{equation} \label{eq:defGPfunctional}
\mathcal{E}_{GP}(\psi) = \int_{\bR^3} \left(  \vert \nabla \psi(x) \vert^2 + V_{ext}(x) \vert \psi(x)\vert^2 + 4\pi\an \vert \psi(x)\vert^4 \right) dx.
\end{equation}

For the rest of the paper, we will lighten the notation and write $\int$ instead of $\int_{\mathbb{R}^3}$. Furthermore, we will write $\Vert \cdot \Vert$ for the $L^2$-norm and indicate other $L^p$-norms by a suitable subscript.
The Gross--Pitaevskii functional $\cE_{GP} $ admits a unique normalized, strictly positive minimizer $\pn \in L^2(\bR^3)$. It satisfies the Euler-Lagrange equation
\begin{equation}\label{eq:euler} -\Delta \pn + V_{ext} \pn + 8\pi \frak{a}_0 |\pn|^2 \pn = \eps_{GP} \pn \end{equation} 
with the Lagrange multiplier $\eps_{GP} := \cE_{GP} (\ph) + 4\pi \frak{a}_0 \| \pn \|_4^4$. As first shown in \cite{LS1}, the ground state of \eqref{eq:defHN} exhibits complete Bose-Einstein condensation in the state $\pn$. More precisely, if $\gamma^{(1)}_N = \operatorname{tr}_{2,\dots,N}|\psi_N\rangle\langle\psi_N|$ denotes the one-particle reduced density associated with the ground state  of \eqref{eq:defHN}, then 
\begin{equation}\label{eq:BEC0}\lim_{N\to\infty} \langle\pn, \gamma_N^{(1)}\pn\rangle = 1.  \end{equation}
This implies that, in the ground state of (\ref{eq:defHN}), the fraction of particles in the state $\pn$ approaches one, as $N \to \infty$. The convergence in \eqref{eq:BEC0} was later extended in \cite{LS2,NRS} to any sequence $\psi_N$ of approximate ground states, satisfying 
\[ \lim_{N\to \infty} \frac1N \langle \psi_N, H_N\psi_N\rangle = \cE_{GP}(\pn).\]

For translation invariant systems (particles trapped in the box $\Lambda = [0 ;1]^3$, with periodic boundary conditions), \eqref{eq:LSY} (stating, in this case, that $E_N / N \to 4 \pi \frak{a}_0$) and \eqref{eq:BEC0} (establishing Bose-Einstein condensation in the zero-momentum mode $\ph (x) = 1$ for all $x \in \Lambda$) have been proved in \cite{BBCS1, BBCS4,H} to hold with the optimal rate of convergence. This result was recently generalised in \cite{NNRT} (extending the approach of \cite{BFS,FS}) to trapped systems described by the Hamilton operator (\ref{eq:defHN}), under the assumption of sufficiently small scattering length $\mathfrak{a}_0$. 


Our goal in this paper is to obtain optimal bounds on the rate of convergence in (\ref{eq:LSY}) and (\ref{eq:BEC0}), with no restriction on the size of the scattering length. To reach this goal, we adapt and extend the approach developed in \cite{BBCS4} for the translation invariant case. To this end, we will impose, throughout the rest of this paper, the following conditions:   
\begin{equation}\label{eq:asmptsVVext}
\begin{split}
(1) \;& V\in L^3(\bR^3) , V(x)\geq 0 \text{ for a.e.} \;x\in\bR^3, V \text{ spherically symmetric, }  \operatorname{supp}(V) \text{ compact},  \\
(2)\; & V_{ext} \in C^3(\bR^3; \bR), V_{ext}(x) \to \infty \text{ as } |x| \to \infty, \\
& \exists\;  C > 0 \; \forall \;  x,y\in\bR^3: V_{ext}(x+y) \leq C (  V_{ext}(x)+C)(  V_{ext}(y)+C), \\
& \nabla V_{ext} , \Delta V_{ext} \text{ have at most exponential growth as }  |x|\to \infty.
\end{split}
\end{equation}
Note in particular that all polynomials in $x^2$ with positive leading coefficient satisfy condition $(2)$ in \eqref{eq:asmptsVVext}. The assumptions that $V_{ext}(x+y) \leq C (  V_{ext}(x)+C)(  V_{ext}(y)+C)$ and that $\nabla V_{ext} , \Delta V_{ext}$ grow at most exponentially allow for certain simplifications of our analysis and are needed for technical reasons only. 

\begin{theorem}\label{thm:main} Assume \eqref{eq:asmptsVVext} and let $E_N$ denote the ground state energy of (\ref{eq:defHN}). Then, there exists a constant $C > 0$ such that 
\begin{equation}\label{eq:ENlow}
E_N \geq N \cE_{GP} (\pn) - C 
\end{equation}
and 
\[ H_N\geq N\mathcal{E}_{GP}(\pn) + C^{-1}\sum_{i=1}^N \big( 1- |\pn\rangle\langle\pn|_i\big)- C.\]
In particular, if $\psi_N\in L^2_s(\bR^{3N})$ with $\|\psi_N\|=1$ is an sequence of approximate ground states such that 
\[ \langle \psi_N, H_N\psi_N\rangle \leq N\cE_{GP}(\pn) + \zeta, \]
for a $\zeta > 0$, then the reduced density $\gamma_N^{(1)}$ associated with $\psi_N$ satisfies 
\[ 1-\langle\pn, \gamma_N^{(1)}\pn\rangle \leq \frac{(C+\zeta)}N.\]
\end{theorem}

{\it Remark:} Our techniques could also be used to prove an upper bound for $E_N$ matching (\ref{eq:ENlow}), implying that $|E_N - N \cE_{GP} (\pn) | \leq C$. We do not show it, because it would require some non-trivial additional work (this is a consequence of our choice, leading to some technical simplifications, to work on the Fock space $\cF^{\leq N}$ rather than on $\cF^{\leq N}_{\perp \ph}$, where we impose orthogonality to $\ph$; we will explain this point in the next section) and because it is already established in \cite{NNRT} (the upper bound there does not require restrictions on the size of the 
potential)

{\it Remark:} We plan to apply Theorem \ref{thm:main} in a separate paper to determine the low-energy spectrum of (\ref{eq:defHN}) and to establish the validity of the predictions of Bogoliubov theory, extending recent results obtained in \cite{BBCS2, BBCS3} for the translation invariant setting. 

\vspace{0.5cm}

{\it Acknowledgements.} B.S. gratefully acknowledge support from the NCCR SwissMAP and from the Swiss National Foundation of Science through the SNF Grant ``Effective equations from quantum dynamics'' and the SNF Grant ``Dynamical and energetic properties of Bose-Einstein condensates''.

\section{Excitation Hamiltonians} 

We introduce the bosonic Fock space 
\[ \cF = \bigoplus_{n \geq 0} L^2_s (\bR^{3n}) = \bigoplus_{n \geq 0} L^2 (\bR^3)^{\otimes_s n} \]
On $\cF$, we consider creation and annihilation operators, satisfying the 
canonical commutation relations 
\begin{equation}\label{eq:ccr} [a (g), a^* (h) ] = \langle g,h \rangle , \quad [ a(g), a(h)] = [a^* (g), a^* (h) ] = 0 \end{equation}
for all $g,h \in L^2 (\bR^3)$. We also introduce position and momentum-space operator-valued distributions $a_x, a_x^*$ and $\hat{a}_p, \hat{a}^*_p$, for $x,p \in \bR^3$, so that 
\[ a(f) = \int \bar{f} (x) \, a_x \, dx = \int \bar{\hat{f}} (p) \, \hat{a}_p \, dp , \qquad a^* (f) = \int f(x) \, a_x^* \, dx = \int  \hat{f} (p) \, \hat{a}_p^* \, dp.  \]

In terms of these operator-valued distributions, the number of particles operator $\cN$, defined by $(\cN \Psi)^{(n)} = n \Psi^{(n)}$ for every $\Psi \in \cF$, takes the form 
\[ \cN = \int a_x^* a_x \, dx = \int \hat{a}_p^* \hat{a}_p \, dp . \]
More generally, given an operator $A:L^2(\mathbb{R}^3)\to L^2(\mathbb{R}^3)$ with kernel $A(x;y)$, we define its second quantization $d\Gamma (A)$ acting in $\cF$ through
		\[ d\Gamma(A) = \int dxdy\, A(x;y) a^*_xa_y.\]
In particular, $\cN = d\Gamma(1)$ is the second quantization of the identity operator. 

It is simple to check that creation and annihilation operators are bounded, with respect to $\cN^{1/2}$. In fact, we find that 
\begin{equation}\label{eq:abd} \| a (f) \Psi \| \leq \| f \| \  \| \cN^{1/2} \Psi \|, \quad \| a^* (f) \Psi \| \leq \| f \| \ \| (\cN+1)^{1/2} \Psi \|.
\end{equation}

To describe excitations of the Bose-Einstein condensate, we also define the truncated Fock spaces
\[ \cF^{\leq N}  = \bigoplus_{n = 0}^N L^2 (\bR^3)^{\otimes_s n} , \qquad \text{and } \quad \cF^{\leq N}_{\perp \ph}  = \bigoplus_{n = 0}^N L_{\perp \ph}^2 (\bR^3)^{\otimes_s n} \]
defined over $L^2 (\bR^3)$ and, respectively, over $L^2_{\perp \ph} (\bR^3)$, the orthogonal complement of the condensate wave function $\ph$ in $L^2 (\bR^3)$ ($\ph$ is the unique minimizer of (\ref{eq:defGPfunctional})). Since they do not preserve the number of particles, creation and annihilation operators are not well-defined on $\cF^{\leq N}$ and $\cF^{\leq N}_{\perp \ph}$ (but notice that products of a creation and an annihilation operators are well-defined). They are replaced by modified creation and annihilation operators 
\[ b^* (f) = a^* (f) \, \sqrt{\frac{N-\cN}{N}} , \qquad \quad b (f) = \sqrt{\frac{N- \cN}{N}} \, a (f) . \]
For every $f \in L^2 (\bR^3)$, $b(f)$ and $b^* (f)$ map $\cF^{\leq N}$ to $\cF^{\leq N}$. If moreover $f \perp \ph$, we also find $b (f), b^* (f) : \cF^{\leq N}_{\perp \ph} \to \cF^{\leq N}_{\perp \ph}$. From (\ref{eq:abd}) we have 
\[ \| b (f) \Psi \| \leq \| f \| \| \cN^{1/2} \Psi \|, \quad \| b^* (f) \Psi \| \leq \| f \| \| (\cN+1)^{1/2} \Psi \|.
\] 
Also here, it is convenient to introduce operator-valued distributions $b_x , b_x^*$, for any $x \in \bR^3$, and, in momentum space $\hat{b}_p, \hat{b}^*_p$, for any $p \in \bR^3$. They satisfy the commutation relations (focussing here on position space operators) 
\begin{equation}\label{eq:comm-b}
\begin{split}  [ b_x, b_y^* ] &= \left( 1 - \frac{\cN}{N} \right) \delta (x-y) - \frac{1}{N} a_y^* a_x, \hspace{1cm}[ b_x, b_y ] = [b_x^* , b_y^*] = 0 
\end{split} \end{equation}
and
\begin{equation}\label{eq:comm-b2}
\begin{split}
[b_x, a_y^* a_z] &=\delta (x-y) b_z, \qquad 
[b_x^*, a_y^* a_z] = -\delta (x-z) b_y^*.
\end{split} \end{equation}
In particular, it follows that $[ b_x, \cN ] = b_x$ and $[ b_x^* , \cN ] = -b_x^*$.

We factor out the Bose-Einstein condensate applying a unitary map $U_N : L^2_s (\bR^{3N}) \to \cFpn$, first introduced in \cite{LNSS}. To define $U_N$, we observe that any $\psi_N \in L^2_s (\bR^{3N})$ can be uniquely decomposed as
		\[ \psi_N = \alpha_0 \ph_0^{\otimes N} + \alpha_1 \otimes_s \ph_0^{\otimes (N-1)} + \dots + \alpha_N \]
with $\alpha_j \in L^2_{\perp \pn} (\bR^3)^{\otimes_s j}$ for every $j = 0, \dots , N$. Thus, we can set $U_N \psi_N = \{ \alpha_0, \alpha_1, \dots , \alpha_N \} \in \cFpn$. It is then possible to show that $U_N$ is unitary; see \cite{LNSS} for details.  

With $U_N$, we can define the excitation Hamiltonian $ \cL_N = U_N H_N U_N^*$, acting on a dense subspace of $\cFpn$. 
The action of $U_N$ on creation and annihilation operators is given by 
 \begin{equation}\label{2.1.UNconjugation}
  \begin{split}
    		U_{N} a^* (\pn ) a (\pn) U_{N}^*  & = N-\cN, \\
    		U_{N} a^*(f)a (\pn) U_{N}^* &= a^*(f) \sqrt{N-\cN} = \sqrt{N} \, b^* (f), \\
    		U_{N} a^* (\pn) a(g) U_{N}^* &= \sqrt{N-\cN} a(g) = \sqrt{N} \, b (g), \\
    		U_{N} a^* (f) a(g)U_{N}^* &= a^*(f)a(g)
    		\end{split}
    		\end{equation}
for all $ f,g\in L^2_{\perp \pn} (\bR^3)$, where $ \cN$ denotes the number of particles operator in $ \cFpn$. Writing $H_N$ in second quantized form and using (\ref{2.1.UNconjugation}), we proceed as in \cite[Section 4]{LNSS} and find that
		\[ \cL_N =  \cL^{(0)}_{N}+\cL^{(1)}_{N} + \cL^{(2)}_{N} + \cL^{(3)}_{N} + \cL^{(4)}_{N} , \]
where, in the sense of forms in $ \cFpn$, we have that
	\begin{equation}\label{eq:cLNj}
	\begin{split}
	\mathcal{L}_N^{(0)} &= \big\langle \pn, \big[ -\Delta + V_{ext} + \frac{1}{2} \big(N^3 V(N\cdot) * \vert \pn \vert^2\big) \big] \pn \big\rangle (N-\cN) \\
	& \qquad -  \frac12\big\langle \pn,  \big(N^3 V(N\cdot) * \vert \pn \vert^2 \big)\pn \big\rangle (\cN+1)(1-\cN/N) ,  \\
	\mathcal{L}_N^{(1)} &=   \sqrt{N} b\left(  \left( N^3 V(N\cdot) * \vert \pn\vert^2   -8\pi \frak{a}_0|\pn|^2 \right)\pn\right)\\
	&\hspace{0.5cm} - \frac{\mathcal{N}+1}{\sqrt N}  b\left( \left( N^3 V(N\cdot) * \vert \pn\vert^2 \right) \pn \right) + \text{h.c.},\\
	\mathcal{L}_N^{(2)} &= \int dx\; \left( \nabla_x a_x^* \nabla_x a_x + V_{ext}(x)a_x^{*} a_x \right)
		\\
		&\hspace{0.5cm} + \int  dxdy\;  N^3 V(N(x-y)) \vert\pn(y) \vert^2 \Big(b_x^* b_x - \frac{1}{N} a_x^* a_x \Big)   \\
		&\hspace{0.5cm}+ \int   dxdy\; N^3 V(N(x-y)) \pn(x) \pn(y) \Big( b_x^* b_y - \frac{1}{N} a_x^* a_y \Big)  \\
		&\hspace{0.5cm}+ \frac{1}{2}  \int  dxdy\;  N^3 V(N(x-y)) \pn(y) \pn(x) \Big(b_x^* b_y^*  + \text{h.c.} \Big), \\
	\mathcal{L}_N^{(3)} &=  \int   N^{ 5/2} V(N(x-y)) \pn(y) \big( b_x^* a^*_y a_x + \text{h.c.} \big) dx dy, \\
	\mathcal{L}_N^{(4)} &= \frac{1}{2} \int  dxdy\; N^2 V(N(x-y)) a_x^* a_y^* a_y a_x .
	\end{split}
	\end{equation}	

While $\cL_N$ maps $\cFpn$ to itself, the operators $ \cL_N^{(j)}$, $j\in \{0,1,2,3,4\}$ are also well defined as operators on $\cF^{\leq N}$. 
In the following, it will be technically convenient to identify the $ \cL_N^{(j)}$, through Eq. \eqref{eq:cLNj}, as acting in $\cF^{\leq N}$ and in this case we will denote them by $\wt\cL_N^{(j)}$. Moreover, we define 
		\be  \label{eq:defwtcLN} \wt \cL_N =  \wt\cL^{(0)}_{N}+\wt\cL^{(1)}_{N} +\wt  \cL^{(2)}_{N} + \wt \cL^{(3)}_{N} + \wt \cL^{(4)}_{N} \ee
as a self-adjoint operator acting on a dense subspace in $\cF^{\leq N}$. In particular, if $\Gamma(q)$ is the orthogonal projection from $\cF^{\leq N}$ onto $\cFpn$, defined by $(\Gamma (q) \Psi )^{(n)} = q^{\otimes n} \Psi^{(n)}$ for every $n \in \{ 0, 1, \dots , N \}$, with $q = 1- |\ph \rangle \langle \ph|$, then $ \cL_N = \Gamma(q) \wt \cL_N \Gamma(q)$ in $\cFpn$.

The vacuum expectation of $\cL_N$ is given, up to terms of order one, by 
\[ N \int \left[  \vert \nabla \pn(x) \vert^2 + V_{ext}(x) \vert \pn(x)\vert^2 +  \frac{1}{2} \widehat{V}(0)\vert \pn(x)\vert^4 \right] dx \]
Since $\hat{V} (0) > 8 \pi \frak{a}_0$, the difference between $\langle \Omega, \cL_N \Omega\rangle$ and the ground state energy of $\cL_N$ (or of $H_N$) is very large, of order $N$, in the limit $N \to \infty$. The point is that, through the map $U_N$, we expand $H_N$ around the energy of the pure condensate wave function $\pn^{\otimes N}\in L^2_s(\bR^{3N})$. In the Gross-Pitaevskii regime, however, it is well-known that short scale correlations among particles play a crucial role; they even affect the energy to leading order. To extract the missing correlation energy, we will conjugate $\cL_N$ with two unitary operator, given by the exponential of a quadratic and a cubic operator in (modified) creation and annihilation operators. We follow here the basic strategy of \cite{BBCS4}. Loosely speaking, through unitary conjugation, we renormalize the singular interaction, producing a soft, mean-field potential whose ultraviolet behavior is easy to control. This procedure extracts the missing correlation energy from the higher order terms in $\cL_N$ and, at the same time, it generates the coercivity needed to show Theorem \ref{thm:main}.

In our first renormalization step, we will conjugate $ \cL_N$ with a generalized Bogoliubov transformation. To define  the kernel of the quadratic phase, we consider the ground state 
solution of the Neumann problem 
\begin{equation}\label{eq:scatl} \left[ -\Delta + \frac{1}{2} V \right] f_{\ell} = \lambda_{\ell} f_\ell \end{equation}
on the  ball $|x| \leq N\ell$, for some $0 < \ell < 1$. For simplicity, we omit here the $N$-dependence in the notation for $f_\ell$ and for $\lambda_\ell$. By radial symmetry of the interaction $V$, $f_\ell$ is radially symmetric and we normalize it such that $f_\ell (x) = 1$ if $|x| = N \ell$. By scaling, $f_\ell (N.)$ solves 
\[ \left[ -\Delta + \frac{ N^2}{2} V (Nx) \right] f_\ell (Nx) = N^2 \lambda_\ell f_\ell (Nx) \]
on the ball where $|x| \leq \ell$. Later in our analysis, we will choose the parameter $\ell > 0$ to be sufficiently small, but it will always be of order one, independent of $N$. We then extend $f_\ell (N.)$ to $\bR^3$, by setting $f_{N,\ell} (x) = f_\ell (Nx)$ if $|x| \leq \ell$ and $f_{N,\ell} (x) = 1$ for $x \in \bR^3$ with $|x| > \ell$. Thus, $ f_{\ell}(N.)$ solves the equation  
		\begin{equation}\label{eq:scatlN}
		 \left( -\Delta + \frac{N^2}{2} V (N.) \right) f_{N,\ell} = N^2 \lambda_\ell f_{N,\ell} \chi_\ell  
		\end{equation}
where $\chi_\ell$ denotes the characteristic function of the ball $ B_\ell(0)$ of radius $\ell$, centered at the origin in $\bR^3$. Finally, we denote by $ w_\ell$ the function $ w_\ell = 1-f_\ell $. Notice that by scaling, $ w_\ell(N.)$ has compact support in $ B_\ell(0)$, for all $N\in\mathbb{N}$ sufficiently large. Defining the Fourier transform of $ w_\ell $ through
		\[ \widehat w_\ell (p) = \int_{}dx\; w_\ell(x)e^{-2\pi ipx}, \]
we see that $ w_\ell (N.)$ has Fourier transform
		\[ \int_{ }dx\; w_\ell(N x)e^{-2\pi ipx} = \frac1{N^3}\widehat w_\ell (p/N) \]
and we record that \eqref{eq:scatlN} implies that 
		\[
		\begin{split}  - 4\pi^2p^2 \widehat{w}_\ell (p/N) +  \frac{N^2}{2} ( \widehat{V} (./N) \ast \widehat{f}_{N,\ell}) (p) = N^5 \lambda_\ell ( \widehat{\chi}_\ell \ast \widehat{f}_{N,\ell}) (p). \end{split} \]
The next lemma collects important properties of $ f_\ell, w_\ell$ and the Neumann eigenvalue $ \lambda_\ell$.
\begin{lemma} \label{3.0.sceqlemma}
Let $V \in L^3 (\bR^3)$ be non-negative, compactly supported and spherically symmetric. Fix $\ell> 0$ and let $f_\ell$ denote the solution of \eqref{eq:scatl}. 
\begin{enumerate}
\item [i)] We have that 
\begin{equation}\label{eq:lambdaell} 
  \lambda_\ell = \frac{3\frak{a}_0 }{(\ell N)^3} \left(1 + \mathcal{O}\left(\frac{\frak{a}_0}{\ell N}\right) \right)
\end{equation}
\item[ii)] We have $0\leq f_\ell, w_\ell \leq 1$ and there exists a constant $C > 0$ such that  
\begin{equation} \label{eq:Vfa0} 
\left|  \int_{\bR^3}  V(x) f_\ell (x) dx - 8\pi \frak{a}_0  \right| 
\leq \frac{C \frak{a}_0}{\ell N} \, 
\end{equation}
for all $\ell \in (0;1)$, $N \in \bN$.
\item[iii)] There exists a constant $C>0$ such that 
			\begin{equation}\label{3.0.scbounds1} 
			w_\ell(x)\leq \frac{C}{|x|+1} \quad\text{ and }\quad |\nabla w_\ell(x)|\leq \frac{C }{|x|^2+1}. 
			\end{equation}
for all $x\in \bR^3$, $\ell \in (0;1)$ and $N \in \bN$ large enough. Moreover,
		\[
		\Big| \frac{1}{(N \ell)^2} \int_{\bR^3} \text{w}_{\ell}(x) dx - \frac 2 5 \pi \frak{a}_0   \Big| \leq   \frac{C \mathfrak{a}_0^2}{N \ell}
		\]
for all $\ell \in (0;1)$ and $N \in \bN$ large enough.
\item[iv)] There exists a constant $C > 0$ such that 
		\begin{equation}\label{eq:whwl} |\widehat{w}_\ell (p)| \leq \frac{C}{|p|^2} \end{equation}
for all $p \in \bR^3$, $\ell \in (0;1)$ and $N \in \bN$ large enough. 
\end{enumerate}        
\end{lemma}
\begin{proof}
This has already been proved in \cite[Appendix B]{BBCS3}, based on \cite[Lemma A.1]{ESY0} and \cite[Lemma 4.1]{BBCS4}. 
\end{proof}

Let us now define the correlation kernel for the generalized Bogoliubov transformation that we will be using. We denote by $G:\bR^3\to\mathbb{R}$ the rescaled function 
		\be \label{eq:defG} G(x) = -N w_\ell(Nx) \ee
which has compact support in $ B_\ell(0)$ and which, by \eqref{eq:whwl}, satisfies for all $p\in\bR^3$ 
		\be\label{eq:whGbnd} |\widehat G (p) | \leq \frac{C}{ |p|^2}. \ee
We denote by $ \chi_H$ the characteristic function of $ \{p\in \bR^3: |p|\geq \ell^{-\alpha} \}$, for fixed $\alpha>0$, and by $ \wch{\chi}_H$ its inverse Fourier transform. Finally, we define $\eta_H \in L^2(\bR^3\times \bR^3)$ by
		\be \label{eq:defetaH} \eta_H (x,y) = (G\ast \wch \chi_H ) (x-y) \pn(x)\pn(y).  \ee		
The next lemma summarizes basic properties of the kernel $\eta_H\in L^2(\bR^3\times \bR^3)$.  

\begin{lemma}\label{lm:bndsetaH}
Assume \eqref{eq:asmptsVVext}, let $\ell\in (0;1)$ and let $\alpha>0$. Set $\eta_{H,x}(y)= \eta_{H}(x;y)$ for $x\in\bR^3$. Then, there exists $C>0$, uniform in $N$, $\ell\in (0;1)$ and $x\in\bR^3$, such that 
	\be  \label{eq:bndsetaH}
	\Vert \eta_{H} \Vert \leq C \ell^{\alpha/2}, \hspace{0.5cm}  \Vert \eta_{H,x} \Vert \leq C \ell^{\alpha/2} \vert \pn(x) \vert, \hspace{0.5cm} \Vert \nabla_1 \eta_{H} \Vert, \,\Vert \nabla_2 \eta_{H} \Vert  \leq C \sqrt{N}.
	\end{equation}
	Furthermore, identifying $\eta_H(x;y)$ with the kernel of a Hilbert-Schmidt operator on $L^2(\bR^3)$ and $\eta_H^{(n)} (x;y)$ with the kernel of its $n$-th power, we have for $n\geq 2$ and $x,y\in\bR^3$ that 
	\begin{equation}\label{eq:bndsetaHn}
	\vert \eta_{H}^{(n)} (x;y) \vert 
	\leq \Vert \eta_{H,x} \Vert \cdot \Vert \eta_{H,y} \Vert \cdot \Vert \eta_{H} \Vert^{n-2} \leq  C \ell^{\alpha} \Vert \eta_{H} \Vert^{n-2} \vert \pn(x) \vert \cdot \vert \pn(y) \vert
	\end{equation}
and, for all $N$ sufficiently large, that
	\begin{equation}\label{eq:etapwbnd} |\eta_H(x;y)|\leq CN |\pn(x) | |\pn(y)| \leq CN. \end{equation}
 
\end{lemma}
\begin{proof}
By Theorem \ref{thm:gpmin1}, $\| \pn\|_\infty\leq C<\infty$ so that together with Eq. \eqref{eq:whGbnd}, we find that
	\begin{align*}
	\Vert \eta_{H,x} \Vert^2
	&= \int dy \ \vert (G\ast \widecheck{\chi}_H)(x-y)\vert^2 \cdot \vert \pn(x) \vert^2 \cdot \vert \pn(y) \vert^2 \\
	&\leq \Vert \pn \Vert_{\infty}^2 \cdot \vert \pn(x) \vert^2 \cdot \Vert G \ast \widecheck{\chi}_H \Vert^2
	\leq C \vert \pn(x) \vert^2 \int_{\vert p \vert \geq \ell^{-\alpha}} dp \ \vert \widehat{G}(p) \vert^2 \\
	&\leq C \vert \pn(x) \vert^2 \int_{\vert p \vert \geq \ell^{-\alpha}} dp \ \frac{1}{\vert p \vert^4}
	= C \ell^{\alpha} \vert \pn(x) \vert^2. 
	\end{align*}
This concludes the first two bounds in \eqref{eq:bndsetaH}, once we integrate the right hand side of the last estimate. To bound the gradient, we proceed similarly and find for $i=1,2$ that
	\begin{align*}
	\Vert \nabla_i  {\eta}_H \Vert^2
	&\leq  \int dx dy \ \vert \nabla(G \ast \widecheck{\chi}_H)(x-y)\vert^2 \cdot \vert \pn(x) \vert^2 \cdot \vert \pn(y) \vert^2   \\
	& \hspace{0.5cm}+ \int dx dy \ \vert (G \ast \widecheck{\chi}_H)(x-y)\vert^2 \cdot \vert \nabla \pn(x) \vert^2 \cdot \vert \pn(y) \vert^2\\
	&\leq C \Vert   (G\ast \widecheck{\chi}_H) \Vert_{H^1}^2
	\leq C \Vert  G \Vert_{H^1}^2
	\leq C \int dx \ \frac{N^4}{(\vert N x \vert^2 + 1)^2} = CN.
	\end{align*}
Notice that we used the pointwise estimate \eqref{3.0.scbounds1} in the second to last step. The estimates in \eqref{eq:bndsetaHn} follow directly from \eqref{eq:bndsetaH} and Cauchy-Schwarz. Finally, \eqref{eq:etapwbnd} follows from $\|\pn\|_\infty\leq C, \|G\|_\infty\leq N$ as well as $\| G\ast \widecheck{\chi_{H^c}}\|_\infty \leq \|G\|_1\|\widecheck{\chi_{H^c}}\|_\infty \leq C\ell^{1-3\alpha/2}\leq CN$ for $N$ large enough. Here, $\widecheck{\chi}_{H^c}$ denotes the inverse Fourier transform of the characteristic function on $ P_H^c$.
\end{proof}

With the kernel $\eta_H$, we consider the quadratic expression 
\begin{equation}\label{eq:B} B = \frac{1}{2} \int dx dy \, \eta_H (x;y) \, 
\left[  b_x^* b_y^* - b_x b_y \right] \end{equation} 
and the unitary operator $e^B : \cF^{\leq N} \to \cF^{\leq N}$ (since $\eta_H$ is real, the operator $B$ is antisymmetric). Because of its similarity with a Bogoliubov transformation (which would have $a_x, a_y$, instead of the modified fields $b_x, b_y$ in (\ref{eq:B})), we call $e^B$ a generalized Bogoliubov transformation. Notice that we do not project $\eta_H$ into the orthogonal complement of the condensate wave function $\ph$. As a consequence, $B$ and $e^B$ do not map $\cF_{\perp \ph}^{\leq N}$ into itself. This is not a problem for us, because we perform our analysis on the larger space $\cF^{\leq N}$ and only at the end (in Section \ref{sec:proof}) we switch back to the right space. An important properties of the unitary operator $e^B$ is that it preserves the number of particles, up to corrections of order one. The following lemma was proved in \cite[Lemma 3.1]{BS}; it is based on the observation that $\| \eta_H \| \leq C \ell^{\alpha/2} \leq C$. 
\begin{lemma} \label{lm:Npow} Let $B$ be the antisymmetric operator defined in (\ref{eq:B}). For every $n \in \bZ$ there exists a constant $C > 0$ such that  
		\[ e^{-B} (\cN +1)^n  e^B \leq C (\cN+1)^{n}  \]
as an operator inequality on $\cF^{\leq N}$.
\end{lemma}
Other important properties of the generalized Bogoliubov transformation $e^B$ will be discussed at the beginning of Section \ref{sec:GNell} (in particular, we will show there that, on states with few excitations, $e^B$ acts like a standard Bogoliubov transformation, up to small errors).

With $ \wt \cL_N $ from \eqref{eq:defwtcLN}, we can now define the quadratically renormalized excitation Hamiltonian 
\be  \label{eq:defGNell}\cG_{N} = e^{-B} \wt \cL_N e^{B}. \ee
The next proposition summarizes important properties of $\cG_{N}$. Before stating it, let us introduce the notation 
\[
\begin{split}
\cK & = \int dx\; a^*_x( -\Delta_x) a_x, \hspace{0.5cm} \cV_N = \frac12 \int dxdy\; N^2V(N(x-y))a^*_x a^*_y a_y a_x, \\
\cV_{ext} & = \int dx \; V_{ext} (x) a_x^* a_x, \hspace{.4cm}  \cH_N = \cK + \cV_{ext} + \cV_N
\end{split}
\]

\begin{prop}\label{prop:GNell}
Assume \eqref{eq:asmptsVVext} and let $\eta_H$ be defined as in \eqref{eq:defetaH}, with an $\alpha > 3$. Let 
\begin{equation}\label{eq:propGNell}
\begin{split}
\cG_N^\text{eff} &= N\cE_{GP}(\pn)  -\eps_{GP}\,\cN + 4\pi \an \| \pn\|_4^4\,\cN^2/N   + \cH_N  +2 \,\widehat{V}(0) \int dx  \vert \pn(x) \vert^2 b_x^* b_x \\
	&\quad+ 4\pi a_0 \int dx dy \, \widecheck{\chi}_{H^c}(x-y) \pn(x) \pn(y) ( b_x^* b_y^*+\emph{h.c.}) \\
	&\quad + \frac{1}{\sqrt{N}} \int dx dy \ N^3 V(N(x-y)) \pn(x) (b_x^* a_y^* a_x + \emph{h.c.})
\end{split} \end{equation}
where $\eps_{GP}$ has been introduced in (\ref{eq:euler}) and $\chi_{H^c}$ denotes the characteristic function of the set $\{ p\in\bR^3: |p|\leq \ell^{-\alpha}\}$. Then we have $\cG_N = \cG_N^\text{eff} + \cE_{\cG_N}$, where we can bound 	
\begin{equation*}
	\begin{split}
	\pm \cE_{\cG_{N}} \leq&\; C \ell^{(\alpha-3)/2} (\cN+\cK+\cV_N) + C\ell^{-5\alpha/2}N^{-1}(\cN+1)^2 + C \ell^{-4\alpha}  	\end{split}
	\end{equation*}
for a constant $C$, independent of $N \in \bN$ and $\ell \in (0;1)$.  
\end{prop}

{\it Remark:} In (\ref{eq:B}), we could have projected the kernel $\eta_H$ orthogonally to $\ph$. In this way, we could have defined $\cG_N$ as an operator on $\cF^{\leq N}_{\perp \ph}$ (conjugating $\cL_N$, rather than $\wt{\cL}_N$, with $e^B$). Also with this definition of $\cG_N$, we could have proven bounds similar to those in Prop. \ref{prop:GNell}, at the expenses of a longer proof (this procedure would also give an upper bound for the ground state energy, as in the remark after Theorem \ref{thm:main}).  

While the vacuum expectation of $\cG_N$ gives now, to leading order, the correct ground state energy of the Hamiltonian (\ref{eq:defHN}), the estimates in Prop. \ref{prop:GNell} are still not enough to show a bound for $\cN$ in low-energy states. The main issue is the cubic term in \eqref{eq:propGNell}. With Cauchy-Schwarz we can bound it with the positive quartic term $\cV_N$ (contained in $\cH_N$), but this produces a negative quadratic contribution that kills coercivity. To overcome this problem, we have to renormalize the excitation Hamiltonian $\cG_{N}$ again (actually, only its main part $\cG_N^\text{eff}$), but this time through the exponential of an expression cubic in creation and annihilation operators. 

If $ \chi_H$ denotes the characteristic function of the set $\{ p\in\bR^3: |p|\geq \ell^{-\alpha}\} $ and $G$ is defined as in \eqref{eq:defG}, we define the kernel 
	\begin{equation} \label{eq:defnuH}
	\nu_{H}(x;y) = (G* \widecheck{\chi}_H)(x-y) \pn(y).
	\end{equation}
For our analysis below, it will also be useful to introduce	
	\[
	\wt{k}(x;y) = - N w_\ell(N(x-y)) \pn(y)
	\]
so that in particular $\nu_{H}(x;y) = \wt k (x;y) - (G* \widecheck{\chi}_{H^c})(x-y) \pn(y) $. The next lemma summarizes important properties of $\nu_H$ and its relation to $\wt k$.
\begin{lemma}\label{lm:bndsnuH}
		Assume \eqref{eq:asmptsVVext} and let $\ell\in (0;1)$. Then $\nu_H$ satisfies 
		\begin{equation} \label{eq:nuHL2bnd}
		\Vert \nu_{H} \Vert \leq C \ell^{\alpha/2},\hspace{0.5cm} \Vert \nu_{H,x} \Vert \leq C \ell^{\alpha/2}, \hspace{0.5cm} 
		\Vert \nu_{H,y} \Vert \leq C \ell^{\alpha/2} \vert \pn(y) \vert
		\end{equation}
for all $ x,y\in \bR^3$ and for all $N$ sufficiently large. Moreover, denoting by $ \widehat{\nu}_{H}(p;q) = \wh{G}(p) \chi_H(p) \wh{\pn}(p+q)$ the Fourier transform of $ \nu_H$ as a function in $ L^2(\bR^3\times \bR^3)$, we have for all $ p,q\in \bR^3$ that	
		\begin{equation} \label{eq:whnuHL2bnd}
		\Vert \widehat{\nu}_{H,p} \Vert \leq \frac{C}{\vert p \vert^2} \chi_H(p)
		, \quad \Vert \widehat{\nu}_{H,q} \Vert \leq C\ell^{\alpha/2}.
		\end{equation}
\end{lemma}

\begin{proof}

Using \eqref{eq:whGbnd} we find that 
		\begin{align*}
		\Vert \nu_{H} \Vert \leq \Vert G * \widecheck{\chi}_H \Vert \, \Vert \pn \Vert  \leq C \Vert \wh{G} \chi_H \Vert \leq C \ell^{\alpha/2}
		\end{align*}
as well as
		\begin{align*}
		\Vert \nu_{H,x} \Vert \leq \Vert \pn \Vert_\infty  \Vert G * \widecheck{\chi}_H \Vert \leq C \ell^{\alpha/2}, \hspace{0.5cm} \Vert \nu_{H,y} \Vert = \vert \pn(y) \vert   \Vert G * \widecheck{\chi}_H \Vert \leq C \ell^{\alpha/2} \vert \pn(y) \vert.
		\end{align*}
Finally, \eqref{eq:whnuHL2bnd} follows from $ \widehat{\nu}_{H}(p;q) = \wh{G}(p) \chi_H(p) \wh{\pn}(p+q)$ and the estimate \eqref{eq:whGbnd}.
\end{proof}
	
A part from $\chi_H$, we will also need a second cutoff, localising on small momenta (we are constructing the cubic operator (\ref{definition A with l}), involving three particles; one particle should have small momentum, the other two large momenta, similarly to \cite[Eq. (5.1)]{BBCS4} in the translation invariant case). Here, it is convenient to use the Gaussian function  
\begin{equation} \label{definition chi Ll}
\gl(p) =  e^{- (\ell^{\beta}p)^2 }
\end{equation}
for an exponent $0 < \beta < \alpha$ ($g_L$ localises on momenta $|p| \lesssim \ell^{-\beta} \ll \ell^{-\alpha}$). Notice that the inverse Fourier transform of $\gl$ is given by
\[\cgl(x) = (\sqrt{\pi}  \ell^{-\beta} )^3 e^{ -(\pi \ell^{-\beta} x)^2 }.\]
In particular, it satisfies
\begin{equation} \label{Lp norms of widecheck chi L,l}
\Vert \cgl \Vert_1 =1, \quad 	\Vert \cgl \Vert = C \ell^{-\frac{3 \beta}{2}}.
\end{equation}	

With $ \nu_H$ defined as in \eqref{eq:defnuH} and $\gl$ from \eqref{definition chi Ll}, we introduce the operator 
\begin{equation} \label{definition A with l}
A = \frac{1}{\sqrt{N}} \int dx dy dz\, \nu_H(x;y) \cgl(x-z) \big( b_x^* a_y^* a_z -\text{ h.c.}\big).
\end{equation}
Since $A$ is anti-symmetric, $e^A:\cF^{\leq N} \to \cF^{\leq N}$ is a unitary map. An important observation is that conjugation with $e^A$ only increases the number of particles by a constant 
of order one, independent of $N$ (this result is similar to Lemma \ref{lm:Npow}, for the action of the generalized Bogoliubov transform $e^B$). 
\begin{lemma}\label{lm:Npowcubic}
Assume \eqref{eq:asmptsVVext}, let $ \ell\in (0;1)$, $ t\in [-1;1]$, $\alpha >0$ and let $k\in \mathbb{Z}$. Then, there exists a constant $C=C(k)>0$ such that in the sense of forms on on $\mathcal{F}^{\leq N}$, we have 
		\begin{equation} \label{rough bound e-A (N+1) eA with l}
		e^{-tA} (\mathcal{N}+1)^k e^{tA} \leq C (\mathcal{N}+1)^k.
		\end{equation}
\end{lemma}
	\begin{proof}
The proof is based on a Gronwall argument. Given $ \xi\in \cF^{\leq N}$, we define $f_\xi$ by $f_\xi(s) = \langle \xi, e^{-sA} (\mathcal{N}+1)^k e^{sA} \xi \rangle$. Taking its derivative yields
	\begin{align*}
	\partial_s f_\xi(s) = 2 \text{Re}\,  \langle \xi, e^{-sA}[(\mathcal{N}+1)^k, A] e^{sA} \xi \rangle 
	\end{align*}
and it is straight-forward to verify that
	\begin{align*}
	[(\mathcal{N}+1)^k, A] = \frac{1}{\sqrt{N}} \int dx dy dz \ \nu_{H}(x;y) \cgl(x-z) b_x^* a_y^* a_z \big((\mathcal{N}+2)^k -(\mathcal{N}+1)^k\big) + \text{h.c.}
	\end{align*}
By the mean value theorem, there exists a function $\Theta: \mathbb{N} \rightarrow (0;1)$ such that
	\begin{align*}
	(\mathcal{N}+2)^k - (\mathcal{N}+1)^k = k (\mathcal{N}+ \Theta(\mathcal{N}) + 1)^{k-1}.
	\end{align*}
Thus, together with Cauchy-Schwarz, \eqref{eq:nuHL2bnd} and \eqref{Lp norms of widecheck chi L,l}, we obtain that
	\begin{align*}
	&\vert \partial_s f_\xi (s) \vert\\
	&\leq \frac{2k}{\sqrt{N}} \int dx dy dz \ \vert \nu_{H}(x;y) \vert  \cgl(x-z) \Vert a_x a_y (\mathcal{N}+1)^\frac{k-1}{2} e^{s A}\xi \Vert \Vert a_z (\mathcal{N}+1)^\frac{k-1}{2} e^{sA}\xi \Vert \\ 
	&\leq \frac{C }{\sqrt{N}}  \bigg( \int dx dz \, \| \nu_{H,x}\|^2| \cgl(x-z) | \Vert a_z (\mathcal{N}+1)^\frac{k-1}{2} e^{sA}\xi \Vert ^2 \bigg)^{1/2}\\
	&\hspace{1.5cm}\times  \bigg( \int dx dy dz \, | \cgl(x-z) | \Vert a_x a_y (\mathcal{N}+1)^\frac{k-1}{2} e^{sA}\xi \Vert ^2 \bigg)^{1/2} \\
	&\leq \frac{C \ell^{\alpha/2}}{\sqrt{N}} \Vert (\mathcal{N}+1)^\frac{k}{2} e^{sA} \xi \Vert  \Vert (\mathcal{N}+1)^\frac{k+1}{2} e^{sA} \xi \Vert \leq C \Vert (\mathcal{N}+1)^\frac{k}{2} e^{sA} \xi \Vert^2= C f_\xi (s).
	\end{align*}
Since $C$ is indepedent of $\xi\in\cF^{\leq N}$, the claim follows from Gronwall's inequality.
	\end{proof}
Now, recalling the definition of $ \cG_{N}^{\text{eff}}$ in \eqref{eq:propGNell}, let us define the cubically renormalized excitation Hamiltonian $ \cJ_N $ through 
		\begin{equation}\label{eq:defJNell} \cJ_N = e^{-A}  \cG_{N}^{\text{eff}} e^A.
		\end{equation}
		
\begin{prop} \label{prop:JNell}
Assume \eqref{eq:asmptsVVext}, let $\ell\in (0;1)$ and fix $ 2\beta > \alpha > 7\beta /5$ as well as $\alpha>4$. Moreover, assume that $N\in\mathbb{N}$ is sufficiently large and that $\ell\in (0;1)$ is sufficiently small (but fixed, independently of $N$). Then, there exist $\kappa > 0$ and a constant $C>0$, independent of $N$ and $\ell$, such that 
		\begin{equation}\label{eq:propJNell}
		\begin{split} 
		\cJ_N \geq \; &N \cE_{GP}(\pn) + \frac12 d\Gamma (-\Delta_x + V_{ext}(x) + 8\pi \an |\pn(x)|^2 - \eps_{GP}) \\ &- C \ell^\kappa \cN - C \ell^{-4\alpha} (\cN+1)^2/N - C \ell^{-2\beta}.
		\end{split} 
		\end{equation}
\end{prop}

\section{Proof of Theorem \ref{thm:main}}
\label{sec:proof}

It is enough to prove that there exist constants $c, C>0$, independent of $N\in\mathbb{N}$, such that for all sufficiently large $N$, we have in the sense of forms in $ \cF_{\bot \pn}^{\leq N}$ that 
		\begin{equation}\label{eq:coercbnd}
		\cL_N \geq N\cE_{GP}(\pn) + c\, \cN - C.
		\end{equation} 
To prove \eqref{eq:coercbnd}, we first localize the operator $\wt \cL_N$ defined in \eqref{eq:defwtcLN}, based on an argument from \cite{LNSS} (see, in particular, \cite[Prop. 6.1]{LNSS}). To this end, let $\delta\in (0;1)$ (it will be determined below) and let $0\leq f,g\leq 1$ be two smooth, real-valued functions such that $ f^2+g^2=1$, $ f(x)=1$ for $|x|\leq 1/2$ as well as $f(x)=0 $ for $|x|\geq1$. With the notation $f_{\delta N} = f (\cN / \delta N)$, $g_{\delta N} = g (\cN / \delta N)$, we observe that 
		\[ \wt\cL_N = f_{\delta N}\wt\cL_N f_{\delta N} +g_{\delta N}\wt\cL_N g_{\delta N} +\frac12 \big( [f_{\delta N}, [f_{\delta N}, \wt \cL_N]] + [g_{\delta N}, [g_{\delta N}, \wt \cL_N]] \big),  \]
Thus, in the sense of forms on $\cF^{\leq N}$,  we have 
		\begin{equation}\label{eq:locLN}
		\begin{split}
		\cL_N =&\, \Gamma(q)   f_{\delta N}\wt\cL_N f_{\delta N} \Gamma(q) + \Gamma(q) g_{\delta N}\wt\cL_N g_{\delta N} \Gamma(q) \\
		& +\frac12\Gamma(q)  \big( [f_{\delta N}, [f_{\delta N}, \wt \cL_N]] + [g_{\delta N}, [g_{\delta N}, \wt \cL_N]] \big)  \Gamma(q).
		\end{split}
		\end{equation}
Recall here that $q = 1- |\ph \rangle \langle \ph|$ and  that $\Gamma(q)$ is the orthogonal projection from $\cF^{\leq N}$ onto $\cF_{\bot}^{\leq N}$. The second line in \eqref{eq:locLN} contains error terms and they can be controlled as follows. Observing that
		\[\begin{split}
		&[f_{\delta N}, [f_{\delta N}, \wt \cL_N]] \\
		&=   \sqrt{N} b^*\left(  \left( N^3 V(N\cdot) * \vert \pn\vert^2   -8\pi \frak{a}_0|\pn|^2 \right)\pn\right)\big (f_{\delta N} (\cN+1)- f_{\delta N}(\cN) \big)^2\\
	&\hspace{0.5cm} -   b^*\left( \left( N^3 V(N\cdot) * \vert \pn\vert^2 \right) \pn \right)\frac{\mathcal{N}+1}{\sqrt N}\big (f_{\delta N} (\cN+1)- f_{\delta N}(\cN) \big)^2 \\
		&\hspace{0.5cm}+ \frac{1}{2}  \int  dxdy\;  N^3 V(N(x-y)) \pn(y) \pn(x) b_x^* b_y^* \big (f_{\delta N} (\cN+2)- f_{\delta N}(\cN) \big)^2  \\
	  &\hspace{0.5cm}+\int  dxdy\, N^{ 5/2} V(N(x-y)) \pn(y)   b_x^* a^*_y a_x\big (f_{\delta N} (\cN+1)- f_{\delta N}(\cN) \big)^2  + \text{h.c.} ,  
		\end{split}\]
and similarly for $ [g_{\delta N}, [g_{\delta N}, \wt \cL_N]]$, a straight-forward application of Cauchy-Schwarz, $ \|\pn\|_\infty\leq C ,  \| N^3V(N.)\ast |\pn|^2\|_\infty \leq C$ and the mean value theorem implies that 
		\[ [f_{\delta N}, [f_{\delta N}, \wt \cL_N]] + [g_{\delta N}, [g_{\delta N}, \wt \cL_N]]\geq - C \big(\| f'\|^2_\infty + \|g'\|^2_\infty\big) \frac{(  \cV_N + N)}{\delta^2N^2}. \]
Hence, we find the lower bound 
		\begin{equation}\label{eq:locLN2}
		\begin{split}
		\cL_N \geq &\, \Gamma(q)   f_{\delta N}\wt\cL_N f_{\delta N} \Gamma(q) + \Gamma(q) g_{\delta N}\wt\cL_N g_{\delta N} \Gamma(q) - C\, \Gamma(q)  \frac{  \cV_N}{\delta^2 N^2}  \Gamma(q) - C.
		\end{split}
		\end{equation}

In the next step, we control the first contribution on the r.h.s. in \eqref{eq:locLN2} through Prop. \ref{prop:GNell} and Prop. \ref{prop:JNell}. We define $\cG_{N}$ as in \eqref{eq:defGNell}, $\cJ_N$ as in \eqref{eq:defJNell} and we choose the parameters $ \alpha, \beta>0$ such that $ 2 \beta > \alpha > 7\beta /5$ and $\alpha >5$. Moreover, we assume in the following that $N$ is sufficiently large and that $\ell\in(0;1)$ is sufficiently small. These assumptions ensure that the conditions in Prop. \ref{prop:GNell} and Prop. \ref{prop:JNell} are satisfied. Then, by Prop. \ref{prop:GNell} and Lemma \ref{lm:Npow}, we find that
		\[\begin{split}
		f_{\delta N}\wt\cL_N f_{\delta N} & = f_{\delta N}e^B \cG_{N}^{\text{eff}}  e^{-B}  f_{\delta N} + f_{\delta N}e^B \cE_{\cG_{N}} e^{-B}  f_{\delta N}\\
		&\geq N\cE_{GP}(\pn)f_{\delta N}^2 +f_{\delta N}e^B \big(\cG_{N}^{\text{eff}}  -N\cE_{GP}(\pn)\big)  e^{-B}  f_{\delta N}\\
		&\hspace{0.4cm} -C\ell^{\kappa_1} f_{\delta N}e^B \big( \cN+\cK +\cV_N\big) e^{-B}f_{\delta N} - C\ell^{-\kappa_2} N^{-1} (\cN+1)^{2} f_{\delta N}^2 - C\ell^{-\kappa_2}\\
		&\geq N\cE_{GP}(\pn)f_{\delta N}^2 +f_{\delta N}e^B \big(\cG_{N}^{\text{eff}}  -N\cE_{GP}(\pn)\big)  e^{-B}  f_{\delta N}\\
		&\hspace{0.4cm} -C\ell^{\kappa_1} f_{\delta N}e^B \big( \cN+\cK +\cV_N\big) e^{-B}f_{\delta N} - C\ell^{-\kappa_2} \delta \cN f_{\delta N}^2 - C\ell^{-\kappa_2}
		\end{split} \] 
for suitable $\kappa_1, \kappa_2>0$. Before we can continue and apply Prop. \ref{prop:JNell}, we need to bound the error proportional to $ (\cN+\cK+\cV_N)$ in terms of $\cG_{N}^{\text{eff}}$, as defined in (\ref{eq:propGNell}). To this end we observe that, since $\chi_{H^c}(p)\geq 0$,
\[\begin{split}
\int dxdy\, \widecheck\chi_{H^c}(x-y) \pn(x)\pn(y) (b^*_x + b_x)(b^*_y + b_y) \geq 0,
\end{split}\]	
Denoting by $A$ the operator with kernel $A (x;y) = \check{\chi}_{H^c} (x-y) \ph (x) \ph (y)$, and observing that, with \eqref{eq:expdecaypn}, $\| A \| \leq \sup_x \int |A(x;y)| dy \leq C$, uniformly in $\ell \in (0;1)$, we conclude that 
 \[\begin{split}
\int dxdy\, \widecheck\chi_{H^c}(x-y) \pn(x)\pn(y) \big(b^*_x  b^*_y +\text{h.c.}\big) 
&\geq  -2 d\Gamma (A) - \widecheck\chi_{H^c}(0) \geq - C \cN - C \ell^{-3\alpha} \end{split}\]
With this bound (and with Cauchy-Schwarz), (\ref{eq:propGNell}) then easily implies that 
\[ \cG_{N}^{\text{eff}} -N\cE_{GP}(\pn) \geq \frac12 \cH_N - C  \cN  - C\ell^{-\alpha}.   \]
Hence, we obtain that 
		\begin{equation}\label{eq:fLNbnd1}
		\begin{split}
		f_{\delta N}\wt\cL_N f_{\delta N} &\geq N\cE_{GP}(\pn)f_{\delta N}^2 +(1- C\ell^{\kappa_1})f_{\delta N}e^B \big(\cG_{N}^{\text{eff}}  -N\cE_{GP}(\pn)\big)  e^{-B}  f_{\delta N}\\
		&\hspace{0.4cm} - C\big(\ell^{\kappa_1} + \ell^{-\kappa_2}\delta\big) \cN f_{\delta N}^2 - C\ell^{-\kappa_2}.
		\end{split} 
		\end{equation}
Now, we apply Prop. \ref{prop:JNell}, Lemma \ref{lm:Npow} and Lemma \ref{lm:Npowcubic} which yields similarly as above
		\[\begin{split}
		&f_{\delta N}e^B \big(\cG_{N}^{\text{eff}}  -N\cE_{GP}(\pn)\big)  e^{-B}  f_{\delta N}\\
		&=  f_{\delta N}e^Be^A \big(\cJ_N   -N\cE_{GP}(\pn)\big)  e^{-A}e^{-B}  f_{\delta N}\\ 
		&\geq \frac12 f_{\delta N}e^Be^A  d\Gamma ( -\Delta_x + V_{ext}(x) + 8\pi \an |\pn(x)|^2 - \eps_{GP}) e^{-A}e^{-B}  f_{\delta N} \\
		&\hspace{0.4cm} - C\big(\ell^{\kappa_1} + \ell^{-\kappa_2}\delta \big) \cN f_{\delta N}^2 - C\ell^{-\kappa_2}.
		\end{split}\]		
From (\ref{eq:euler}) and under the assumptions \eqref{eq:asmptsVVext}, the operator $h_{GP} = -\Delta + V_{ext}(x) + 8\pi \an |\pn(x)|^2 - \eps_{GP}$ has the simple eigenvalue 0 at the bottom of its spectrum (with eigenvector $\ph$) and a positive gap $\lambda_1 > 0$ on top of it; see, for example,  \cite[Theorem XIII.47]{RS4}. Applying once more Lemmas \ref{lm:Npow} and \ref{lm:Npowcubic}, we therefore find that 
\begin{equation}\label{eq:G-E}\begin{split}
&f_{\delta N}e^B \big(\cG_{N}^{\text{eff}}  -N\cE_{GP}(\pn)\big)  e^{-B}  f_{\delta N}\\
&\geq \frac{\lambda_1}2 f_{\delta N}e^Be^A   ( \cN-a^*(\pn)a(\pn) ) e^{-A}e^{-B}  f_{\delta N}   - C\big(\ell^{\kappa_1} + \ell^{-\kappa_2}\delta  \big) \cN f_{\delta N}^2 - C\ell^{-\kappa_2} \end{split}\end{equation} 
With the commutators 
\[ \begin{split} 
\Big[ a^* (\ph ) a(\ph) , A \Big] =\; & \frac{1}{\sqrt{N}} \int dx dy dz\, \nu_H(x;y) \cgl(x-z) \pn(x)  b^*(\pn) a_y^* a_z  \\ & +\frac{1}{\sqrt{N}} \int dx dy dz\, \nu_H(x;y) \cgl(x-z) \pn(y)  b_x^* a^*(\pn) a_z 
\\ & -\frac{1}{\sqrt{N}} \int dx dy dz\, \nu_H(x;y) \cgl(x-z) \pn(z)  b_x^* a_y^* a(\pn)   +\text{h.c.} \\
 \Big[ a^* (\ph) a(\ph) , B \Big] =  \; &\int dxdy\, \eta_H(x;y)  \pn(y) b^*_x b^*(\pn)  +\text{h.c.} 
 \end{split} \]
and with Lemma \ref{lm:bndsnuH} and Lemma \ref{lm:Npowcubic}, we find  
\[ -e^B e^A a^*(\pn)a(\pn)e^{-A} e^{-B} \geq -a^*(\pn)a(\pn) - C\ell^{\alpha/2} (\cN+1) .\]
Thus, inserting in (\ref{eq:G-E}) and then in (\ref{eq:fLNbnd1}), and using again 
Lemmas \ref{lm:Npow} and \ref{lm:Npowcubic}, we arrive at 
\[\begin{split}
		f_{\delta N}\wt\cL_N f_{\delta N} 
		&\geq N\cE_{GP}(\pn)f_{\delta N}^2 +\wt c\,(1- C\ell^{\kappa_1}) \, \cN  f_{\delta N}^2  - \frac{\lambda_1}2(1- C\ell^{\kappa_1})    a^*(\pn)a(\pn)     f_{\delta N}^2\\
		&\hspace{0.4cm} - C\big(\ell^{\kappa_1} + \ell^{-\kappa_2}\delta \big) \cN f_{\delta N}^2 - C\ell^{-\kappa_2}, 
		\end{split} \]
where the constants $\wt c, C>0$ are independent of $\ell\in (0;1)$, $\delta\in(0;1)$ and $N\in \mathbb{N}$. 

We now set $ \delta = \ell^{2\kappa_2}$ and choose $\ell$ sufficiently small, so that the last bound implies
		\begin{equation}\label{eq:fLNbnd2}
		\begin{split}
		f_{\delta N}\wt\cL_N f_{\delta N} 
		&\geq N\cE_{GP}(\pn)f_{\delta N}^2 +c_1  \, \cN  f_{\delta N}^2  - \frac{\lambda_1}4     a^*(\pn)a(\pn)     f_{\delta N}^2   - C. 
		\end{split} 
		\end{equation}
Here, the positive constant $c_1>0$ is independent of $\ell\in (0;1)$, $\delta\in(0;1)$ and $N\in \mathbb{N}$.
For the rest of the proof, we fix this choice of $\delta = \ell^{4\kappa_2}$ and this (sufficiently small) value of $\ell\in (0;1)$ so that \eqref{eq:fLNbnd2} holds true. Then, substituting \eqref{eq:fLNbnd2} into \eqref{eq:locLN2}, we get
		\[\cL_N \geq \,  N\cE_{GP}(\pn)f_{\delta N}^2 \Gamma(q) + c_1   \, \cN  f_{\delta N}^2\Gamma(q) +  g_{\delta N} \Gamma(q)\wt\cL_N\Gamma(q) g_{\delta N}  - C\, \Gamma(q)  \frac{  \cV_N}{\delta^2 N^2}  \Gamma(q) - C\]
in the sense of forms in $\cF_{\perp \ph}^{\leq N}$. Notice that we used that $ \Gamma(q)a^*(\pn)a(\pn)\Gamma(q) = 0 $ and $ [\Gamma(q), \cN] = 0$. Since $ \Gamma(q)\wt \cL_N \Gamma(q) = \cL_N$ in the sense of forms in $\cF_{\perp \ph}^{\leq N}$ and since $ \Gamma(q)_{|\cF^{\leq N}_{\bot \pn}} = \operatorname{id}_{|\cF^{\leq N}_{\bot \pn}}$, the previous bound translates to 
		\begin{equation}\label{eq:fLNbnd3}
		 \cL_N \geq \,  N\cE_{GP}(\pn)f_{\delta N}^2   + c_1   \, \cN  f_{\delta N}^2  +  g_{\delta N}  \cL_N g_{\delta N}  - C\,    \frac{  \cV_N}{\delta^2 N^2}    - C\end{equation}
in the sense of forms in $\cF_{\bot \pn}^{\leq N}$

Next, we notice that the contribution proportional to $ g_{\delta N }   \cL_N g_{\delta N}$ in \eqref{eq:fLNbnd3} can be controlled as explained in \cite{BBCS4}, using the results from \cite{LS2, NRS}, i.e. \eqref{eq:BEC0}. This argument shows that there exists some constant $c_2>0$ such that for all sufficiently large $N\in\mathbb{N}$
		\begin{equation} \label{eq:aprioriinfo}\begin{split} 
		g_{\delta N }  \cL_N  g_{\delta N} &\geq  N\cE_{GP}(\pn)g_{\delta N }^2 +   \left[\inf_{\substack{ \xi\in \cF_{\bot \pn}^{\leq N}\cap\cF_{\bot \pn}^{\geq \delta N/2} :\|\xi\|=1 }} \Big( \frac1N\langle\xi, \cL_N \xi \rangle -  \cE_{GP}(\pn)\Big) \right] N g_{\delta N}^2 \\ 
		& \geq  N\cE_{GP}(\pn)g_{\delta N }^2 +   c_2 N g_{\delta N}^2
		 \end{split}\end{equation} 
where $ \cF_{\bot \pn}^{\geq \delta N/2} = \{ \xi \in \cF_{\bot \pn}: \chi(\cN \geq \delta N/2) \xi = \xi\}$. Hence, plugging \eqref{eq:aprioriinfo} into \eqref{eq:fLNbnd3}, we conclude that
		 \[\begin{split}
		 \cL_N \geq & \,  N\cE_{GP}(\pn)    + c_1   \, \cN  f_{\delta N}^2  +   c_2 N g_{\delta N}^2  - C\,    \frac{  \cV_N}{\delta^2 N^2}    - C
		 \geq \,  N\cE_{GP}(\pn)    + \wt c   \, \cN     - C\,    \frac{  \cV_N}{\delta^2 N^2}    - C
		 \end{split} \]
in the sense of forms in $\cF_{\bot \pn}^{\leq N}$, where $\wt c = \min(c_1,c_2)>0$. 

Finally, to control the error that contains the potential energy $ \cV_N$, we recall the definitions \eqref{eq:cLNj} and a simple computation involving Cauchy-Schwarz shows that 
		\[\cL_N \geq \frac12 \cH_N - C N \geq \frac12 \cV_N - C N \]
for some constant $C>0$, independent of $N\in\mathbb{N}$. Therefore, we find  
		\[ \begin{split}
		 \cL_N   \geq & \,  N\cE_{GP}(\pn)    + \wt c   \, \cN     - C\,    \frac{  \cL_N}{\delta^2 N^2}    - C,
		 \end{split} \]
so that, by choosing $N\in\mathbb{N}$ sufficiently large, we have proved that
		 \[\begin{split}
		 \cL_N \geq & \,  N\cE_{GP}(\pn)    + c   \, \cN      - C
		 \end{split} \]
for some constant $ c>0$ that is independent of $N$. \qquad \quad \qed

\section{Analysis of $\cG_{N}$}  \label{sec:GNell}

In this section we analyse the operator $\cG_{N} = e^{-B} \wt{\cL}_N e^B$, as defined in \eqref{eq:defGNell}. To compute the action of the generalized Bogoliubov transform $e^B$ on $\wt{\cL}_N$, we are going to compare it with the action of a standard Bogoliubov transformation. Interpreting  \eqref{eq:defetaH} as the integral kernel of a Hilbert-Schmidt operator on $L^2 (\bR^3)$, we define   
\[ \sinh_{\eta_H} = \sum_{j=0}^\infty \frac{\eta_H^{(2k+1)}}{(2k+1)!},\hspace{0.5cm} \cosh_{\eta_H} = \sum_{j=0}^\infty \frac{\eta_H^{(2k)}}{(2k)!}.  \]
In addition, we define the Hilbert-Schmidt operators $\text{p}_{\eta_H} $ and $\text{r}_{\eta_H}$ by
		\begin{equation}\label{eq:defpretaH} \text{p}_{\eta_H} = \sinh_{\eta_H} - \,\eta_H = \sum_{j=1}^\infty \frac{\eta_H^{(2k+1)}}{(2k+1)!}, \hspace{0.5cm} \text{r}_{\eta_H} = \cosh_{\eta_H} - \,\text{id}_{} =  \sum_{j=1}^\infty \frac{\eta_H^{(2k)}}{(2k)!}.
		\end{equation}
Using \eqref{eq:bndsetaH} one obtains for $\ell\in (0;1)$ small enough
\begin{equation} \label{eq:bndpr}
  \Vert \text{p}_{\eta_H} \Vert, \Vert \text{r}_{\eta_H} \Vert \leq C\ell^{\alpha}, \quad \vert \text{p}_{\eta_H}(x,y) \vert, \vert \text{r}_{\eta_H}(x,y) \vert \leq C\ell^\alpha \ph(x) \ph(y).
\end{equation}

The following lemma, whose proof is an adaptation of the translation-invariant case \cite[Lemma 3.4]{BBCS4}, shows that, on states with few excitations, $e^B$ acts approximately like a standard Bogoliubov transformation.  
\begin{lemma}\label{lm:d-bds}
	Let $n\in  \mathbb{Z}$, and let $f\in L^2(\bR^3)$. Let $d_{\eta_H}(f)$ as well as $d_{\eta_H,x}$ be defined as
	\begin{equation} \label{eq:defd}
	e^{-B} b(f) e^{B} = b( \cosh_{\eta_H}(f)) + b^*( \sinh_{\eta_H}(\overline{f})) + d_{\eta_H}(f),
	\end{equation}
	respectively
	\begin{equation} \label{eq:defdx}
	e^{-B} b_x e^{B} = b( \cosh_{\eta_H, x} ) + b^*( \sinh_{\eta_H, x}) + d_{\eta_H,x}.
	\end{equation}
	Then, there exists a constant $C>0$ such that
	\begin{equation} \label{eq:bnddf}
	\begin{split}
	\Vert (\cN+1)^{n/2}d_{\eta_H}(f) \xi \Vert & \leq \frac{C\ell^{\alpha/2} }{N}  \|f\|  \Vert (\mathcal{N}+1)^{(n+3)/2} \xi \Vert, \\
	\Vert (\cN+1)^{n/2}d_{\eta_H}(f)^* \xi \Vert & \leq \frac{C\ell^{\alpha/2}}{N}  \|f\|  \Vert (\mathcal{N}+1)^{(n+3)/2} \xi \Vert
	\end{split}
	\end{equation}	
	and such that, for all $x\in\bR^3$, we have that
	\begin{equation} \label{eq:bnddx}
	\begin{split}
	\Vert (\cN+1)^{n/2}d_{\eta_H,x} \xi \Vert & \leq \frac{C}{N} \Big[ \ell^{\alpha/2}  \Vert a_x (\mathcal{N}+1)^{(n+2)/2} \xi \Vert +  \|\eta_{H,x}\| \Vert (\mathcal{N}+1)^{(n+3)/2} \xi \Vert\Big].
	\end{split}
	\end{equation}
	Furthermore, if we set $\overline{d}_{\eta_H,x} = d_{\eta_{H,x}} +  (\cN/N) b^*( \eta_{H,x}) $, it holds true that
	\begin{equation} \label{eq:bndaydxbar}
	\begin{split}
	&\Vert (\cN+1)^{n/2}a_y \overline{d}_{\eta_H,x} \xi \Vert\\
	&\leq \frac{C }{N}  \Big[  \|\eta_{H,x}\|\|\eta_{H,y}\|  \Vert (\mathcal{N}+1)^{(n+2)/2} \xi \Vert  + \ell^{\alpha/2} |\eta_H(x;y)|   \Vert (\mathcal{N}+1)^{(n+2)/2} \xi \Vert \\
	&\hspace{1.5cm}+     \|\eta_{H,y}\|    \Vert a_x (\mathcal{N}+1)^{(n+1)/2}\xi \Vert  +  \ell^{\alpha/2}\|\eta_{H,x}\|  \Vert a_y (\mathcal{N}+1)^{(n+3)/2} \xi \Vert \\
	&\hspace{1.5cm}+  \ell^{\alpha/2} \Vert a_x a_y (\mathcal{N}+1)^{(n+2)/2} \xi \Vert \Big]
	\end{split}
	\end{equation}
	and, finally, we have that
	\begin{equation} \label{eq:bnddxdy}
	\begin{split}
	&\Vert (\cN+1)^{n/2}d_{\eta_H,x} d_{\eta_H,y} \xi \Vert \\
	&\leq \frac{C}{N^2} \Big[    \|\eta_{H,x}\| \|\eta_{H,y}\|  \Vert (\mathcal{N}+1)^{(n+6)/2} \xi \Vert + \ell^{\alpha/2}  |\eta_H(x;y)|   \Vert (\mathcal{N}+1)^{(n+4)/2}\xi \Vert  \\
	&\hspace{1.2cm} + \ell^{\alpha/2}\|\eta_{H,y}\|  \vert   \Vert a_x (\mathcal{N}+1)^{(n+5)/2)} \xi \Vert+   \ell^{\alpha/2} \| \eta_{H,x}\| \Vert a_y(\mathcal{N}+1)^{ (n+5)/2} \xi \Vert  \\
	&\hspace{1.2cm} +  \ell^{\alpha}\Vert a_x a_y (\mathcal{N}+1)^{(n+4)/2} \xi \Vert \Big].
	\end{split}
	\end{equation}
\end{lemma}

From the decomposition \eqref{eq:defwtcLN}, we can write 
\[\cG_{N} =  \cG_{N}^{(0)} +\cG_{N}^{(1)}+\cG_{N}^{(2)}+\cG_{N}^{(3)}+\cG_{N}^{(4)}, \]
where, for $j\in \{0,1,2,3,4\}$, we set
\[
\cG_{N}^{(j)} = e^{-B}\wt\cL_{N}^{(j)}e^{B}.
\]
In the following subsections, we will analyze the main contributions $\cG_{N}^{(j)} $ separately and, in Section \ref{sec:proofpropGNell}, we combine these results to conclude Proposition \ref{prop:GNell}.

\subsection{Analysis of $\cG_{N}^{(0)}$} 

From \eqref{eq:cLNj}, we recall that 
		\begin{equation}\label{eq:wtLN0}
		\begin{split}
		 \widetilde{\mathcal{L}}_N^{(0)} &= \big\langle \pn, \big[ -\Delta + V_{ext} + \frac{1}{2} \big(N^3 V(N\cdot) * \vert \pn \vert^2\big) \big] \pn \big\rangle (N-\cN) \\
	& \qquad -  \frac12\big\langle \pn,  \big(N^3 V(N\cdot) * \vert \pn \vert^2 \big)\pn \big\rangle (\cN+1)(1-\cN/N).
		\end{split}
		\end{equation}
\begin{prop}\label{prop:GNell0} There exists a constant $C>0$ such that
		\[\begin{split}
		 \cG_{N}^{(0)} = &\, \big\langle \pn, \big[ -\Delta + V_{ext} + \frac{1}{2} \big(N^3 V(N\cdot) * \vert \pn \vert^2\big) \big] \pn \big\rangle (N-\cN) \\
	& \qquad -  \frac12\big\langle \pn,  \big(N^3 V(N\cdot) * \vert \pn \vert^2 \big)\pn \big\rangle (\cN+1)(1-\cN/N) + \cE_{N,\ell}^{(0)},
		 \end{split}\]
where the self-adjoint operator $\cE_{N,\ell}^{(0)}$ satisfies 
		\[ \pm \cE_{N,\ell}^{(0)} \leq C \ell^{\alpha/2} (\cN+1) \]
for all $\alpha>0$ and $\ell\in(0;1)$.
\end{prop}
\begin{proof} 
We start with the observation that 
		\[\begin{split}
		 e^{-B} \cN e^{B} -\cN = \int_0^1 ds \,\bigg( \int dxdy\, \eta_H(x;y) e^{-s B}b^*_xb^*_y e^{sB}\bigg) +\text{h.c.}
		\end{split}\]
This implies together with Lemma \ref{lm:Npow} and Cauchy-Schwarz that
		\[ \pm \big(e^{-B} \cN e^{B} -\cN\big) \leq C\ell^{\alpha/2}(\cN+1). \]
Similarly, it is straight-forward to prove that 
 \[ \pm \big(e^{-B} \cN^2 e^{B} -\cN^2\big) \leq C\ell^{\alpha/2}(\cN+1)^2. \]
If we use these two observations together with the bounds
		\[ | \langle \pn, ( -\Delta + V_{ext})\pn\rangle|\leq \cE_{GP}(\pn)\leq C,\; \langle\pn, (N^3 V(N\cdot) * \vert \pn \vert^2) \pn \big\rangle|\leq \|V\|_1\|\pn\|_\infty^2\leq  C, \]
the proposition follows directly from the definition of $\wt{\cL}_N^{(0)}$ in Eq. \eqref{eq:wtLN0}.
\end{proof}

\subsection{Analysis of $\cG_{N}^{(1)}$}

From \eqref{eq:cLNj}, we recall that 
		\begin{equation}\label{eq:wtLN1}
		\begin{split}
		\wt{\mathcal{L}}_N^{(1)} &=   \sqrt{N} b\left(  \left( N^3 V(N.) * \vert \pn\vert^2   -8\pi \frak{a}_0|\pn|^2 \right)\pn\right)\\
	&\hspace{1cm} - \frac{\mathcal{N}+1}{\sqrt N}  b\left( \left( N^3 V(N\cdot) * \vert \pn\vert^2 \right) \pn \right) + \text{h.c.}
		\end{split}
		\end{equation}
For the statement of the next proposition, let us define $h_N\in L^2(\bR^3)\cap L^\infty(\bR^3)$ by
		\begin{equation}\label{eq:defhN} h_N = \big(N^3(Vw_\ell) (N.)\ast |\pn|^2\big)\pn. \end{equation}
\begin{prop}\label{prop:GNell1} There exists a constant $C>0$ such that
		\[\begin{split}
		 \cG_{N}^{(1)} = &\, \big[ \sqrt{N} b(\cosh_{\eta_H}(h_N) ) + \sqrt{N} b^*(\sinh_{\eta_H}(h_N) )+\emph{h.c.}\big] + \cE_{N,\ell}^{(1)},
		 \end{split}\]
where the self-adjoint operator $\cE_{N,\ell}^{(1)}$ satisfies 
		\[ \pm \cE_{N,\ell}^{(1)} \leq C \ell^{\alpha/2} (\cN+1) + C\ell^{-\alpha/2} N^{-1}(\cN+1)^2+ C\ell^{-1}.   \]
for all $\alpha>0$ and $\ell\in(0;1)$.
\end{prop}
\begin{proof}
First of all, we notice that a simple application of Cauchy-Schwarz and the fact that $ \| (N^3 V(N\cdot) * \vert \pn\vert^2 ) \pn\|_2\leq C$ imply that
		\[ \pm\bigg( \frac{\mathcal{N}+1}{\sqrt N}  b\left( \left( N^3 V(N.) * \vert \pn\vert^2 \right) \pn  \right) +\text{h.c.}\bigg) \leq C \ell^{\alpha/2} (\cN+1) + C\ell^{-\alpha/2} N^{-1}(\cN+1)^2. \]
By Lemma \ref{lm:Npow}, we therefore obtain that
		\[\begin{split} &\pm e^{-B} \bigg( \frac{\mathcal{N}+1}{\sqrt N}  b\left( \left( N^3 V(N.) * \vert \pn\vert^2 \right) \pn \right)+\text{h.c.}\bigg)e^{B} \\
		&\hspace{4.5cm}\leq C \ell^{\alpha/2} (\cN+1) + C\ell^{-\alpha/2} N^{-1}(\cN+1)^2. \end{split}\]
This controls the conjugation of the second contribution to $\wt{\cL}_N^{(1)}$ in \eqref{eq:wtLN1}. To deal with the first term on the right hand side of \eqref{eq:wtLN1}, we first observe that
		\[\begin{split}
		\left( N^3 V(N.) * \vert \pn\vert^2   -8\pi \frak{a}_0|\pn|^2 \right)\pn = h_N + \left( N^3 (Vf_\ell)(N.) * \vert \pn\vert^2   -8\pi \frak{a}_0|\pn|^2 \right)\pn
		\end{split}\]
and we find with Lemma \ref{3.0.sceqlemma} $ii)$ that
		\[ \begin{split}
		\|  N^3 (Vf_\ell)(N.) * \vert \pn\vert^2   -8\pi \frak{a}_0|\pn|^2 \|_\infty  \leq &\,  \sup_{x\in\bR^3}   \int dy\,   (Vf_\ell)(y) \big| \pn^2( x-y/N)- \pn^2(x)\big|   +\frac{C}{\ell N}\\
		 \leq &\, C\|\nabla\pn\|_\infty \|\pn\|\|V\| N^{-1}+C\ell^{-1}N^{-1}\leq C\ell^{-1}N^{-1}.
		\end{split} \]	
Note that $\|\nabla \pn\|_\infty\leq C$ by Appendix \ref{apx:gpfunctional}. By Cauchy-Schwarz, this implies that
		\[\pm e^{-B} \big(\sqrt{N} b\left(  \left( N^3 (Vf_\ell)(N.) * \vert \pn\vert^2   -8\pi \frak{a}_0|\pn|^2 \right)\pn\right)+\text{h.c.}  \big)e^{B} \leq C\ell^{-1}\]
and it only remains to control
		\[e^{-B}\sqrt{N} b( h_N ) e^{B} +\text{h.c.}  = \sqrt{N} b(\cosh_{\eta_H}(h_N) ) + \sqrt{N} b^*(\sinh_{\eta_H}(h_N) ) + \sqrt N d_{\eta_H}(h_N)+\text{h.c.} \]
However, by Eq. \eqref{eq:bnddf} from Lemma \ref{lm:d-bds}, we know that for all $\xi\in \cF^{\leq N}$ we have that
		\[  | \langle\xi,   \sqrt N d_{\eta_H}(h_N)\xi\rangle|   \leq C\ell^{\alpha/2}N^{-1/2}\|h_N\| \| (\cN+1)^{1/2}\xi\|\| (\cN+1)\xi\|\leq C\ell^{\alpha/2}   \langle\xi, (\cN+1)\xi\rangle, \]
and hence, collecting the previous estimates, we conclude the proposition.
\end{proof}

\subsection{Analysis of $\cG_{N}^{(2)}$}
From \eqref{eq:cLNj}, we recall that $\wt{\mathcal{L}}_N^{(2)} = \cK + \cV_{ext} + \wt{\cL}_N^{(2,V)}$, where
		\[
		\begin{split}
		\cK &= \int dx\; \nabla_x a_x^* \nabla_x a_x ,\hspace{0.5cm} \cV_{ext} =  \int dx\,  V_{ext}(x)a_x^{*} a_x \\
		\wt{\cL}_N^{(2,V)} &=  \int  dx \,\big(N^3 V(N.)\ast \vert \pn \vert^2\big)(x) \Big(b_x^* b_x - \frac{1}{N} a_x^* a_x \Big)   \\
		&\hspace{0.5cm}+ \int   dxdy\; N^3 V(N(x-y)) \pn(x) \pn(y) \Big( b_x^* b_y - \frac{1}{N} a_x^* a_y \Big)  \\
		&\hspace{0.5cm}+ \frac{1}{2} \int  dxdy\;  N^3 V(N(x-y)) \pn(y) \pn(x) \Big(b_x^* b_y^*  + \text{h.c.} \Big).
		\end{split}
		\]
In the following, we will analyze the contributions $ e^{-B} \cK e^{B}$, $ e^{-B} \cV_{ext} e^{B}$ and $ e^{-B} \wt{\cL}_N^{(2,V)} e^{B}$ separately. The analysis of the kinetic energy is the most involved so let us first treat the contributions $ e^{-B} \cV_{ext} e^{B}$ and $ e^{-B} \wt{\cL}_N^{(2,V)} e^{B}$.

\begin{prop}\label{prop:Vext} There exists a constant $C>0$ such that
		\[\begin{split}
		   e^{-B} \cV_{ext} e^{B} = \cV_{ext} + \cE_{N,\ell}^{(ext)} ,
		 \end{split}\]
where the self-adjoint operator $\cE_{N,\ell}^{(ext)}$ satisfies 
		\[ \pm \cE_{N,\ell}^{(ext)} \leq C\ell^{\alpha/2}(\cN+1).   \]
for all $\alpha>0$ and $\ell\in(0;1)$.
\end{prop}
\begin{proof}
We start with the observation that
		\[ e^{-B} \cV_{ext} e^{B} - \cV_{ext} =\frac12 \int_0^1ds\, \bigg(\int dx\, V_{ext}(x)e^{-sB} b^*_xb^*(\eta_{H,x})e^{sB}  \bigg) +\text{h.c.} \]
In the appendix we show that the assumptions (2) in \eqref{eq:asmptsVVext} imply that the external potential $V_{ext}(x)$ has at most exponential growth as $|x|\to\infty$ while, by \eqref{eq:expdecaypn}, the minimizer $\pn(x)$ has exponential decay as $|x|\to\infty$ with arbitrary rate. This implies in particular that $ \int dx\, V_{ext}^2(x) |\pn(x)|^2\leq C.$ As a consequence, Cauchy-Schwarz, Lemma \ref{lm:Npow} and Lemma \ref{lm:bndsetaH} imply that 
		\[\begin{split}
		&\bigg|   \int dx\, V_{ext}(x)\langle\xi, e^{-sB} b^*_xb^*(\eta_{H,x})e^{sB}\xi\rangle   \bigg| \\
		& \leq  \bigg( \int dx\, V_{ext}^2(x) \|\eta_{H,x}\|^2  \|  (\cN+1)^{1/2}e^{sB}\xi\|^2\bigg)^{1/2}\bigg( \int dx\,  \| a_xe^{sB}\xi\|^2\bigg)^{1/2}\\
		&\leq C\ell^{\alpha/2}\bigg( \int dx\, V_{ext}^2(x) |\pn(x)|^2\bigg)^{1/2} \langle \xi, (\cN+1)\xi \rangle\leq C\ell^{\alpha/2}\langle \xi, (\cN+1)\xi \rangle
		\end{split}\]
uniformly in $s\in [0;1]$, which proves the claim.
\end{proof}
\begin{prop}\label{prop:GNell2V} There exists a constant $C>0$ such that
		\[\begin{split}
		 &e^{-B} \wt{\cL}_N^{(2,V)} e^{B} \\
		 &=  \int dx \big(N^3 V(N.)\ast \vert \pn \vert^2\big)(x) b_x^* b_x   + \int dx dy \ N^3 V(N(x-y)) \pn(x)\pn(y) b_x^*b_y \\
		&\hspace{0.5cm} + \frac{1}{2} \int dx dy  \ N^3 V(N(x-y)) \pn(x) \pn(y)\big ( b_x^* b_y^*  + b_x b_y  \big)  \\
		&\hspace{0.5cm}+ \int dx dy \ N^3 V(N(x-y)) \pn(x) \pn(y) \eta_H(x;y) \left(\frac{N-\mathcal{N}}{N}\right)\left( \frac{N-\mathcal{N}-1}{N}\right) +\cE_{N,\ell}^{(2,V)},
		 \end{split}\]
where the self-adjoint operator $\cE_{N,\ell}^{(2,V)}$ satisfies 
		\[ \pm \cE_{N,\ell}^{(2,V)} \leq C\ell^{\alpha/2} (\cN+ \cV_N + 1).   \]
for all $\alpha>0$, $\ell\in(0;1)$ and $N\in\mathbb{N}$ sufficiently large.
\end{prop}
\begin{proof}
We split	$e^{-B}\mathcal{L}_N^{(2,V)} e^{B} = \sum_{j=1}^{5} F_j $, setting
	\begin{equation}\label{eq:F12345}\begin{split}
	F_1 &:= \int   dx \, \big(N^3 V(N.)\ast \vert \pn \vert^2\big)(x)  e^{-B} b_x^* b_x  e^{B}   , \\
	F_2 &:= \int dxdy \,  N^3 V(N(x-y)) \pn(x) \pn(y)  e^{-B} b_x^* b_y  e^{B}   ,  \\
	F_3 &:= -\frac{1}{N} \int dx  \, \big(N^3 V(N.)\ast \vert \pn \vert^2\big)(x) e^{-B} a_x^* a_x e^{B}    ,  \\
	F_4 &:=  -\frac{1}{N}\int dxdy \,  N^3 V(N(x-y)) \pn(x) \pn(y)  e^{-B} a_y^* a_x  e^{B}   ,\\
	F_5 &:=  \frac{1}{2} \int dxdy \, N^3 V(N(x-y)) \pn(y) \pn(x) e^{-B}\big ( b_x^* b_y^*  + b_x b_y  \big) e^{B}.
	\end{split}\end{equation}
To estimate these terms, we use the decomposition \eqref{eq:defdx} and the bounds in Lemma \ref{lm:d-bds}. With Lemma \ref{lm:Npow} and Cauchy-Schwarz, we obtain 
\begin{equation} \label{F1F2}\begin{split}
	\pm \bigg( F_1 -  \int dx \, \big(N^3 V(N.)\ast \vert \pn \vert^2\big)(x) b_x^* b_x \bigg)&\leq C \ell^{\alpha/2} (\mathcal{N}+1), \\
	\pm \bigg( F_2 - \int dx dy \, N^3 V(N(x-y)) \pn(x)\pn(y) b_x^* b_y\bigg) &\leq   C \ell^{\alpha/2} (\mathcal{N}+1)
	\end{split}
	\end{equation}
and, similarly, that
	\begin{equation} \label{F3F4}
	\pm F_3, \pm F_4 \leq N^{-1}  \|V\|_1\|\pn\|_\infty^2 (\mathcal{N}+1) \leq C\ell^{\alpha/2} (\mathcal{N}+1) 
	\end{equation}
for all $N$ sufficiently large. 
Hence, let us focus on the analysis of $F_5$. We split it into $F_5 = F_{51} + F_{52} + F_{53} $, where
	\begin{equation}\label{eq:defF515253}\begin{split}
	F_{51}
	&:= \frac{1}{2} \int dx dy \ N^3 V(N(x-y)) \pn(x) \pn(y) \\
	&\qquad \times \big(b(\cosh_{ \eta_H , x}) + b^*(\sinh_{\eta_H,x})\big)\big(b(\cosh_{ \eta_H , y}) + b^*(\sinh_{\eta_H,y})\big) + \text{h.c.},\\
	F_{52} &:=  \int dx dy \ N^3 V(N(x-y)) \pn(x) \pn(y) \big[d_{\eta_H,x} b(\cosh_{ \eta_H , y}) + d_{\eta_H,x} b^*(\sinh_{\eta_H,y}) \\
	&\hspace{4.5cm} +  b(\cosh_{ \eta_H , x})d_{\eta_H,y} + b^*(\sinh_{\eta_H,x})d_{\eta_H,y}\big] + \text{h.c.},\\
	F_{53} &:= \frac{1}{2} \int dx dy \ N^3 V(N(x-y)) \pn(x) \pn(y) d_{\eta_H,x} d_{\eta_H,y} + \text{h.c.}
	\end{split}\end{equation}
and consider the different contributions separately. We use $\sinh_{\eta_H} = \eta_H + \text{p}_{\eta_H}$ and $\cosh_{\eta_H} = 1 + \text{r}_{\eta_H}$ and the estimates in Lemma \ref{lm:d-bds} to rewrite $F_{51}$ as
	\begin{equation} \label{F51}
	\begin{split}
	F_{51}
	&= \frac{1}{2} \int dx dy  \ N^3 V(N(x-y)) \pn(x) \pn(y) \big( b_x b_y + b^*_x b^*_y\big)\\
	&\quad+ \int dx dy  \ N^3 V(N(x-y)) \pn(x) \pn(y) \eta_H(x;y) \left(\frac{N-\mathcal{N}}{N}\right) + \mathcal{E}^{V}_{51}
	\end{split}
	\end{equation}
	with
	\begin{equation} \label{F51 error}
	\pm \mathcal{E}_{51}^{V} \leq C \ell^{\alpha/2} (\mathcal{N}+1).
	\end{equation}
Let us switch to $F_{52}$, defined in \eqref{eq:defF515253}. We write $ F_{52} = F_{521} + F_{522} + F_{523} + F_{524}$ with
	\begin{align*}
	F_{521} &:= \frac{1}{2} \int dx dy \ N^3 V(N(x-y)) \pn(x) \pn(y) d_{\eta_H,x} b(\cosh_{ \eta_H , y})  + \text{h.c.},  \\
	F_{522} &:= \frac{1}{2} \int dx dy \ N^3 V(N(x-y)) \pn(x) \pn(y) d_{\eta_H,x} b^*(\sinh_{\eta_H,y}) + \text{h.c.},  \\
	F_{523} &:= \frac{1}{2} \int dx dy \ N^3 V(N(x-y)) \pn(x) \pn(y) b^*(\sinh_{\eta_H,x})d_{\eta_H,y} + \text{h.c.}, \\
	F_{524} &:= \frac{1}{2} \int dx dy \ N^3 V(N(x-y)) \pn(x) \pn(y) b(\cosh_{ \eta_H , x}) d_{\eta_H,y} + \text{h.c.}
	\end{align*}
Now, by applying \eqref{eq:bnddx}, we find that
	\begin{align*}
	\vert \langle \xi, d_{\eta_H,x} b(\cosh_{ \eta_H , y}) \xi \rangle \vert
	&\leq C \ell^{\alpha/2} \Vert (\mathcal{N}+1)^{1/2} \xi \Vert  \Big(  N^{-1/2}\Vert a_x a_y \xi \Vert + |\pn(x)|\Vert a_y \xi \Vert \Big) \\
	&\hspace{0.5cm}+C \ell^{\alpha/2}\|\text{r}_{\eta_H,y}\| \Vert (\mathcal{N}+1)^{1/2} \xi \Vert    \Big(\Vert a_x \xi \Vert +|\pn(x)|  \Vert (\mathcal{N}+1)^{1/2} \xi \Vert \Big) 
	\end{align*}
and, similarly, that 
	\begin{align*}
	\vert \langle \xi, d_{\eta_H,x} &b^*(\sinh_{\eta_H,y}) \xi \rangle \vert \\ 
	&\leq C \ell^{\alpha/2} \Vert (\mathcal{N}+1)^\frac{1}{2} \xi \Vert \Big( N^{-1}\vert \sinh_{\eta_H}(x;y)\vert \cdot \Vert (\mathcal{N}+1)^\frac{1}{2} \xi \Vert\Big)\\
	&\hspace{0.5cm} +  C \ell^{\alpha/2}\|\sinh_{\eta_H,y}\| \Vert (\mathcal{N}+1)^\frac{1}{2} \xi \Vert \Big(\Vert a_x   \xi \Vert  +|\ph (x)| \Vert (\mathcal{N}+1)^{1/2} \xi \Vert\Big).
	\end{align*}
Together with the fact that 
		$$ \sup_{x,y\,\in \bR^3} N^{-1}|\sinh_{\eta_H}(x;y) | \leq N^{-1}\big(| \eta_H(x;y) | + C\big)\leq C $$
by \eqref{eq:etapwbnd}, this yields together with the estimates in Lemma \ref{lm:d-bds}
	\begin{equation} \label{F521522}
	\pm F_{521} \leq C \ell^{\alpha/2} (\mathcal{N}  + \mathcal{V}_N + 1 ), \hspace{0.5cm} \pm F_{522} \leq C\ell^{\alpha/2} (\mathcal{N}+1).
	\end{equation}
A similar (but simpler) argument involving \eqref{eq:bnddx} also shows that 
		\begin{equation}\label{F523}
		\pm F_{523} \leq C \ell^{\alpha/2} (\mathcal{N}+1).
		\end{equation}
Next, let us switch to $ F_{524}$, defined above. Here, we first compute
	\begin{align*}
	b_x \mathcal{N} b^*(\eta_{H,y}) = b^*(\eta_{H,y}) b_x (\mathcal{N}+1) + \eta_H(x;y)  (1-\cN/N) (\mathcal{N}+1) -  a^*(\eta_{H,y}) a_x (\mathcal{N}+1)/N.
	\end{align*}
If we recall the notation $ \overline d_{\eta_H,y} = d_{\eta_H,y} + (\cN/N) b^*(\eta_{H,y})$, we therefore obtain
	\[
	F_{524} = - \int dx dy \ N^3 V(N(x-y)) \pn(x) \pn(y) \eta_H(x;y) \left(\frac{N-\mathcal{N}}{N}\right)\left( \frac{\mathcal{N}+1}{N}\right) + \mathcal{E}_{524}^{(V)}
	\]
where 
	\begin{align*}
	\mathcal{E}_{524}^{(V)} 
	&=\frac{1}{2}\int dx dy \ N^3 V(N(x-y)) \pn(x) \pn(y) \big[ b_x \overline{d}_{\eta_H,y} + b(\text{r}_{\eta_H,x}) d_{\eta_H,y} \\
	&\qquad \;\;-  N^{-1}b^*(\eta_{H,y}) b_x (\mathcal{N}+1) +  N^{-2}a^*(\eta_{H,y})a_x (\mathcal{N}+1) \big] +\text{h.c.} 
	\end{align*}
Now, using \eqref{eq:bnddx} and \eqref{eq:bndaydxbar} and proceeding as in the previous steps, we find that
	\begin{equation} \label{F523 error}
	\pm \mathcal{E}_{524}^{(V)}
	\leq C \ell^{\alpha/2} \big(\mathcal{N}  + \mathcal{V}_N +1 \big).
	\end{equation}
Collecting the previous bounds, this controls the contribution $ F_{52}$, defined in \eqref{eq:defF515253}. 	
	
Finally, to control the last contribution $F_{53} $, defined in \eqref{eq:defF515253}, we use \eqref{eq:bnddxdy} and the fact that $N^{-1} \vert\eta_H(x;y) \vert \leq C $ so that
	\begin{equation} \label{F53}
	\pm F_{53} \leq C \ell^{\alpha/2} \big(\mathcal{N}  + \mathcal{V}_N +1 \big).
	\end{equation}
In summary, if we combine  \eqref{eq:F12345}, \eqref{F1F2}, \eqref{F3F4}, \eqref{F51}, \eqref{F51 error}, \eqref{F521522}, \eqref{F523}, \eqref{F523 error} and \eqref{F53}, we obtain the claim.
\end{proof}

Finally, we study the kinetic energy. We start with two auxiliary lemmas. 
\begin{lemma}\label{lm:Deltarp}
	Assume \eqref{eq:asmptsVVext}, let $\ell\in (0;1)$ be sufficiently small and let $\alpha>0$. Then, there exists $C>0$, independent of $N$ and $\ell$, such that for all $n\geq 2$ and $j\in \{1,2\}$ 
	\begin{equation} \label{L2 estimate nabla widetilde eta_H, delta widetilde eta_H}
	\Vert \nabla_j \eta_H^{(n)} \Vert, \ \Vert \Delta_j \eta_H^{(n)} \Vert \leq C \ell^{\alpha/2} \Vert \eta_H \Vert^{n-2}.
	\end{equation}
	As a consequence, we have that
	\begin{equation} \label{L2 estimate nabla pr, delta pr}
	\Vert \nabla_j \emph{p}_{\eta_H} \Vert, \ \Vert \Delta_j \emph{p}_{\eta_H} \Vert  , \Vert \nabla_j \emph{r}_{\eta_H} \Vert, \ \Vert \Delta_j \emph{r}_{\eta_H} \Vert \leq C \ell^{\alpha/2}.
	\end{equation}
\end{lemma}
	\begin{proof}
	Notice first of all that \eqref{L2 estimate nabla pr, delta pr} follows indeed directly from \eqref{L2 estimate nabla widetilde eta_H, delta widetilde eta_H} and the definition \eqref{eq:defpretaH}. To prove 	\eqref{L2 estimate nabla widetilde eta_H, delta widetilde eta_H}, let us first consider $\Vert \Delta_1 \eta_H^{(n)} \Vert$. Since $\eta_H$ is symmetric, this takes also care of the case $j=2$. Moreover, note that it suffices to consider the case $n=2$. By the triangle inequality, we then have that
		\begin{align*}
		\Vert \Delta_1 \eta_H^{(2)} \Vert &\leq \left( \int dx dz \ \left\vert \int dy \ (\Delta G * \widecheck{\chi}_H)(x-y) \pn(x) \pn(y) \eta_{H}(y;z) \right\vert^2 \right)^{1/2}  \\
		&\quad+ \left(2 \int dx dz \left\vert \int dy \ (\nabla G * \widecheck{\chi}_H)(x-y) \cdot \nabla \pn(x) \pn(y) \eta_{H}(y;z) \right\vert^2  \right)^{1/2} \\
		&\quad + \left( \int dx dz \left\vert \int dy \ ( G * \widecheck{\chi}_H)(x-y) \Delta \pn(x) \pn(y) \eta_{H}(y;z) \right\vert^2  \right)^{1/2} \\
		&=: T_1 + T_2 + T_3.
		\end{align*}
Let us start with $T_1$. After switching to Fourier space, using the bound \eqref{eq:decphhat} and using \eqref{eq:whGbnd} that $\vert \widehat{\Delta G}(p) \vert = \vert p \vert^2 \vert \widehat{G}(p) \vert \leq C$, we get
		\begin{align*}
		 &T_1^2 
		= \int dr_1 dr_2 ds_1 ds_2 \ \wh{G}(r_1) \chi_H(r_1) \wh{G}(r_2) \chi_H(r_2) \widehat{\vert \pn \vert^2}(s_1+ s_2) \widehat{\vert \pn \vert^2}(r_1 - s_1) \\
		&\hspace{1cm}\times \widehat{\vert \pn \vert^2}(r_2-s_2) \widehat{\vert \pn \vert^2}(r_1 + r_2) \widehat{\Delta G}(s_1) \chi_H(s_1) \widehat{\Delta G}(s_2) \chi_H(s_2) \\
		&\leq C \int_{\vert r_1\vert, \vert r_2 \vert, \vert s_1 \vert, \vert s_2 \vert \geq \ell^{-\alpha}} dr_1 dr_2 ds_1 ds_2 \ \frac{1}{\vert r_1 \vert^2} \frac{1}{\vert r_2 \vert^2} \frac{1}{(1+ \vert s_1 + s_2 \vert)^4}\frac{1}{(1+\vert r_1 - s_1 \vert)^4} \\
		&\hspace{4cm}\times \frac{1}{(1+\vert r_2 - s_2\vert)^4}\frac{1}{(1+ \vert r_1 + r_2 \vert)^4} \\
		&\leq C \int_{\vert r_1\vert, \vert r_2 \vert \geq \ell^{-\alpha}} dr_1 dr_2 \ \frac{1}{\vert r_1 \vert^2} \frac{1}{\vert r_2 \vert^2} \frac{1}{(1+ \vert r_1 + r_2 \vert)^4} 
		\leq C \int_{\vert r \vert \geq \ell^{-\alpha}} dr\, \frac{1}{\vert r \vert^4} 
		= C \ell^\alpha.
		\end{align*}
For the second term $T_2$, we use that $\big\Vert \widehat{\vert \nabla \pn \vert^2} \big\Vert_\infty \leq \Vert \vert \nabla \pn \vert^2 \Vert_1\leq C$ by Lemma \ref{lem:gpmin2} and that $ | \widehat{\nabla G}(p) | \leq \ell^{\alpha} C$ for $|p|\geq \ell^{-\alpha}$, which yields by similar computation $T_2^2 \leq C \ell^{2\alpha}$. 
		Similarly one deals with $T_3$ and the bounds on $\Vert \nabla_j \eta_H^{(n)} \Vert$ can be proved analogously.
\end{proof}
\begin{lemma}\label{lm:nablad-bds}
Assume \eqref{eq:asmptsVVext}, let $\ell\in (0;1)$ be sufficiently small and let $\alpha>0$. Then, there exists a constant $C>0$ such that 
		\begin{equation} \label{estimate int a star nabla widetilde a nabla widetilde }
		\pm \int dx \ a^* (\nabla_x \eta_{H,x}) a(\nabla_x \eta_{H,x}) \leq C \ell^{\alpha/2} (\mathcal{N}+1).
		\end{equation}
Moreover, recalling \eqref{eq:defd}, \eqref{eq:defdx} and $ \overline{d}_{\eta,x} = d_{\eta,x}+(\cN/N)b^*(\eta_{x})$, we have that
		\begin{equation} \label{eq:nabladx1}
		\begin{split}
		&\left(\int dx \ \Vert \sqrt{N} (\mathcal{N} + 1)^{-1/2} \nabla_x d_{s\eta_H,x} \xi \Vert^2 \right)^{1/2} 
		\leq C    \Vert (\mathcal{N}+\cK+1)^{1/2} \xi \Vert  , \\
		&\left(\int dx \ \Vert \sqrt{N} (\mathcal{N} + 1)^{-1/2} \nabla_x \overline{d}_{s\eta_H,x} \xi \Vert^2 \right)^{1/2} 
		\leq C \ell^{\alpha/2}  \Vert (\mathcal{N}+\cK+1)^{1/2} \xi \Vert   
		\end{split}
		\end{equation}
and that
		\begin{equation} \label{eq:nabladx2}
		\begin{split}
		&\left( \int dx \ \Vert \sqrt{N} (\mathcal{N}+1)^{- 1/2} d_{s\eta_H}(\nabla_x \eta_{H,x}) \xi \Vert^2  \right)^{1/2}
		\leq C\ell^{\alpha/2} \Vert (\mathcal{N}+1)^{1/2} \xi \Vert, \\	
		&\left( \int dx \ \Vert \sqrt{N} (\mathcal{N}+1)^{- 1/2} d_{s\eta_H}^*(\nabla_x \eta_{H,x}) \xi \Vert^2  \right)^{1/2}
		\leq C\ell^{\alpha/2} \Vert (\mathcal{N}+1)^{1/2} \xi \Vert
		\end{split}
		\end{equation}
for all $\xi\in \cF^{\leq N}$ and all $s\in [0;1]$.
\end{lemma}
\begin{proof}
We first note that, using arguments very similar to those in the proof of the previous Lemma \ref{lm:Deltarp}, it is simple to show that
		\[
		 \int dy dz \left\vert \int dx \ \nabla_x \eta_H (y;x) \cdot \nabla_x \eta_H(z;x) \right\vert^2  \leq C \ell^{\alpha}.
		\]
Hence, Cauchy-Schwarz implies 
	\[\begin{split}
	&  \pm\bigg( \int dx \,   a^* (\nabla_x \eta_{H,x}) a(\nabla_x \eta_{H,x})  \bigg)  \\
	& =  \pm\bigg( \int  dydz\, \bigg[ \int dx\,    \nabla_x \eta_{H}(y;x)  \nabla_x \eta_{H}(z;x) \bigg] a^*_ya_z\bigg) \leq C \ell^{\alpha}   (\mathcal{N}+1). 
	\end{split}\]
This proves the first bound \eqref{estimate int a star nabla widetilde a nabla widetilde }. To prove the bounds \eqref{eq:nabladx1} and \eqref{eq:nabladx2}, we proceed similar as in the proof of Lemma \ref{lm:d-bds} (which can be found in \cite[Lemma 3.4]{BBCS4}) and use Lemma \ref{lm:bndsetaH}, Lemma \ref{lm:Deltarp} as well as the bound \eqref{estimate int a star nabla widetilde a nabla widetilde }; we skip the details.
\end{proof}
We are now ready to analyze the kinetic energy. 
\begin{prop}\label{prop:e-BKeB} There exists a constant $C>0$ such that
		\[\begin{split}
		  e^{-B }\cK e^{B}= \,&\mathcal{K} + \int dx \ \big[b(\nabla_x \eta_{H,x}) \nabla_x b_x + \emph{h.c.}\big] \\
		  &+ \Vert \nabla_1 \eta_H \Vert^2 (1-\cN/N)  (1-\cN/N-1/N) + \mathcal{E}_{N, \ell}^{(K)}
		 \end{split}\]
where the self-adjoint operator $\cE_{N,\ell}^{(K)}$ satisfies 
		\[ \pm \cE_{N,\ell}^{(K)} \leq  C\ell^{(\alpha-3)/2} (\cN+ \cK+\cV_N +1) + C\ell^{-5\alpha/2} N^{-1}(\cN+1)^2   \]
for all $\alpha>3$, $\ell\in(0;1)$ and all $N$ large enough. 
\end{prop}
\begin{proof}
Using a first order Taylor expansion and the identity \eqref{eq:defd}, we have that
	\begin{equation}\label{eq:defG123}
	\begin{split}
	&e^{-B} \mathcal{K} e^{B} - \mathcal{K} \\
	&= \bigg(\int_0^1 ds \int dx \   \big[ b(\cosh_{s\eta_H}(\nabla_x \eta_{H,x})) + b^*(\sinh_{s\eta_H}(\nabla_x \eta_{H,x}) \big] \\
	&\hspace{2.5cm} \times \big[ \nabla_x b(\cosh_{s \eta_H , x}) + \nabla_x b^*(\sinh_{s\eta_H,x})  \big] +\text{h.c.}\bigg) \\
	&\quad + \bigg(\int_0^1 ds \int dx \Big[  \big(b(\cosh_{s\eta_H}(\nabla_x \eta_{H,x})) + b^*(\sinh_{s\eta_H}(\nabla_x \eta_{H,x}))  \big) \nabla_x d_{s\eta_H, x} \\
	&\hspace{3cm}+ d_{s\eta_H}(\nabla_x \eta_{H,x})  \big( \nabla_x b(\cosh_{s \eta_H , x}) + \nabla_x b^*(\sinh_{s\eta_H,x}) \big)  \Big] +\text{h.c.}\bigg) \\
	&\quad +\bigg( \int_0^1 ds \int dx \  d_{s\eta_H}(\nabla_x \eta_{H,x}) \nabla_x d_{s\eta_H, x}  +\text{h.c.}\bigg)\\
	&=: \text{G}_1 + \text{G}_2 + \text{G}_3.
	\end{split}
	\end{equation}

Let us start to analyze $\text{G}_1$. Integrating by parts and using \eqref{eq:bndpr}, Lemma \ref{lm:Deltarp} as well as the bound \eqref{estimate int a star nabla widetilde a nabla widetilde } we conclude that 
	\begin{equation}\label{eq:G1bnd}
	\text{G}_1 = \int dx \ \big[b(\nabla_x \eta_{H,x}) \nabla_x b_x + \text{h.c.}\big] + \Vert \nabla_1 \eta_H \Vert^2  (1-\cN/N) + \mathcal{E}_1,
	\end{equation}
	where the error $\cE_1$ satisfies
	\[
	\pm \mathcal{E}_1 \leq C \ell^{\alpha/2} (\mathcal{N}+1).
	\]	
	
Next we extract the relevant terms from $\text{G}_2$, defined in \eqref{eq:defG123}. We split $\text{G}_2$ into
	\begin{align*}
	\text{G}_2
	&=  \int_0^1 ds \int dx \ \big[ b^*(\sinh_{s\eta_H}(\nabla_x \eta_{H,x})) \nabla_x d_{s\eta_H, x}  + \text{h.c.} \big] \\
	&\quad + \int_0^1 ds \int dx \ \big[d_{s\eta_H}(\nabla_x \eta_{H,x}) \nabla_x b(\cosh_{s \eta_H , x}) + \text{h.c.}\big] \\
	&\quad + \int_0^1 ds \int dx \ \big[ b(\cosh_{s\eta_H} (\nabla_x \eta_{H,x})) \nabla_x d_{s\eta_H, x} + \text{h.c.} \big] \\
	&\quad + \int_0^1 ds \int dx \ \big[d_{s\eta_H}(\nabla_x \eta_{H,x}) \nabla_x b^*(\sinh_{s \eta_H , x}) + \text{h.c.}\big] =:  \text{G}_{21}+\text{G}_{22}+\text{G}_{23}+\text{G}_{24}.
	\end{align*}
From Lemma \ref{lm:Deltarp}, Lemma \ref{lm:nablad-bds} and $\Vert \sinh_{\eta_H}(\nabla_x \eta_{H,x}) \Vert \leq C\Vert \nabla_x \eta_H^{(2)} \Vert$, we easily find that
		\[ \pm \text{G}_{21}\leq C\ell^{\alpha/2} (\cN+\cK+1),\hspace{0.5cm} \pm \text{G}_{22}\leq C\ell^{\alpha/2} (\cN+\cK+1), \]
so let us continue with the analysis of $\text{G}_{23}$. We split it into
	\begin{align*}
	G_{23}
	&= \int_0^1 ds \int dx \ \left[ b\big( \text{r}_{s\eta_H}( -\Delta_x \eta_{H,x})\big) d_{s\eta_H, x} + \text{h.c.} \right] \\
	&\quad + \int_0^1 ds \int dx \ \left[ b(\nabla_x \eta_{H,x}) \nabla_x \overline{d}_{s\eta_H,x} + \text{h.c.} \right] \\
	&\quad - \int_0^1 ds \int dx \ \left[ b(\nabla_x \eta_{H,x})  (\cN/N) b^*(\nabla_x s\eta_{H,x}) + \text{h.c.} \right]=: G_{231} + G_{232}+ G_{233}
	\end{align*}
and, similarly as above, it is simple to see that
	\[\pm G_{231} \leq C\ell^{\alpha/2}(\mathcal{N}+1).\]
To control the contribution $\text{G}_{232}$, we use that
	\begin{align*}
	- \Delta_x \eta_H (y;x) &= - \Delta G (y-x) \pn(x) \pn(y) +2\nabla G (y-x) \nabla \pn(x) \pn(y) \\
	&\quad - G(y-x)  \Delta\pn(x) \pn(y) + \Delta_x\big[(G * \widecheck{\chi}_{H^c})(y-x) \pn(x) \pn(y)\big]
	\end{align*}
so that integrating by parts implies
	\begin{align*}
	G_{232}
	&= \int_0^1 ds \int dx dy \ ( -\Delta G)(y-x) \pn(x) \pn(y)  b_y \overline{d}_{s\eta_H,x}   \\ 
	&\quad + \int_0^1 ds \int dx dy \  G(y-x)  \Delta\pn(x)   \pn(y) b_y \overline{d}_{s\eta_H,x}  \\
	&\quad + \int_0^1 ds \int dx dy \ G (y-x) \nabla \pn(x) \pn(y) b_y \nabla_x \overline{d}_{s\eta_H,x}  \\
	&\quad + \int_0^1 ds \int dx dy \ \Delta_x \big((G * \widecheck{\chi}_{H^c})(y-x) \pn(x) \pn(y)\big)b_y \overline{d}_{s\eta_H,x}  +\text{h.c.} 
	\end{align*}
Using the scattering equation \eqref{eq:scatlN}, we have that 
		\[\begin{split}
		&\int_0^1 ds \int dx dy \ ( -\Delta G)(y-x) \pn(x) \pn(y)  b_y \overline{d}_{s\eta_H,x} \\
		& =  -\int_0^1 ds \int dx dy \left[ N^3 V(N(y-x)) - N^3\lambda_\ell  f_\ell(N(y-x)) \right] \chi_\ell(x-y) \pn(x) \pn(y) b_y \overline{d}_{s\eta_H,x}
		\end{split}\]
With Lemma \ref{3.0.sceqlemma} and Lemma \ref{lm:d-bds}, we conclude, proceeding 
in the usual way, 	
		\[\begin{split}
		&\bigg|\int_0^1 ds \int dx dy \ ( -\Delta G)(y-x) \pn(x) \pn(y)  \langle \xi, b_y \overline{d}_{s\eta_H,x} \xi\rangle \bigg|\\
		&\leq C\ell^{\alpha/2} \int dx dy \left[ N^3 V(N(x-y)) + C\ell^{-3} \chi_\ell(x-y)  \right] | \pn(x) | | \pn(y)|  \Vert (\mathcal{N}+1)^{1/2} \xi \Vert \\
		&\hspace{2cm} \times N^{-1}   \big[   \ell^{\alpha/2}  \Vert (\mathcal{N}+1)^{1/2} \xi \Vert + |\eta_H(x;y)|   \Vert (\mathcal{N}+1)^{1/2} \xi \Vert +   \ell^{\alpha/2} \Vert a_x  \xi \Vert \\
		& \hspace{3.2cm} +  \ell^{\alpha/2} \Vert a_y (\mathcal{N}+1)  \xi \Vert +  \Vert a_x a_y (\mathcal{N}+1)^{1/2} \xi \Vert \Big] \\
		&\leq   C\ell^{(\alpha-3)/2}\langle\xi, (\cN+1)\xi\rangle + C\ell^{\alpha/2}\langle \xi,  \cV_N \xi\rangle.
		\end{split}\]
Using once more Lemma \ref{lm:d-bds} and Lemma \ref{lm:nablad-bds}, we also find that
	\[\begin{split}
	&\bigg| \int_0^1 ds \int dx dy \  G(y-x)  \Delta\pn(x)   \pn(y) \langle\xi, b_y \overline{d}_{s\eta_H,x}\xi\rangle \bigg| \\
	& + \bigg|\int_0^1 ds \int dx dy \ G (y-x) \nabla \pn(x) \pn(y) \langle\xi, b_y \nabla_x \overline{d}_{s\eta_H,x}\xi\rangle \bigg|\\
	&\leq C   \int_0^1 ds\, \Vert (\mathcal{N}+1)^{1/2} \xi \Vert  \bigg[\left( \int dx \ \Vert \nabla_x \overline{d}_{s\eta_H,x} \xi \Vert^2 \right)^{1/2}  + \left( \int dx \ \Vert \overline{d}_{s\eta_H,x} \xi \Vert^2 \right)^{1/2}\bigg] \\
	&\leq C \ell^{\alpha/2} \langle \xi, (\mathcal{N} + \mathcal{K} +1)\xi \rangle
	\end{split}\]
and, since \[ \vert \Delta_x [(G * \widecheck{\chi}_{H^c})(x-y) \pn(x) \pn(y)] \vert \leq C\ell^{-3\alpha} [\vert \pn(x) \vert + \vert \nabla_x \pn(x) \vert + \vert \Delta_x \pn(x) \vert] \vert \pn(y)\vert ,\] we have furthermore by Cauchy-Schwarz and \eqref{eq:bndaydxbar} that 
	\[\begin{split}
	&\bigg| \int_0^1 ds \int dx dy \ \Delta_x \big((G * \widecheck{\chi}_{H^c})(y-x) \pn(x) \pn(y)\big) \langle\xi, b_y \overline{d}_{s\eta_H,x} \xi\rangle \bigg| \leq \frac{C\ell^{-5\alpha/2}}{N}  \langle\xi, (\cN+1)^{2}\xi\rangle.   
	\end{split}\]
In summary, this proves that 
		\[\pm \text{G}_{232}\leq  C\ell^{(\alpha-3)/2} (\cN+1) + C\ell^{\alpha/2}(\cK+\cV_N) + C\ell^{-5\alpha/2} N^{-1}(\cN+1)^2. \]
and since
	\begin{align*}
	G_{233}
	&= -\Vert \nabla_1 \eta_H \Vert^2\, \frac{\mathcal{N}+1}{N} \frac{N-\cN}{N}  -  \int dx \   a^*(\nabla_x \eta_{H,x}) a(\nabla_x \eta_{H,x})\frac{\mathcal{N}+1}{N} \frac{N-\cN}{N},  
	\end{align*}
we easily deduce, with Lemma \ref{lm:nablad-bds}, that
	\[\begin{split} \pm &\bigg(G_{23}  +\Vert \nabla_1 \eta_H \Vert^2 \,\frac{\mathcal{N}+1}{N} \frac{N-\cN}{N}\,\bigg) \\
	 &\hspace{1.5cm}\leq C\ell^{(\alpha-3)/2} (\cN+1) + C\ell^{\alpha/2}(\cK+\cV_N) + C\ell^{-5\alpha/2} N^{-1}(\cN+1)^2. \end{split}\]
		
With very similar arguments, one can show that 
		\[\pm \text{G}_{24}\leq C\ell^{\alpha/2}(\cN+1) \]
so that, in summary, we have
		\begin{equation}\label{eq:G2bnd}
		\text{G}_2 = -\Vert \nabla_1 \eta_H \Vert^2 \,\frac{\mathcal{N}+1}{N} \frac{N-\cN}{N} + \cE_2
		\end{equation}
for an error $\cE_{2}$ that satisfies 
		\[\pm\cE_2 \leq C\ell^{(\alpha-3)/2} (\cN+1) + C\ell^{\alpha/2}(\cK+\cV_N) + C\ell^{-5\alpha/2} N^{-1}(\cN+1)^2. \]

Going back to \eqref{eq:defG123}, we finally use once more Lemmas \ref{lm:d-bds} and \ref{lm:nablad-bds} to deduce		
		\begin{equation} \label{eq:G3bnd} \pm \text{G}_3\leq C\ell^{\alpha/2}(\cN+ \cK+1).
		\end{equation}
Hence, collecting \eqref{eq:G1bnd}, \eqref{eq:G2bnd} and \eqref{eq:G3bnd} proves the proposition.
\end{proof}

\subsection{Analysis of $\cG_{N}^{(3)}$}
From \eqref{eq:cLNj}, we recall that 
		\[
		\begin{split}
		\wt{\mathcal{L}}_N^{(3)} &=    \int dxdy\,  N^{ 5/2} V(N(x-y)) \pn(y) \big( b_x^* a^*_y a_x + \text{h.c.} \big) 
		\end{split}
		\]
Let us also recall the definition of $h_N = \big(N^3(Vw_\ell) (N.)\ast |\pn|^2\big)\pn$ from \eqref{eq:defhN}. 
\begin{prop}\label{prop:GNell3} There exists a constant $C>0$ such that
		\[\begin{split}
		 \cG_{N}^{(3)} = &\,  \int dxdy\,  N^{ 5/2} V(N(x-y)) \pn(y)  \big( b_x^* a^*_y a_x + \emph{h.c.}\big) \\
		 &\; -  \Big[ \sqrt N b\big(\cosh_{\eta_H} (h_N) \big) +\sqrt N b^*\big(\sinh_{\eta_H} (h_N) \big) + \emph{h.c.} \Big] + \cE_{N,\ell}^{(3)}, 
		 \end{split}\]
where the self-adjoint operator $\cE_{N,\ell}^{(3)}$ satisfies 
		\[ \pm \cE_{N,\ell}^{(3)} \leq C \ell^{\alpha/2} ( \mathcal{N} + \mathcal{V}_N + 1) +  C N^{-1/2} (\mathcal{N}+1)^{3/2} + C\ell^{-\alpha}.   \]
for all $\alpha>0$, $\ell\in(0;1)$ and $N$ large enough.
\end{prop}
\begin{proof}
We use the identity 
		\[ e^{-B}a^*_y a_x e^{B} = a^*_y a_x  + \int_0^1 ds\, e^{-sB} \big[ b(\eta_{H,y}) b_x+ b^*_y b^*(\eta_{H,x}) \big]  e^{sB}\]
to split $ \cG_{N}^{(3)}$ into  $ \cG_{N}^{(3)} = \text{J}_1+ \text{J}_2 + \text{J}_3 +\text{h.c.}$, where
		\begin{equation}\label{eq:J123}
		\begin{split}
		\text{J}_1 &:= \int dxdy\,  N^{ 5/2} V(N(x-y)) \pn(y)   e^{-B}b_x^*e^{B} a^*_y a_x , \\
		\text{J}_2 &:=\int dxdy\,  N^{ 5/2} V(N(x-y)) \pn(y)   e^{-B}b_x^*e^{B}\int_0^1 ds\, e^{-sB}  b(\eta_{H,y}) b_x e^{sB}, \\
		\text{J}_3 &:= \int dxdy\,  N^{ 5/2} V(N(x-y)) \pn(y)   e^{-B}b_x^*e^{B}\int_0^1 ds\, e^{-sB}   b^*_y b^*(\eta_{H,x}) e^{sB}. 
		\end{split}
		\end{equation}
We start with the analysis of $\text{J}_1$. Using \eqref{eq:defd}, we have that 
		\begin{equation}\label{eq:J1split}\begin{split}
		\text{J}_1 &= \int dxdy\,  N^{ 5/2} V(N(x-y)) \pn(y)    b^*_x  a^*_y a_x \\
		&\hspace{0.5cm} + \int dxdy\,  N^{ 5/2} V(N(x-y)) \pn(y)   \big[  b^*(\text{r}_{\eta_H,x}) + b(\text{p}_{\eta_H,x})\big]  a^*_y a_x  \\
		&\hspace{0.5cm} + \int dxdy\,  N^{ 5/2} V(N(x-y)) \pn(y)    d_{\eta_H,x}^*a^*_y a_x \\
		&\hspace{0.5cm} + \int dxdy\,  N^{ 5/2} V(N(x-y)) \pn(y)  b(\eta_{H,x})a^*_y a_x\\
		&=: \int dxdy\,  N^{ 5/2} V(N(x-y)) \pn(y)    b^*_x  a^*_y a_x  + \text{J}_{11}+\text{J}_{12}+\text{J}_{13}.
		\end{split}\end{equation}
First, it is simple to see that 	
		\[ \pm   \text{J}_{11}   \leq  C\ell^{\alpha} (\cN+1)\]
and, by \eqref{eq:bndpr} and the bound \eqref{eq:bndaydxbar} from Lemma \ref{lm:d-bds}, we also find that
		\[\begin{split}
		|\langle \xi, \text{J}_{12}\xi\rangle| &\leq \int dxdy\,  N^{ 5/2} V(N(x-y)) |\pn(y)|  \| a_y \overline{d}_{\eta_H,x}\xi\| \|a_x\xi\| \\
		&\hspace{0.5cm} +   \int dxdy\,  N^{ 5/2} V(N(x-y)) |\pn(y)|  \| a_y  (\cN^{3/4}/N) a^*(\eta_{H,x})\xi\| \|\cN^{1/4}a_x\xi\|\\
		&\leq C\ell^{\alpha/2} \|(\cN+\cV_N+1)^{1/2}\xi\|^2 + C N^{-1/2} \| (\cN+1)^{3/4}\xi\|^2  .
		\end{split}\]		
Note that we used $ |\eta_H(x;y)|\leq CN $ for all $N$ large enough, by \eqref{eq:etapwbnd}. 

Going back to \eqref{eq:J1split} and recalling the definition \eqref{eq:defhN}, we finally see that 	
		\[\begin{split}
		\text{J}_{13} &= \int dxdy\,  N^{ 5/2} V(N(x-y)) \pn(y) \big[  \eta_{H}(y;x) b_x  + a^*_y a_xb(\eta_{H,x})\big] \\
		& = -\sqrt N b(h_N)  - \int dxdy\,  N^{ 5/2} V(N(x-y))(G\ast \widecheck{\chi}_{H^c}) (x-y)  \pn^2 (y) \pn(x)   b_x  \\
		&\hspace{0.5cm} + \int dxdy\,  N^{ 5/2} V(N(x-y))\pn(y) a^*_y a_x b(\eta_{H,x}), 
		\end{split}\]
where $ \widecheck{\chi}_{H^c}$ denotes the inverse Fourier transform of the characteristic function of the set $\{p\in \bR^3: |p|\leq \ell^{-\alpha}\}$. Using that $\|(G\ast \widecheck{\chi}_{H^c})\|_\infty\leq C\ell^{-\alpha}$, by \eqref{eq:whGbnd}, we deduce that
		\[ \pm \big(\text{J}_{13} +\sqrt N b(h_N)  \big) \leq  C\ell^{-\alpha} + C\ell^{\alpha/2}(\cN+1) \]
and hence, if we collect the previous estimates, we have proved that 
		\begin{equation}\label{eq:bndJ1}
		\begin{split}
		&\pm \bigg(J_1 - \int dxdy\,  N^{ 5/2} V(N(x-y)) \pn(y)    b^*_x  a^*_y a_x + \sqrt N b(h_N) \bigg) \\
		&\hspace{3cm} \leq C\ell^{\alpha/2}(\cN+ \cV_N +1) +  C\ell^{-\alpha} N^{-1/2}(\cN+1)^{3/2}  + C \ell^{-\alpha}.
		\end{split}
		\end{equation}	
	
Next, we bound $\text{J}_2$, defined in \eqref{eq:J123}. We apply as usual the identity \eqref{eq:defd}, Lemma \ref{lm:d-bds} and Cauchy-Schwarz to estimate
		\begin{equation}\label{eq:J2bnd} \begin{split}
		| \langle \xi, \text{J}_2\xi\rangle| &\leq \int dxdy\,  N^{ 5/2} V(N(x-y)) |\pn(y)| \|\eta_{H,y}\|  \\
		&\hspace{1cm}\times  \int_0^1ds\,  \|    b(\cosh_{\eta_H,x})\xi + b^*(\sinh_{\eta_H,x})  \xi  \|\| (\cN+1)^{1/2} b_x e^{sB}\xi\|\\
		&\hspace{0.5cm} + \int dxdy\,  N^{ 5/2} V(N(x-y)) |\pn(y)|  \|\eta_{H,y}\|   \\
		&\hspace{1cm}\times \int_0^1ds\,\|   d_{\eta_H,x}  \xi  \| \| (\cN+1)^{1/2} b_x e^{sB}\xi\| \\
		&\leq C\ell^{\alpha/2}  \| (\cN+1)^{1/2}\xi\|^2\\
		&\hspace{0.5cm}+  C\ell^{\alpha}\int dxdy\,  N^{ 5/2} V(N(x-y))   \\
		&\hspace{1cm}\times \int_0^1ds\, \big[ \|   a_x  \xi  \| + |\pn(x)|  \|   (\cN+1)^{1/2}  \xi  \|\big]\| (\cN+1)^{1/2} b_x e^{sB}\xi\| \\
		&\leq C\ell^{\alpha/2}  \| (\cN+1)^{1/2}\xi\|^2.
		\end{split}\end{equation}

Finally, let us analyze the contribution $\text{J}_3$, defined in \eqref{eq:J123}. We split this contribution into $ \text{J}_3 = \text{J}_{31} + \text{J}_{32}+ \text{J}_{33}$, where 
		\[\begin{split}
		\text{J}_{31} &= \int dxdy\,  N^{ 5/2} V(N(x-y)) \pn(y)   e^{-B}b_x^*e^{B}\\
		&\hspace{1cm} \times \int_0^1 ds\, \Big(  e^{-sB}   b^*_y e^{sB} - b_y^*  \Big) e^{-sB} b^*(\eta_{H,x}) e^{sB} ,\\
		\text{J}_{32} & = \int dxdy\,  N^{ 5/2} V(N(x-y)) \pn(y)   \Big( e^{-B}b_x^*e^{B} - b(\eta_{H,x}) \Big)b^*_y\\
		&\hspace{1cm} \times \int_0^1 ds\,   e^{-sB}   b^*(\eta_{H,x}) e^{sB}, \\
		\text{J}_{33} & = \int dxdy\,  N^{ 5/2} V(N(x-y)) \pn(y) b(\eta_{H,x})b^*_y  \int_0^1 ds\, e^{-sB}  b^*(\eta_{H,x}) e^{sB}. 
		\end{split}\]
To control the error terms $\text{J}_{31}$ and $\text{J}_{32}$, we proceed as before to bound 
		\[\begin{split} 
		|\langle \xi, \text{J}_{31}\xi \rangle | &\leq  \int dxdy\,  N^{ 5/2} V(N(x-y)) |\pn(y)|  \| b_x e^{B}\xi\| \\
		&\hspace{1cm} \times \int_0^1 ds\,\big\| \big[b^*(\text{r}_{s\eta_H,y}) +b(\sinh_{s\eta_H,y})\big] e^{-sB} b^*(\eta_{H,x}) e^{sB}\xi \big\| \\
		&\hspace{0.5cm} + C\ell^{\alpha/2}\int dxdy\,  N^{ 5/2} V(N(x-y)) |\pn(y)|  |\pn(x)|\| (\cN+1)^{1/2}\xi \big\|  \\
		&\hspace{1cm}\times  \int_0^1ds\, \big\| d_{s\eta_H,y} \big[ b(\cosh_{\eta_H,x}) + b^*(\sinh_{\eta_H,x}) + d_{\eta_H,x}\big]  \xi \big\| \\
		&\leq C\ell^{\alpha/2} \langle\xi , (\cN+\cV_N+1)\xi\rangle
		\end{split}\]
as well as
		\[\begin{split} 
		|\langle \xi, \text{J}_{32}\xi \rangle | &\leq  \int dxdy\,  N^{ 5/2} V(N(x-y)) |\pn(y)|  \| b_x e^{B}\xi\| \\
		& \leq  C\ell^{\alpha/2}\int dxdy\,  N^{ 5/2} V(N(x-y)) |\pn(y)|  |\pn(x)|\| (\cN+1)^{1/2}\xi \big\|  \\
		&\hspace{1.5cm}\times  \big\|  b_y \big[ b(\cosh_{\eta_H,x}) + b^*(\text{p}_{\eta_H,x}) + \overline{d}_{\eta_H,x} - \frac{\cN}{N} b^*(\eta_{H,x})\big]  \xi \big\| \\
		&\leq C\ell^{\alpha/2} \langle\xi , (\cN+\cV_N+1)\xi\rangle. 
		\end{split}\]
	Here we used in the in the last step \eqref{eq:bndaydxbar} to bound the term with $\overline{d}$ and further used $b_y b^*(\eta_{H,x}) = a^*(\eta_{H,x}) a_y \left(1-\cN/N \right)+ \eta_H(x,y) \left(1-\cN /N \right)$ and $\Vert \eta_H \Vert_{\infty} \leq C N$.
To control the last term $\text{J}_{33}$, on the other hand, we first rewrite it as  
		\[\begin{split}
		\text{J}_{33} & = -\sqrt N \int dxdy\,  N^{ 3} (Vw_\ell)(N(x-y)) |\pn(y)|^2\pn(x) \\
		&\hspace{1cm}\times \int_0^1 ds\, \big[ b^*(\cosh_{s\eta_H}(\eta_{H,x})) + b(\sinh_{s\eta_H}(\eta_{H,x}))\big] \\
		&\hspace{0.3cm} - \int dxdy\,  N^{ 5/2} V(N(x-y))(G\ast \widecheck{\chi}_{H^c} )(x-y) |\pn(y)|^2\pn(x)  \\
		&\hspace{1cm}\times \int_0^1ds\,  \big[ b^*(\cosh_{s\eta_H}(\eta_{H,x})) + b(\sinh_{s\eta_H}(\eta_{H,x}))\big] \\
		&\hspace{0.3cm} + \int dxdy\,  N^{ 5/2} V(N(x-y)) \pn(y)  \eta_H(x;y) \int_0^1ds\, d^*_{s\eta_H}(\eta_{H,x}) \\
		&\hspace{0.3cm} + \int dxdy\,  N^{ 5/2} V(N(x-y)) \pn(y)a^*_y  a(\eta_{H,x}) (1-\cN/N) \int_0^1 ds\, e^{-sB}  b^*(\eta_{H,x}) e^{sB},
		\end{split}\]
where $ \widecheck{\chi}_{H^c}$ denotes the inverse Fourier transform of the characteristic function of the set $ \{p\in\bR^3: |p|> \ell^{-\alpha}\}$. With very similar arguments as before, we find that  
		\[\begin{split}
		&\pm\bigg( \text{J}_{33} + \sqrt N \int dxdy\,  N^{ 3} (Vw_\ell)(N(x-y)) |\pn(y)|^2\pn(x) \\
		&\hspace{2cm}\times \int_0^1 ds\, \big[ b^*(\cosh_{s\eta_H}(\eta_{H,x})) + b(\sinh_{s\eta_H}(\eta_{H,x}))\big] \bigg)\\
		&\hspace{0.5cm}\leq C\ell^{\alpha/2} (\cN +1)+ C\ell^{-\alpha}.
		\end{split}\]	
Finally, if we recall the definition \eqref{eq:defhN}, we observe that
		\[\begin{split}
		&    \int dxdy\,  N^{ 3} (Vw_\ell)(N(x-y)) |\pn(y)|^2\pn(x)  \int_0^1 ds\, \big[ b^*(\cosh_{s\eta_H}(\eta_{H,x})) + b(\sinh_{s\eta_H}(\eta_{H,x}))\big] \\
		& =  \int dxdy\,  N^{ 3} (Vw_\ell)(N(x-y)) |\pn(y)|^2\pn(x)  \big[ b^*(\sinh_{\eta_H,x}) + b(\text{r}_{\eta_H,x})\big]\\
		& = \Big[ b\big(\cosh_{\eta_H}(h_N)\big)  +  b^*\big(\sinh_{\eta_H}(h_N)\big)\Big] -   b(h_N).
		\end{split} \]
Hence, the previous bounds together with \eqref{eq:bndJ1} and \eqref{eq:J2bnd} prove the proposition. 		
\end{proof}

\subsection{Analysis of $\cG_{N}^{(4)}$}
From \eqref{eq:cLNj}, we recall that 
		\[
		\begin{split}
		\wt{\mathcal{L}}_N^{(4)} &= \frac{1}{2} \int  dxdy\; N^2 V(N(x-y)) a_x^* a_y^* a_y a_x.
		\end{split}
		\]
For the analysis of $\cG_{N}^{(4)}=e^{-B}\wt{\mathcal{L}}_N^{(4)}e^{B}$, we will use the following Lemma which is a straightforward consequence of Lemmas \ref{lm:bndsetaH}, \ref{lm:d-bds} and the decomposition \eqref{eq:defd}; we omit its proof.
\begin{lemma}\label{lm:auxquartic}
	Assume \eqref{eq:asmptsVVext} and let $\ell\in(0;1)$ be sufficiently small. Then, there exists a constant $C>0$ such that 
	\[\begin{split}
	&\Vert (\mathcal{N}+1)^{n/2} e^{-sB} b_x b_y e^{sB} \xi \Vert \\
	&\leq C\ell^{\alpha} |\pn(x)||\pn(y)| \Vert (\mathcal{N}+1)^{(n+2)/2} \xi \Vert+C\ell^{\alpha/2} |\pn(y)| \Vert a_x (\mathcal{N}+1)^{(n+1)/2} \xi \Vert   \\
	&\quad+C\ell^{\alpha/2}|\pn(x)| \Vert a_y (\mathcal{N}+1)^{(n+1)/2} \xi \Vert  + \vert \eta_H(x;y) \vert   \Vert (\mathcal{N}+1)^{n/2}\xi \Vert +\Vert a_x a_y (\mathcal{N}+1)^{n/2} \xi \Vert
		\end{split}\]
for all $\xi \in \cF^{\leq N}$ and all $s\in [0;1]$. 
\end{lemma}

\begin{prop}\label{prop:GNell4} There exists a constant $C>0$ such that
		\[\begin{split}
		 \cG_{N}^{(4)} = &\,\mathcal{V}_N - \frac{1}{2} \int dx dy \ N^3 (Vw_\ell)(N(x-y)) \pn(x)\pn(y)\big(b_x b_y + b_x^* b_y^*\big)\\
	& + \frac{N}{2} \int dx dy \ N^3 V(N(x-y))  w_\ell^2(N(x-y)) |\pn(x)|^2|\pn(y)|^2  \\
		&\hspace{1.8cm}\times (1-\cN/N) (1-\cN/N-1/N)+\cE_{N,\ell}^{(4)}\,,
		 \end{split}\]
where the self-adjoint operator $\cE_{N,\ell}^{(4)}$ satisfies 
		\[ \pm \cE_{N,\ell}^{(4)} \leq C \ell^{\alpha/2} (\mathcal{N} + \mathcal{V}_N+1) + C\ell^{-\alpha}.   \]
for all $\alpha>0$ and $\ell\in(0;1)$. 
\end{prop}
\begin{proof}
Using the identity
	\begin{align*}
	[a_x^* a_y^* a_y a_x,  b_u b_v] 
	&= -\Big(\delta(x-u) a_y^* a_v + \delta(x-v) a_y^* a_u + \delta(x-u) \delta(y-v)  \\
	&\hspace{1cm} + \delta(x-v) \delta(u-y)+ \delta(y-u) a_x^* a_v + \delta(y-v) a_x^* a_u\Big)b_x b_y ,
	\end{align*}
we have that
	\begin{align*}
	G_{N,\ell}^{(4)} &=e^{-B({\eta_H})} \mathcal{L}_N^{(4)} e^{B} \\
	&= \mathcal{V}_N + \frac{1}{2}\int_0^1 ds \int dx dy \  N^2 V(N(x-y)) e^{-sB} \big[a_x^* a_y^* a_y a_x, B\big] e^{sB} \\
	&= \mathcal{V}_N
	+ \frac{1}{2} \int_0^1 ds \int dx dy \ N^2 V(N(x-y)) \eta_H(x;y) \left(e^{-sB} b_x b_y e^{sB} + \text{h.c.} \right) \\
	&\quad +   \int_0^1 ds \int dx dy  \ N^2 V(N(x-y))  \left( e^{-sB} a_y^* a(\eta_{H,x}) b_x b_y e^{sB} + \text{h.c.} \right).
	\end{align*}
For the conjugation of the quartic term, we use furthermore that
	\begin{align*}
	 e^{-s B} a_y^* a_u e^{s B}
	&= a_y^* a_u + \int_0^s d\tau \ e^{-\tau B} \big[a_y^* a_u, B\big] e^{\tau B} \\
	&= a_y^* a_u + \int_0^s d\tau e^{-\tau B} \Big(b_y^* b^*(\eta_{H,u}) + b(\eta_{H,y}) b_u \Big) e^{\tau B}
	\end{align*}
so that
	\[
	\begin{split}
	\cG_{N}^{(4)} 
	&= \mathcal{V}_N
	+ \frac{1}{2} \int_0^1 ds \int dx dy \ N^2 V(N(x-y)) \eta_H(x;y) \left(e^{-sB} b_x b_y e^{sB} + \text{h.c.} \right) \\
	&\quad +   \int_0^1 ds \int dx dy \ N^2V(N(x-y))   \left(  a_y^* a(\eta_{H,x}) e^{-sB} b_x b_y e^{sB} + \text{h.c.} \right) \\
	&\quad +   \int_0^1 ds \int_0^s d\tau \int dx dy  \ N^2 V(N(x-y))   \\
	&\hspace{3.2cm} \times \left(  e^{-\tau B} b(\eta_{H,y}) b(\eta_{H,x})  e^{\tau B}  e^{-sB} b_x b_y e^{sB} + \text{h.c.}\right) \\
	&\quad +   \int_0^1 ds \int_0^s d\tau \int dx dy  \ N^2 V(N(x-y))  \\
	&\hspace{3.2cm} \times \left(  e^{-\tau B} b_y^* b^*(\eta^{(2)}_{H,x}) e^{\tau B} e^{-sB} b_x b_y e^{sB} + \text{h.c.} \right) \\
	&=: \mathcal{V}_N + \text{W}_1 +\text{W}_2 + \text{W}_3 + \text{W}_4.
	\end{split}
	\]
Combining as usual the bounds from Lemma \ref{lm:bndsetaH} together with Lemma \ref{lm:auxquartic} and Cauchy-Schwarz, a tedious, but simple analysis as in the proof of \cite[Prop. 7.6]{BBCS4}, shows that 
		\[ \pm \text{W}_2\leq C\ell^{\alpha/2} (\cN+\cV_N+1), \;\pm \text{W}_3\leq C\ell^{\alpha} (\cN+\cV_N+1), \;\pm \text{W}_4 \leq C\ell^{\alpha} (\cN+\cV_N+1). \]
We omit the details and focus on the only relevant term $\text{W}_1$ which can be written as
	\begin{align*}
	\text{W}_1
	&=\frac{1}{2} \int_0^1 ds \int dx dy \ N^2 V(N(x-y)) \eta_H(x;y) \big(b(\cosh_{s \eta_H , x})+b^*(\sinh_{s \eta_H , x}) + d_{s\eta_H,x} \big)\\
	&\hspace{2cm} \times  \big(b(\cosh_{s \eta_H , y})+b^*(\sinh_{s \eta_H , y}) + d_{s\eta_H,y} \big) +\text{h.c.} \\
	&= \frac{1}{2} \int_0^1 ds \int dx dy \ N^2 V(N(x-y)) \eta_H(x;y) \big(b(\cosh_{ s\eta_H , x}) b(\cosh_{s \eta_H , y}) + \text{h.c.}\big) \\
	&\quad + \frac{1}{2} \int_0^1 ds \int dx dy \ N^2 V(N(x-y)) \eta_H(x;y) \big( b(\cosh_{s \eta_H , x}) b^*(\sinh_{s \eta_H , y}) + \text{h.c.} \big) \\
	&\quad + \frac{1}{2} \int_0^1 ds \int dx dy \ N^2 V(N(x-y)) \eta_H(x;y) \big( b(\cosh_{s\eta_H,x}) d_{s\eta_H,y} + \text{h.c.} \big)   + \mathcal{E}_1^{(4)} \\
	&=: \text{W}_{11} + \text{W}_{12} + \text{W}_{13} + \mathcal{E}_1^{(4)},
	\end{align*}
where
	\begin{align*}
	\mathcal{E}_1^{(4)}
	&= \frac{1}{2} \int_0^1 ds \int dx dy \ N^2 V(N(x-y)) \eta_H(x;y)  \big(b^*(\sinh_{s \eta_H , x}) + d_{s\eta_H,x}\big)\\
	&\hspace{2.5cm} \times \big(b(\cosh_{s \eta_H , y})+b^*(\sinh_{s \eta_H , y}) + d_{s\eta_H,y} \big) + \text{h.c.}
	\end{align*}
Using \eqref{eq:bnddx}, \eqref{eq:bndaydxbar} and $N^{-1}  | \eta_H(x;y)  | \leq C|\pn(x)||\pn(y)| $ by \eqref{eq:etapwbnd}, we get
	\[
	\begin{split}
	\big\vert \langle \xi, \mathcal{E}_1^{(4)} \xi \rangle \big\vert
	&\leq C \ell^{\alpha/2}  \int_0^1 ds \int dx dy \ N^2 V(N(x-y)) \vert \eta_H(x;y) \vert \Vert (\mathcal{N}+1)^\frac{1}{2} \xi \Vert \\
	&\hspace{0.5cm} \times \Big[ |\pn(x)|\Vert  \left(b(\cosh_{s \eta_H , y})+b^*(\sinh_{s \eta_H , y}) + d_{s\eta_H,y} \right)  \xi \Vert \\
	&\hspace{1cm}+  N^{-1} \Vert  (\cN+1)^{1/2} a_x  \left(b(\cosh_{s \eta_H , y})+b^*(\sinh_{s \eta_H , y}) + d_{s\eta_H,y} \right) \xi \Vert \Big] \\
	&\leq C \ell^{\alpha/2} \langle \xi, (\mathcal{N} + \mathcal{V}_N+1) \xi \rangle.
	\end{split}
	\]
	
Next, let us analyze the contributions $ \text{W}_{11}, \text{W}_{12}$ and $\text{W}_{13}$. We write
	\[\text{W}_{11} = \frac{1}{2}  \int dx dy \ N^2 V(N(x-y)) \eta_H(x;y) \big(b_x b_y + b_x^* b_y^*\big) + \mathcal{E}_{11}^{(4)}\]
for an error $\mathcal{E}_{11}^{(4)}$ that satisfies 
	\[\begin{split}
	\vert \langle \xi, \mathcal{E}_{11}^{(4)} \xi \rangle \vert
	&\leq  C \int_0^1 ds \int dx dy \ N^2 V(N(x-y)) \vert \eta_H(x;y) \vert \Vert (\mathcal{N}+1)^\frac{1}{2} \xi \Vert  \\
	&\hspace{2cm} \times \Vert (\cN+1)^{-1/2} (b_x b(p_{s\eta_H,y}) + b(p_{s\eta_H,x})b_y + b(p_{s\eta_H,x})b(p_{s\eta_H,y})) \xi \Vert \\
	&\leq C\ell^{\alpha/2} \langle \xi, (\mathcal{N}+1) \xi \rangle.
	\end{split}\]
Similarly, we have that
	\[
	\text{W}_{12} = \frac{1}{2} \int dx dy \ N^2 V(N(x-y)) \eta_H(x;y)^2 (1-\cN/N) + \mathcal{E}_{12}^{(4)},
	\]
where 
	\begin{align*}
	\mathcal{E}_{12}^{(4)} 
	&=  \int_0^1 ds \int dx dy \ N^2 V(N(x-y)) \eta_H(x;y)(1-\cN/N) \\
	&\hspace{2cm} \times [a^*(\sinh_{s \eta_H , y}) a(\cosh_{s \eta_H , x}) + \text{p}_{s \eta_H}(x;y) +\langle \text{r}_{s\eta_H,x}, \sinh_{s\eta_H,y}\rangle+ \text{h.c.}]
	\end{align*}
	and thus $\pm \mathcal{E}_{12}^{(4)} \leq C\ell^{\alpha/2} (\mathcal{N}+1)$. 
	
Finally, we have that
	\begin{equation}
	\text{W}_{13} = - \frac{1}{2} \int dx dy \ N^2 V(N(x-y)) \eta_H(x;y)^2 \left(1-\frac{\mathcal{N}}{N}\right)\frac{\mathcal{N}+1}{N} + \mathcal{E}_{13}^{(4)}
	\end{equation}
where, by the bound \eqref{eq:bndaydxbar}, it is simple to see that
	\[\begin{split}
	\vert \langle \xi, \mathcal{E}_{13}^{(4)} \xi \rangle \vert
	&\leq  C \int_0^1 ds \int dx dy \ N^2 V(N(x-y)) \vert \eta_H(x;y) \vert \cdot  \Vert (\mathcal{N}+1)^{1/2} \xi \Vert\\
	&\hspace{2cm}\times\Big[ \Vert (\mathcal{N}+1)^{- 1/2} a^*(s\eta_{H,y})a_x \xi\| + \ell^{\alpha/2}|\pn(x)| \|  d_{s\eta_H,y}\xi\| \\
	&\hspace{2.5cm}+\|(\cN+1)^{-1/2} a_x \overline{d}_{s\eta_H,y} \xi \Vert  \Big] \\
	&\leq C \ell^{\alpha/2} \langle \xi, (\mathcal{N}+\cV_N+1) \xi \rangle.
	\end{split}\]

In summary, the analysis from above proves that 
		\[\begin{split}
		 \cG_{N}^{(4)} = &\,\mathcal{V}_N + \frac{1}{2} \int dx dy \ N^2 V(N(x-y)) \eta_H(x;y) \big(b_x b_y + b_x^* b_y^*\big)\\
	& + \frac{1}{2} \int dx dy \ N^2 V(N(x-y)) \eta_H(x;y)^2  (1-\cN/N) (1-\cN/N-1/N) + \widetilde{\cE}_{N,\ell}^{(4)} ,
		 \end{split}\]
for an error $\widetilde{\cE}_{N,\ell}^{(4)}$ satisfies
		\[\pm \widetilde{\cE}_{N,\ell}^{(4)}\leq C \ell^{\alpha/2}  (\mathcal{N}+\cV_N+1).\]
Replacing finally $ \eta_H(x;y)$ by $ G(x-y)\pn(x)\pn(y) = -Nw_\ell(N(x-y))\pn(x)\pn(y)$ in the first two contributions on the right hand side of the last equation for $\cG_{N}^{(4)} $, we conclude the proposition, using that $ \| G\ast \widecheck{\chi}_{H^c}\|_\infty\leq C \ell^{-\alpha}$ and $ N^{-1}|\eta_H(x;y)|\leq C$, by \eqref{eq:etapwbnd}. 
\end{proof}

\subsection{Proof of Proposition \ref{prop:GNell}}\label{sec:proofpropGNell}

Collecting the results from the previous subsections, we are now ready to prove Proposition \ref{prop:GNell}. Since the proof is similar to the proof of \cite[Theorem 4.4]{BS} and \cite[Prop. 4.2]{BBCS4}, we explain the main steps only. 
\begin{proof}[Proof of Proposition \ref{prop:GNell}]
Let us collect the results of Propositions \ref{prop:GNell0}, \ref{prop:GNell1}, \ref{prop:Vext}, \ref{prop:GNell2V}, \ref{prop:e-BKeB}, \ref{prop:GNell3} and \ref{prop:GNell4}, noting that there is a cancellation between the linear main contributions from $\cG_N^{(1)}$ in Prop. \ref{prop:GNell1} with those of $\cG_N^{(3)}$ in Prop. \ref{prop:GNell3}. We find that
		\begin{equation} \label{eq:ovGN}
		\begin{split}
		\cG_{N} = &\, \big\langle \pn, \big[ -\Delta + V_{ext} + \frac{1}{2} \big(N^3 V(N\cdot) * \vert \pn \vert^2\big) \big] \pn \big\rangle (N-\cN) \\
		& \qquad -  \frac12\big\langle \pn,  \big(N^3 V(N\cdot) * \vert \pn \vert^2 \big)\pn \big\rangle (\cN+1)(1-\cN/N) \\ 
		& + \int dx dy \ N^3 V(N(x-y)) \pn(x) \pn(y) \eta_H(x;y)  (1-\cN/N)(1-\cN/N-1/N)\\
		&+   \int dxdy\, (-\Delta_x \eta_H (x;y)))\eta_H(x;y)  (1-\cN/N)  (1-\cN/N-1/N)  \\
		& + \frac{N}{2} \int dx dy \ N^3 V(N(x-y))  w_\ell^2(N(x-y)) |\pn(x)|^2|\pn(y)|^2  \\
		&\hspace{1.8cm}\times (1-\cN/N) (1-\cN/N-1/N) \\
		&  + \int dx \big(N^3 V(N.)\ast \vert \pn \vert^2\big)(x) b_x^* b_x   + \int dx dy \ N^3 V(N(x-y)) \pn(x)\pn(y) b_x^*b_y \\
		&  + \frac{1}{2} \int dx dy\big(  N^3 (Vf_\ell)(N(x-y)) \pn(x) \pn(y) - \Delta_y\eta_{H}(x;y)   \big)\big ( b_x^* b_y^*  + b_x b_y  \big)  \\
		&+ \int dxdy\,  N^{ 5/2} V(N(x-y)) \pn(y)  \big( b_x^* a^*_y a_x + \text{h.c.}\big) \\
		&+ \cK+ \cV_{ext}+\cV_N+ \widetilde{\cE}^{(1)}_{\cG_{N} },
		\end{split}
		\end{equation}
where the error $\widetilde{\cE}^{(1)}_{\cG_{N} }$ 
satisfies the estimate
		\[\pm \widetilde{\cE}^{(1)}_{\cG_{N} } \leq C\ell^{(\alpha-3)/2}(\cN+\cK+\cV_N+1) + C\ell^{-5\alpha/2}N^{-1}(\cN+1)^2+ C\ell^{-\alpha}. \]
To prove Proposition \ref{prop:GNell}, we need to simplify $ \cG_{N} - \widetilde{\cE}^{(1)}_{\cG_{N} } $ further. 

We start with the terms on the first six lines of \eqref{eq:ovGN}. Recalling that
		\begin{align*}
		- \Delta_x \eta_H (x;y) &= - \Delta G (x-y) \pn(x) \pn(y) -2\nabla G (x-y) \nabla \pn(x) \pn(y) \\
		&\quad - G(x-y)  \Delta\pn(x) \pn(y) + \Delta_x\big[(G * \widecheck{\chi}_{H^c})(x-y) \pn(x) \pn(y)\big]
		\end{align*}
an application of the scattering equation \eqref{eq:scatlN} together with the bounds \eqref{eq:lambdaell}, \eqref{3.0.scbounds1} from Lemma \ref{3.0.sceqlemma} as well as the pointwise bounds $ \|  G\ast \widecheck{\chi}_{H^c}\|_\infty\leq C\ell^{-\alpha} $ and 
		\[\begin{split}
		 \big|  \nabla_x\big[(G * \widecheck{\chi}_{H^c})(x-y) \pn(x) \pn(y)\big| \leq C\ell^{-2\alpha} |\pn(y)| \big[ |\pn(x)|+ |\nabla\pn(x)|  + |\Delta\pn(x)|  \big], \\
		 \big|  \Delta_x\big[(G * \widecheck{\chi}_{H^c})(x-y) \pn(x) \pn(y)\big| \leq C\ell^{-3\alpha} |\pn(y)| \big[ |\pn(x)|  + |\nabla\pn(x)|  + |\Delta\pn(x)| \big]
		 \end{split}\]  
shows that
		\[\begin{split}
		&\big\langle \pn, \big[ -\Delta + V_{ext} + \frac{1}{2} \big(N^3 V(N\cdot) * \vert \pn \vert^2\big) \big] \pn \big\rangle (N-\cN) \\
		& \qquad -  \frac12\big\langle \pn,  \big(N^3 V(N\cdot) * \vert \pn \vert^2 \big)\pn \big\rangle (\cN+1)(1-\cN/N) \\ 
		& + \int dx dy \ N^3 V(N(x-y)) \pn(x) \pn(y) \eta_H(x;y)  (1-\cN/N)(1-\cN/N-1/N)\\
		&+   \int dxdy\, (-\Delta_x \eta_H (x;y)))\eta_H(x;y)  (1-\cN/N)  (1-\cN/N-1/N)  \\
		& + \frac{N}{2} \int dx dy \ N^3 V(N(x-y))  w_\ell^2(N(x-y)) |\pn(x)|^2|\pn(y)|^2  \\
		&\hspace{1.8cm}\times (1-\cN/N) (1-\cN/N-1/N) \\
		& =  \big\langle \pn, \big[ -\Delta + V_{ext} + \frac{1}{2} \big(N^3 (Vf_\ell)(N\cdot) * \vert \pn \vert^2\big) \big] \pn \big\rangle (N-\cN) \\
		& \qquad -  \frac12\big\langle \pn,  \big(N^3 (Vf_\ell)(N\cdot) * \vert \pn \vert^2 \big)\pn \big\rangle  \,\cN (1-\cN/N) + \widetilde{\cE}^{(2)}_{\cG_{N} },\\ 
		\end{split}\]
for an error that satisfies $\pm \widetilde{\cE}^{(2)}_{\cG_{N} }\leq C(\ell^{(\alpha-3)/2}+ \ell^{-3\alpha})$. Using the Gross-Pitaevskii equation \eqref{eq:gpeq}, the bound \eqref{eq:Vfa0} from Lemma \ref{3.0.sceqlemma}, a simple application of the mean value theorem shows furthermore that 
		\[\begin{split} 
		&\big\langle \pn, \big[ -\Delta + V_{ext} + \frac{1}{2} \big(N^3 (Vf_\ell)(N\cdot) * \vert \pn \vert^2\big) \big] \pn \big\rangle (N-\cN) \\
		& \qquad -  \frac12\big\langle \pn,  \big(N^3 (Vf_\ell)(N\cdot) * \vert \pn \vert^2 \big)\pn \big\rangle  \,\cN (1-\cN/N) \\
		& = N\cE_{GP}(\pn) - \eps_{GP}\cN + 4\pi \mathfrak{a}_0\|\pn\|_4^4\, \cN^2/N + \widetilde{\cE}^{(3)}_{\cG_{N} },
		\end{split} \]
up to an error that satisfies $\pm \widetilde{\cE}^{(3)}_{\cG_{N} }\leq C\ell^{-1} $. This shows that 
		\[\begin{split}
		\cG_{N} =&\,  N\cE_{GP}(\pn) - \eps_{GP}\cN + 4\pi \mathfrak{a}_0\|\pn\|_4^4\, \cN^2/N \\
		&  + \int dx \big(N^3 V(N.)\ast \vert \pn \vert^2\big)(x) b_x^* b_x   + \int dx dy \ N^3 V(N(x-y)) \pn(x)\pn(y) b_x^*b_y \\
		&  + \frac{1}{2} \int dx dy\big(  N^3 (Vf_\ell)(N(x-y)) \pn(x) \pn(y) - \Delta_y\eta_{H}(x;y)   \big)\big ( b_x^* b_y^*  + b_x b_y  \big)  \\
		&+ \int dxdy\,  N^{ 5/2} V(N(x-y)) \pn(y)  \big( b_x^* a^*_y a_x + \text{h.c.}\big) \\
		&+ \cK+ \cV_{ext}+\cV_N +  \widetilde{\cE}^{(1)}_{\cG_{N} }  + \widetilde{\cE}^{(2)}_{\cG_{N} } + \widetilde{\cE}^{(3)}_{\cG_{N} }.
		\end{split}\]
		
Now, let us simplify the quadratic contributions on the right hand side of the last equation. First of all, another application of the mean value theorem shows that 
		\[ \begin{split}
		&\pm\bigg( \int dx \big(N^3 V(N.)\ast \vert \pn \vert^2\big)(x) b_x^* b_x -\widehat V(0)\int |\pn(x)|^2b^*_xb_x \\
		 & \hspace{1cm}+ \int dx dy \ N^3 V(N(x-y)) \pn(x)\pn(y) b_x^*b_y  - \widehat V(0)\int |\pn(x)|^2b^*_xb_x\bigg) \\
		 &\hspace{0.5cm}\leq CN^{-1}(\cN+\cK+1).
		\end{split}\]
Here we used for the second line that $\vert \hat{V}(p/N) - \hat{V}(0) \vert \leq C \vert p \vert /N$ and that
\begin{align*}
&\int dx dy \ N^3 V(N(x-y)) \pn(x) \pn(y) b_y^* b_x \\
&\vspace{2cm}= \int dx \ \hat{V}(0) \pn(x)^2 b_x^* b_x 
+ \int dp \left( \hat{V}(p/N) - \hat{V}(0) \right) \widehat{a}^*(\widehat{\pn}_p) \left(1- \frac{\cN}{N}\right) \widehat{a}(\widehat{\pn}_p),
\end{align*}
where $\widehat{a}(\widehat{\pn}_p) = \int dy\, e^{2\pi i px }\pn(x) a_x$. The previous bound follows from 
\begin{equation}\label{eq:aph20} \int |p|^2 \|\widehat{a}(\widehat{\pn}_p)\xi \|^2\leq 2 \int |\nabla\pn(x)|^2\|a_x \xi\|^2 + 2 \int dx  \pn^2(x)\|\nabla_x a_x\xi\|^2\leq C \| (\cN+\cK)^{1/2} \xi \|^2 \end{equation} 
because $\| \ph \|_\infty, \| \nabla \ph \|_\infty \leq C$ by \eqref{exponential decay}. 

This controls the diagonal terms. For the non-diagonal term, we use once more the scattering equation \eqref{eq:scatlN}, the simple identity $\wh\chi_\ell(p) = \pi^{-1} \sin(\ell|p|) /|p|^3- \pi^{-1} \ell  \cos(\ell |p|) /|p|^2$ so that $| \wh \chi_\ell(p)|\leq C\ell |p|^{-2}$ and the pointwise bounds similarly as above to deduce that 
		\[\begin{split}
		&\pm\bigg( \frac{1}{2} \int dx dy\big(  N^3 (Vf_\ell)(N(x-y)) \pn(x) \pn(y) - \Delta_y\eta_{H}(x;y)   \big)\big ( b_x^* b_y^*  + b_x b_y  \big) \\
		&\hspace{1cm} - \frac{1}{2} \int dx dy \big(N^3(Vf_\ell)(N.)\ast \widecheck{\chi}_{H^c}\big)(x-y)\pn(x)\pn(y)\big ( b_x^* b_y^*  + b_x b_y  \big)\bigg)\\
		&\leq  C( \ell^{\alpha/2} + \ell^{(3\alpha-4)/2} )(\cN+\cK+1).
		\end{split}\]
		
Finally, switching to Fourier space, we see from Eq. \eqref{eq:Vfa0} that for all $p\in \bR^3$		
	\begin{equation} \label{N3 Vfl - 8pi a0}
	\begin{split}
	\vert  \widecheck{(V f_\ell)}(p/N)  - 8\pi a_0 \vert
	&\leq  \int dx \ (Vf_\ell)(x) \cdot \vert e^{2\pi i  xp/N  } -1 \vert + \left\vert \int dx \ (Vf_\ell)(x) - 8 \pi a_0 \right\vert\\
	&\leq  CN^{-1}  |p | +CN^{-1}\ell^{-1}
	\end{split}
	\end{equation} 
and therefore
	\[\begin{split}
	&\pm\bigg( \frac{1}{2} \int dx dy \ \left( (N^3 (Vf_l)(N.) \ast \widecheck{\chi}_{H^c}\right)(x-y) \pn(x) \pn(y) \big(b_x b_y + b_x^* b_y^*\big) \\
	& \hspace{1.5cm} -  4\pi a_0 \int dx dy \ \widecheck{\chi}_{H^c}(x-y) \pn(x) \pn(y) \big(b_x b_y + b_x^* b_y^*\big) \bigg) \\
	&\leq  C \big\| \left( (N^3 (Vf_l)(N.) \ast \widecheck{\chi}_{H^c}\right)-8\pi a_0\widecheck{\chi}_{H^c} \big\|_\infty (\cN+1) \\
	&\leq CN^{-1}\int_{|p|\leq \ell^{-\alpha}} dp\, \big(  |p | + \ell^{-1}\big) (\cN+1) \leq C \big(\ell^{-4\alpha} + \ell^{-3\alpha-1}\big).
	\end{split}\]
Hence, in summary, we conclude that
		\[\begin{split}
		\cG_{N} =&\,  N\cE_{GP}(\pn) - \eps_{GP}\cN + 4\pi \mathfrak{a}_0\|\pn\|_4^4\, \cN^2/N \\
		& + \int dx\, a^*_x \big[-\Delta_x + V_{ext}(x)\big] a_x + 2\widehat V(0)\int |\pn(x)|^2b^*_xb_x\\
		&  + 4\pi a_0 \int dx dy \ \widecheck{\chi}_{H^c}(x-y) \pn(x) \pn(y) \big(b_x b_y + b_x^* b_y^*\big)\\
		&+ \int dxdy\,  N^{ 5/2} V(N(x-y)) \pn(y)  \big( b_x^* a^*_y a_x + \text{h.c.}\big) \\
		& +\frac12 \int dxdy\, N^2V(N(x-y))a^*_xa^*_ya_xa_y+   \widetilde{\cE}^{(1)}_{\cG_{N} }  + \widetilde{\cE}^{(2)}_{\cG_{N} } + \widetilde{\cE}^{(3)}_{\cG_{N} }+\widetilde{\cE}^{(4)}_{\cG_{N} },
		\end{split}\]
where the error $ \cE_{\cG_{N}}:=  \widetilde{\cE}^{(1)}_{\cG_{N} }  + \widetilde{\cE}^{(2)}_{\cG_{N} } + \widetilde{\cE}^{(3)}_{\cG_{N} }+\widetilde{\cE}^{(4)}_{\cG_{N} }$ satisfies
		\[\begin{split}
		\pm \cE_{\cG_{N}} \leq&\; C\big( \ell^{(\alpha-3)/2)} + \ell^{(3\alpha-4)/2}\big) (\cN+\cK+\cV_N+1) + C\ell^{-5\alpha/2}N^{-1}(\cN+1)^2 \\ 
		&+C(\ell^{(\alpha-3)/2}+ \ell^{-4\alpha} + \ell^{-3\alpha-1}+ \ell^{-1} \big).  
		\end{split}\]
Choosing $\alpha>3$, this concludes the proof of Proposition \ref{prop:GNell}. 
\end{proof}

\section{Analysis of $\cJ_N$} 

The goal of this section is to show Prop. \ref{prop:JNell} for the excitation Hamiltonian $\cJ_N = e^{-A} \cG_{N}^\text{eff} e^A$, where $\cG_{N}^\text{eff}$ has been introduced in (\ref{eq:propGNell}) and can be decomposed as 
\begin{equation}\label{eq:def-Geff0} \cG_{N}^\text{eff} = \mathcal{D}_N + \mathcal{Q}_N + \mathcal{C}_N + \cH_N \end{equation} 
with $\cH_N = \cK + \cV_{ext} + \cV_N$ and where 
\begin{equation}\label{eq:def-Geff}
\begin{split} 
\cD_N = \; &N \cE_\text{GP} (\ph_0) - \eps_\text{GP} \cN + 4\pi \frak{a}_0 \| \ph_0 \|_4^4 \cN^2 / N, \\
\mathcal{Q}_N = \; &2\hat{V} (0) \int dx |\ph_0 (x)|^2 b_x^* b_x + 4\pi \frak{a}_0 \int dx dy \check{\chi}_{H^c} (x-y) \ph_0 (x) \ph_0 (y) (b_x^* b_y^* + \text{h.c.}), \\
\mathcal{C}_N = \; &\frac{1}{\sqrt{N}} \int dx dy \, N^3 V (N(x-y)) \ph_0 (x) [ b_x^* a_y^* a_x + \text{h.c.} ].
\end{split} \end{equation} 

In the next subsections, we will study the action of the unitary operator $e^A$ on these terms, where we recall from \eqref{definition A with l} that 
\[
A = \frac{1}{\sqrt{N}} \int dx dy dz\, \nu_H(x;y) \cgl(x-z) \big( b_x^* a_y^* a_z -\text{ h.c.}\big)
\]
with $\nu_H$ and $\gl$ as defined in (\ref{eq:defnuH}) and (\ref{definition chi Ll}), respectively, with parameters $\alpha > \beta > 0$. 

First, however, we need to establish some a-priori bounds controlling the growth of the expectation of kinetic and potential energies, generated by $A$.

\subsection{Preliminary estimates} 

First of all, with the next lemma we control the growth of the expectation of the external potential.
\begin{lemma} 
	Assume \eqref{eq:asmptsVVext} and let $\alpha > \beta >0$. Then there exists $C >0$ such that for all $\xi \in \cF_+^{\leq N}$, $t \in [0;1]$, $\ell\in (0;1)$ and $N\in \mathbb{N}$ large enough,
\begin{equation} \label{rough bound e-A Vext eA}
	\Big| \langle \xi, e^{-tA} \mathcal{V}_{ext} e^{tA} \xi \rangle - \langle \xi, \cV_{ext} \xi \rangle \Big| 
	\leq C \ell^{\alpha/2} \big\langle \xi, \big[ \mathcal{V}_{ext} + \cN + 1 \big] \xi \big\rangle.
	\end{equation}
\end{lemma} 
	\begin{proof}
We compute 
\begin{align*}
[\mathcal{V}_{ext}, A] 
= \frac{1}{\sqrt{N}} \int dx dy dz \ [V_{ext}(x) + V_{ext}(y) - V_{ext}(z)] \nu_H(x;y) \cgl(x-z) b_x^* a_y^* a_z + h.c.
\end{align*}
Using \eqref{Lp norms of widecheck chi L,l} and the fact that, by \eqref{eq:Vextbnd}, $\Vert V_{ext} \pn \Vert_{\infty} \leq C$ we get
		\begin{equation}
		\begin{split}
		&\left\vert \frac{1}{\sqrt{N}} \int dx dy dz \ V_{ext}(y) \nu_H(x;y) \cgl(x-z) \langle \xi, b_x^* a_y^* a_z \xi \rangle \right\vert \\
		&\leq \frac{C}{\sqrt{N}} \int dx dy dz \ \vert (G * \widecheck{\chi}_H)(x-y) \vert \cgl(x-z) \Vert a_x a_y \xi \Vert  \Vert a_z \xi \Vert \leq C \ell^\frac{\alpha}{2} \Vert (\mathcal{N}+1)^\frac{1}{2} \xi \Vert^2.
		\end{split}
		\end{equation}
Furthermore, recalling the assumption $V_{ext}(x+y) \leq C (V_{ext}(x) + C) (V_{ext}(y) +C)$ from \eqref{eq:asmptsVVext}, and using \eqref{eq:nuHL2bnd} as well as \eqref{Lp norms of widecheck chi L,l}, we find that
\begin{align*}
		&\left\vert \frac{1}{\sqrt{N}} \int dx dy dz \ V_{ext}(x) \nu_H(x;y) \cgl(x-z) \langle \xi, b_x^* a_y^* a_z \xi \rangle \right\vert \\
		&\leq \frac{1}{\sqrt{N}} \left( \int dx dy dz \ V_{ext}(x) \cgl(x-z) \vert \nu_H(x;y) \vert^2  \Vert a_z \xi \Vert^2 \right)^\frac{1}{2} \\
		& \hspace{3cm} \times \left( \int dx dy dz \ V_{ext}(x) \cgl(x-z) \Vert a_x a_y \xi \Vert^2 \right)^\frac{1}{2} \\
		&\leq \frac{C \ell^\frac{\alpha}{2}}{\sqrt{N}} \left( \int dx dz \ V_{ext}(x+z) \cgl(x) \Vert a_z \xi \Vert^2 \right)^\frac{1}{2}  \Vert \mathcal{V}_{ext}^\frac{1}{2} (\mathcal{N}+1)^\frac{1}{2} \xi \Vert \\
		&\leq \frac{C \ell^\frac{\alpha}{2}}{\sqrt{N}}  (\Vert \mathcal{V}_{ext}^\frac{1}{2} \xi \Vert + \Vert (\mathcal{N}+1)^\frac{1}{2} \xi \Vert) \, \Vert \mathcal{V}_{ext}^\frac{1}{2} (\mathcal{N}+1)^\frac{1}{2} \xi \Vert
		\leq C \ell^\frac{\alpha }{2} \langle \xi , (\mathcal{V}_{ext} + \mathcal{N}+1) \xi \rangle.
		\end{align*}
		Notice that we used in the last step the bound $\| V_{ext} \cgl \|_1 \leq C$, which follows from the fact that $V_{ext}$ grows at most exponentially (see Appendix \ref{apx:gpfunctional}) and the explicit formula $\cgl(x) = (\sqrt{\pi}  \ell^{-\beta} )^3 e^{ -(\pi \ell^{-\beta} x)^2 }$. Similarly, we get
		\begin{align*}
		&\left\vert \frac{1}{\sqrt{N}} \int dx dy dz \ V_{ext}(z) \nu_H(x;y) \cgl(x-z) \langle \xi, b_x^* a_y^* a_z \xi \rangle \right\vert 
		\leq C \ell^\frac{\alpha }{2} \langle \xi, (\mathcal{V}_{ext} + \mathcal{N}+1) \xi \rangle.
		\end{align*}
		Thus, we have
		\[
		\pm [\mathcal{V}_{ext}, A] \leq C \ell^\frac{\alpha}{2} (\mathcal{V}_{ext}+ \mathcal{N}+1) \, .
		\]
With \eqref{rough bound e-A (N+1) eA with l}, Eq. (\ref{rough bound e-A Vext eA}) follows now by Gronwall's lemma, applied to the function $f (t) = \langle \xi, e^{-tA} \cV_{ext} e^{tA} \xi \rangle$. 
\end{proof}

Next, we need to control the growth of the kinetic and potential energy. To estimate contributions arising from the kinetic energy, we will often need to switch to momentum space. We will use the formal  notation $\hat{a}_p = a (e^{2\pi ip \cdot x}) $ and $\hat{a}^*_p = a^* (e^{2\pi ip \cdot x})$ to indicate creation and annihilation operators in momentum space. For $f \in L^2 (\bR^3)$ (interpreted as a function of momentum), we set $\hat{a} (f) = \int \bar{f}(p) \hat{a}_p$ and similarly for $\hat{a}^* (f)$. It is useful to keep in mind that 
\begin{equation}\label{eq:aph} \int dy \, e^{2\pi i p \cdot y} \ph (y) a_y  = \hat{a} (\hat{\ph}_p) \end{equation} 
where $\hat{\ph}_p (q) := \hat{\ph} (p-q)$.  

We will often encounter operators as in (\ref{eq:aph}), with $\ph$ the minimizer (or the square of the minimizer) of the Gross-Pitaevskii energy functional. Switching to position space, we can bound 
\begin{equation}\label{eq:aph1} \int dp \ \| \hat{a} (\hat{\ph}_p)  \xi \|^2  = \int dx  |\ph (x)|^2 \| a_x \xi \|^2 \leq \| \ph \|_\infty^2 \| \cN^{1/2} \xi \|^2 \leq C \| \cN^{1/2} \xi \|^2
\end{equation} 
and, as already discussed in (\ref{eq:aph20}),  
\begin{equation} \label{eq:aph2} 
\int dp \  p^2 \| \hat{a} (\hat{\ph}_p) \xi \|^2 \leq C \| (\cK + \cN)^{1/2} \xi \|^2.
\end{equation}
\begin{lemma} 
Let $\alpha > 4$, $0 < \beta < \alpha$. Then we can write 
\begin{equation}\label{eq:commKA}  [\cK , A] = T_2 + S_1 + S_2 + \delta_{\cK} \end{equation}
where
\begin{align*}
	T_2 &= - \frac{2}{\sqrt{N}} \int dx dy dz \ \nabla_x\nu_H(x;y) \nabla_x \cgl(x-z) [ b_x^* a_y^* a_z + h.c.], \\
	S_1 &=-\frac{1}{\sqrt{N}} \int dx dy dz \ N^3 (Vf_\ell)(N(x-y)) \pn(y) \cgl(x-z) [b_x^* a_y^* a_z + h.c.], \\
	S_2 &= \frac{1}{\sqrt{N}} \int dx dy dz \ [(N^3 (Vf_\ell)(N\cdot)) * \widecheck{\chi}_{H^c}](x-y) \pn(y) \cgl(x-z) [b_x^* a_y^* a_z + h.c.] 
\end{align*}
and where 
\[ | \langle \xi,  \delta_{\cK}  \xi \rangle | \leq C \ell^{(\alpha -4)/2} \langle \xi , \cN \xi \rangle + C \ell^{\alpha/2} \| \cK^{1/2} \xi \| \| (\cN+ 1)^{1/2} \xi \|.  \]
Moreover, we find 
\begin{equation}\label{eq:T2S12} \begin{split}  
\pm T_2 &\leq C \ell^{\alpha/2}  \cK \, , \qquad \pm S_1 \leq C (\cV_N + \cN + 1) \, , \qquad 
\pm S_2 \leq C \cK + C \ell^{-\alpha} (\cN + 1) \, , \end{split} \end{equation}
and thus 
\begin{equation} \label{eq:commKA2} \pm [\cK , A] \leq C (\cK + \cV_N) + C \ell^{-\alpha} (\cN+1). \end{equation} 
\end{lemma} 

\begin{proof} 
With the commutation relations (\ref{eq:ccr}), (\ref{eq:comm-b2}) 
and integration by parts, we obtain 
	\[
		\begin{split}
	[\mathcal{K},A]
	&= \frac{1}{\sqrt{N}} \int dx dy dz \ \nu_H(x;y) \cgl(x-z) [- \Delta_x b_x^* a_y^* a_z - b_x^* \Delta_y a_y^* a_z + b_x^* a_y^* \Delta_z a_z + h.c.] \\
	&= -\frac{1}{\sqrt{N}} \int dx dy dz \ [\Delta_x \nu_H (x;y)+ \Delta_y \nu_H (x;y)] \cgl(x-z) [b_x^* a_y^* a_z + h.c.] \\
	&\quad - \frac{2}{\sqrt{N}} \int dx dy dz \ \nabla_x\nu_H(x;y) \nabla_x \cgl(x-z) [ b_x^* a_y^* a_z + h.c.] \\
	&=: T_1 + T_2.
	\end{split}
	\]
We used here the identity $-\Delta_x \cgl (x-z) + \Delta_z \cgl (x-z) = 0$. We rewrite $T_1$ as 
	\[
	\begin{split}
	T_1&= -\frac{2}{\sqrt{N}} \int dx dy dz \ (\Delta G * \widecheck{\chi}_H)(x-y) \pn(y) \cgl(x-z) [b_x^* a_y^* a_z + h.c.] \\
	&\quad + \frac{1}{\sqrt{N}} \int dx dy dz \ (G* \widecheck{\chi}_H)(x-y) \, \Delta \pn (y) \cgl(x-z) [b_x^* a_y^* a_z + h.c.] \\
	&\quad + \frac{2}{\sqrt{N}} \int dx dy dz \ (G* \widecheck{\chi}_H)(x-y) \, \nabla\pn(y) \cgl(x-z) [b_x^* \nabla_y a_y^* a_z + h.c.] \\
	&=: T_{11} + T_{12} + T_{13}.
	\end{split}
	\]
	Using Young's inequality and \eqref{Lp norms of widecheck chi L,l} we obtain
	\[
	\begin{split}
	\vert \langle \xi, T_{12} \xi \rangle \vert
	&\leq \frac{C}{\sqrt{N}} \Vert G* \widecheck{\chi}_H \Vert \int dx dz \ \cgl(x-z) \Vert a_x (\mathcal{N}+1)^\frac{1}{2} \xi \Vert \Vert a_z \xi \Vert \\
	&\leq C\ell^\frac{\alpha}{2} \Vert (\mathcal{N}+1)^\frac{1}{2} \xi \Vert^2
	\end{split}
	\]
	and
	\[
	\begin{split}
	\vert \langle \xi, T_{13} \xi \rangle \vert
	&\leq \frac{C}{\sqrt{N}} \Vert G* \widecheck{\chi}_H \Vert \int dx dz \ \cgl(x-z) \Vert  \mathcal{K}^\frac{1}{2} a_x \xi \Vert \Vert a_z \xi \Vert \\
	&\leq C\ell^\frac{\alpha}{2} \Vert (\mathcal{N}+1)^\frac{1}{2} \xi \Vert \Vert \mathcal{K}^\frac{1}{2} \xi \Vert.
	\end{split}
	\]
	We are left with $T_{11}$. For this term we use the scattering equation \eqref{eq:scatlN} and get
	\[
	\begin{split}
	T_{11}
	&=-\frac{1}{\sqrt{N}} \int dx dy dz \ N^3 (Vf_\ell)(N(x-y)) \pn(y) \cgl(x-z) [b_x^* a_y^* a_z + h.c.] \\
	&\quad +\frac{1}{\sqrt{N}} \int dx dy dz \ [(N^3 (Vf_\ell)(N\cdot)) * \widecheck{\chi}_{H^c}](x-y) \pn(y) \cgl(x-z) [b_x^* a_y^* a_z + h.c.] \\
	&\quad + \frac{2}{\sqrt{N}} \int dx dy dz \ \left[ N^3 \lambda_\ell (f_\ell(N\cdot) \chi_\ell) * \widecheck{\chi}_H \right](x-y) \pn(y) \cgl(x-z) [b_x^* a_y^* a_z + h.c.] \\
	&=: S_1 + S_2 + S_3 .
	\end{split}
	\]
An explicit calculation shows that $|\hat{\chi}_\ell (p)| = \ell^3 |\hat{\chi}_1 (\ell p)| \leq C \ell |p|^{-2}$. With \eqref{eq:whGbnd}, we find $\| (f_\ell (N.) \chi_\ell) * \check{\chi}_H \| \leq \ell^{1+\alpha/2}$ (for $N$ large enough). From \eqref{eq:lambdaell}, 
\eqref{Lp norms of widecheck chi L,l}, we obtain 
 \[
	\begin{split}
	\vert \langle \xi, S_3 \xi \rangle \vert
	&\leq \frac{C}{\sqrt{N}} N^3 \lambda_\ell \Vert (f_\ell(N\cdot) \chi_\ell) * \widecheck{\chi}_H \Vert \int dx dz \ \cgl(x-z) \Vert a_x (\mathcal{N}+1)^\frac{1}{2} \xi \Vert \Vert a_z \xi \Vert \\
	&\leq C \ell^\frac{(\alpha-4)}{2} \Vert (\mathcal{N}+1)^\frac{1}{2} \xi \Vert^2 .
	\end{split}
	\]
This proves (\ref{eq:commKA}). To show (\ref{eq:T2S12}), we observe that, integrating by parts, 
\begin{align*}
	T_2 &= -\frac{2}{\sqrt{N}} \int dx dy dz \ (G* \widecheck{\chi}_H)(x-y) \nabla \pn(y) \cgl(x-z) [b_x^* a_y^* \nabla_z a_z + h.c.] \\
	&\quad - \frac{2}{\sqrt{N}} \int dx dy dz \ (G* \widecheck{\chi}_H)(x-y) \pn(y) \cgl(x-z) [b_x^* \nabla_y a_y^* \nabla_z a_z + h.c.].
	\end{align*}
	With $\| G * \check{\chi}_H \| = \| \hat{G} \chi_H \| \leq \ell^{\alpha/2}$ and \eqref{Lp norms of widecheck chi L,l} we get 
	\[
	\begin{split}
	\vert \langle \xi, T_2 \xi \rangle \vert 
	&= \frac{C\ell^\frac{\alpha}{2}}{\sqrt{N}} \int dx dz \ \vert \cgl(x-z) \vert \Vert \mathcal{N}^\frac{1}{2} a_x \xi \Vert \Vert \nabla_z a_z \xi \Vert \\
	&\quad + \frac{C\ell^\frac{\alpha}{2}}{\sqrt{N}} \int dx dz \ \vert \cgl(x-z) \vert \Vert \mathcal{K}^\frac{1}{2} a_x \xi \Vert \Vert \nabla_z a_z \xi \Vert \\
	&\leq \frac{C \ell^\frac{\alpha}{2}}{\sqrt{N}} \Vert \cgl \Vert_1 \Vert \mathcal{K}^\frac{1}{2} (\mathcal{N}+1)^\frac{1}{2} \xi \Vert \Vert \mathcal{K}^\frac{1}{2} \xi \Vert  \leq C \ell^{\alpha/2} \| \cK^{1/2} \xi \|^2. 	\end{split}
	\]
By \eqref{Lp norms of widecheck chi L,l}, we have 
	\begin{equation} \label{S1}
	\begin{split}
	\vert \langle \xi, S_1 \xi \rangle \vert
	&\leq \frac{2}{\sqrt{N}} \int dx dy dz \  N^3 V(N(x-y)) \pn(y) \cgl(x-z) \Vert a_x a_y \xi \Vert \Vert a_z \xi \Vert \\
	&\leq C \left( \int dx dy dz \ N^2 V(N(x-y)) \cgl(x-z) \Vert a_x a_y \xi \Vert^2 \right)^\frac{1}{2} \\
	&\qquad \times \left( \int dx dy dz \ N^3 V(N(x-y)) \cgl(x-z) \Vert a_z \xi \Vert^2 \right)^\frac{1}{2} \\
	&\leq C \Vert \mathcal{V}_N^\frac{1}{2} \xi \Vert  \Vert (\mathcal{N}+1)^\frac{1}{2} \xi \Vert.
	\end{split}
	\end{equation}
	For $S_2$ we change to momentum space to get, using (\ref{eq:aph2})
	\[
	\begin{split}
	\vert \langle \xi, S_2 \xi \rangle \vert
	&= \left\vert \frac{1}{\sqrt{N}} \int dp dq \, \widehat{(Vf_\ell )} (p/N) \chi_{H^c}(p) \gl(q) \langle \hat{a} (\hat{\ph}_p) \hat{b}_{q-p} \xi, \hat{a}_q \xi \rangle  \right\vert \\ 
	&\leq \frac{C}{\sqrt{N}} \left[ \int dp  dq \,  p^2 \Vert \hat{a} (\hat{\ph}_p)  \hat{a}_{q-p} \xi \Vert^2  \right]^{1/2} \left[ \int_{|p| \leq \ell^{-\alpha}}  dp dq  \,  |p|^{-2} \| \hat{a}_q \xi \|^2 \right]^{1/2}  \\
	&\leq C  \ell^{-\alpha/2} \| (\cK+ \cN)^{1/2} \xi \| \| (\cN+1)^{1/2} \xi \|.
	\end{split}
	\]
This concludes the proof of (\ref{eq:T2S12}) and thus of (\ref{eq:commKA2}). 	
\end{proof}

To bound the commutator of $A$ with the interaction energy operator $\cV_N$, it is useful to introduce, for any $\theta > 0$, the notation 
\begin{equation}\label{eq:Ktheta} \cK_\theta = \int_{|p| \leq \theta} p^2 \, \hat{a}_p^* \hat{a}_p \end{equation} 
for the kinetic energy of particles having momentum below $\theta$.
\begin{lemma}
Let $\alpha \geq \beta>0$ and $\varepsilon \in (0, \alpha -\beta)$. Then there exists $C>0$ such that for $\ell\in (0;1)$ sufficiently small we have
	\begin{equation} \label{commutator [VN, A]}
	[\mathcal{V}_N, A] = \frac{1}{\sqrt{N}} \int dx dy dz \ \nu_H(x;y) \cgl(x-z) N^2 V(N(x-y)) [b_x^* a_y^* a_z + h.c.] + \delta_V,
	\end{equation}	
	where
	\[
	\pm \delta_V   \leq C \ell^\frac{(\alpha-\beta)}{2}  \big[ \cK_{\ell^{-\beta - \eps}}+ \cV_N + \cN + 1 \big]. \]
Estimating the term on the r.h.s. of (\ref{commutator [VN, A]}), we conclude that
\begin{equation}\label{eq:commVNA2} \pm [\mathcal{V}_N, A] \leq C \big[ \cK_{\ell^{-\beta - \eps}}+ \cV_N + \cN + 1 \big].  \end{equation} 
\end{lemma}
\begin{proof}
With the commutation relations (\ref{eq:ccr}), (\ref{eq:comm-b2}), we find 
 \[
	\begin{split}
	[\mathcal{V}_N, A]
	&= - \frac{1}{\sqrt{N}} \int dx dy dz du \ \nu_H(x;y) \cgl(x-z) N^2 V(N(u-z)) [b_u^* a_x^* a_y^* a_z a_u + h.c.] \\
	&\quad + \frac{1}{\sqrt{N}} \int dx dy dz du \ \nu_H(x;y) \cgl(x-z) N^2 V(N(u-y)) [b_u^* a_y^* a_x^* a_u a_z + h.c. ]\\
	&\quad + \frac{1}{\sqrt{N}} \int dx dy dz du \ \nu_H(x;y) \cgl(x-z) N^2 V(N(x-u)) [b_x^* a_u^* a_y^* a_u a_z + h.c.] \\
	&\quad + \frac{1}{\sqrt{N}} \int dx dy dz \ \nu_H(x;y) \cgl(x-z) N^2 V(N(x-y)) [b_x^* a_y^* a_z + h.c.] \\
	&=: \sum_{j=1}^4 V_j.
	\end{split}
	\]
Switching partially to momentum space and using \eqref{eq:aph1}, \eqref{eq:aph2}, we have that
	\[
	\begin{split}
	\vert \langle \xi, V_1 \xi \rangle \vert
	&\leq  \frac{C}{\sqrt N}\int dpdq du dz\,  N^2 V(N(u-z)) |\widehat{G} (p)| \chi_H (p) \gl (q) \| a_u \hat a_{p+q} \hat a(\widehat{\varphi_p}) \xi\| \| a_za_u\xi\| \\
		&\leq \frac{C\ell^{3\alpha/2}}{\sqrt N}  \bigg(\int dpdq \, p^2  \, \| \hat{a}_{p+q} \hat a(\widehat{\varphi_p}) \xi\|^2  \bigg)^{1/2} \bigg(\int dq\,  |\gl(q)|^2   \bigg)^{1/2} \| \cV_N^{1/2}\xi\|\\
		&\leq C\ell^{3(\alpha-\beta)/2} \langle\xi, (\cK + \cV_N + \cN)\xi\rangle.
	\end{split}
	\]

	We have $\Vert \nu_{H,x} \Vert_1 \leq \Vert G * \widecheck{\chi}_H \Vert \cdot \Vert \pn \Vert \leq C \ell^\frac{\alpha}{2}$. Thus, we get using \eqref{Lp norms of widecheck chi L,l} 
	\[
	\begin{split}
	\vert \langle \xi, V_1 \xi \rangle \vert
	&\leq \frac{C \ell^\frac{\alpha}{2}}{\sqrt{N}} \Vert \cgl \Vert_1 \int dz du \  N^2 V(N(u-z)) \Vert a_u \xi \Vert \cdot \Vert a_z a_u \xi \Vert \\
		&\leq \frac{C \ell^\frac{\alpha}{2}}{N} \Vert \mathcal{V}_N^\frac{1}{2} \xi \Vert \cdot \Vert (\mathcal{N}+1)^\frac{1}{2} \xi \Vert.
	\end{split}
	\]
	
Writing $\nu_H (x,y) = (G* \check{\chi}_H)(x-y) \ph_0 (y)$ and expanding $(G* \check{\chi}_H)$ and also $\cgl$ in Fourier space we can estimate 
\[ \begin{split} 
|\langle \xi , V_2 \xi \rangle | &\leq \frac{1}{\sqrt{N}} \left[ \int d\sigma dp  du dy \, N^2 V(N (u-y)) \, \frac{\gl (p)}{p^2} \, \| a_u a_y \hat{a}_{\sigma+p} \xi \|^2 \right]^{1/2} \\ &\hspace{1cm} \times \left[  
\int d\sigma dp  du dy \,  |\hat{G} (\sigma) \chi_H (\sigma )|^2 \, N^2 V(N (u-y)) \gl (p) \, p^2  \| a_u \hat{a}_p \xi \|^2  \right]^{1/2} \\ &\leq C   \ell^{(\alpha-\beta)/2}  \| \cV_N^{1/2}  \xi \|  \, \| (\cK_{\ell^{-\beta-\eps}}  + \cN + 1)^{1/2} \xi \|, \end{split} \]
where, to bound the parenthesis in the second line, we divided the $p$-integral in the two domains $|p| < \ell^{-\beta -\eps}$ (where we can use the operator $\cK_{\ell^{-\beta-\eps}}$) and $|p| \geq \ell^{-\beta -\eps}$, where we use the estimate $\sup_{|p| > \ell^{-\beta -\eps}} p^2 e^{-p^2 \ell^{2\beta}} < C$. 
The term $V_3$ can be bounded similarly. This proves \eqref{commutator [VN, A]}. To show (\ref{eq:commVNA2}), we use $\vert \nu_H(x;y) \vert \leq C N \pn(y)$, $\eqref{Lp norms of widecheck chi L,l}$ and \eqref{rough bound e-A (N+1) eA with l} to estimate 
	\[
	\begin{split}
	|\langle \xi, V_4 \xi \rangle | &\leq C \int dx dy dz \ N^\frac{5}{2} V(N(x-y)) \cgl(x-z) \pn(y) \Vert a_x a_y \xi \Vert \cdot \Vert a_z \xi \Vert \\
	&\leq C \Vert \cgl \Vert_1^\frac{1}{2} \cdot \Vert \mathcal{V}_N^\frac{1}{2} \xi \Vert \left( \int dx dy dz \ N^3 V(N(x-y)) \cgl(x-z) \Vert a_z \xi \Vert^2  \right)^\frac{1}{2} \\
	&\leq C \Vert \mathcal{V}_N^\frac{1}{2} \xi \Vert \cdot \Vert (\mathcal{N}+1)^\frac{1}{2} \xi \Vert. 
	\end{split}
	\]
\end{proof}

Combining the last three lemmas, we obtain a bound for the growth of the Hamilton operator $\cH_N = \cK + \cV_{ext} + \cV_N$. 
\begin{lemma}
	Assume \eqref{eq:asmptsVVext}. Let $\alpha > \beta >0$ and $\alpha\geq 4$. There exists $C >0$ such that for all $t \in [0;1]$, all $\ell\in (0;1)$ and all $N\in \mathbb{N}$ large enough
	\begin{equation} \label{rough estimate e -sA HN e sA}
	e^{-tA} \mathcal{H}_N e^{tA} \leq C \mathcal{H}_N + C\ell^{-\alpha} (\mathcal{N}+1).
	\end{equation}
\end{lemma}
\begin{proof}
	We define $f_\xi(s) := \langle \xi, e^{-sA} [ \cK + \cV_N ] e^{sA} \xi \rangle$. Then
	\begin{align*}
	\partial_s f_\xi(s) = \langle \xi, e^{-sA}[ \cK + \mathcal{V}_N , A] e^{sA} \xi \rangle.	\end{align*}
Inserting the bounds (\ref{eq:commKA2}), (\ref{eq:commVNA2}) (with $\cK_{\ell^{-\beta-\eps}} \leq \cK$) and applying \eqref{rough bound e-A (N+1) eA with l}, we obtain 
\[ \begin{split}  \partial_s f_\xi(s)  &\leq C f_\xi (s) + C \ell^{-\alpha} \langle \xi , e^{-sA} (\cN+1) e^{sA} \xi \rangle \\ &\leq 
C f_\xi (s) + C \ell^{-\alpha} \langle \xi , e^{-sA} (\cN+1) e^{sA} \xi \rangle \leq 
C f_\xi (s) + C \ell^{-\alpha} \langle \xi , (\cN+1) \xi \rangle \end{split} \]
for all $s \in [0;1]$. With Gronwall's Lemma and using (\ref{rough bound e-A Vext eA}) to estimate the growth of the external potential, we arrive at (\ref{rough estimate e -sA HN e sA}).
\end{proof} 

The estimate (\ref{rough estimate e -sA HN e sA}) is still not optimal (because of the large factor $\ell^{-\alpha}$ in front of $(\cN_+ + 1)$). To improve this bound, we first have to study the growth of the operator $\cK_\theta$, defined in (\ref{eq:Ktheta}) measuring the kinetic energy of particles with momenta smaller than $\theta < \ell^{-\alpha} - \ell^{-\beta}$.

\begin{lemma}\label{lm:Ktheta} 
Assume \eqref{eq:asmptsVVext}. Let $0<\beta < \alpha \leq 4 \beta$, $\alpha \geq 4$, $\varepsilon \in (0;\alpha -\beta)$. Then there exists a constant $C$ such that for all $s\in [0;1]$ all $\ell > 0$ small enough and all $\ell^{-\beta-\eps} \leq \theta < \ell^{-\alpha} - \ell^{-\beta-\varepsilon}$, we have  
	\begin{equation} \label{rough estimate e -sA K theta e sA}
	e^{-sA} \mathcal{K}_\theta e^{sA} \leq C \mathcal{K}_\theta + C\ell^{2 (\alpha - \beta -\varepsilon)} (\mathcal{H}_N  + \cN +1).
	\end{equation}
\end{lemma}
\begin{proof}
In momentum space, we have 
\begin{equation}\label{eq:Amom}
\begin{split} 
	A &= \frac{1}{\sqrt{N}} \int du dv dw \, \widehat{\nu}_H(v;w) \gl(u) [\hat{b}_{u+v}^* \hat{a}_w^* \hat{a}_u - h.c.] \\ &= \frac{1}{\sqrt{N}} \int du dv dw \, \hat{G} (v) \chi_H (v) \hat{\ph}_0 (v+w)  \gl(u) [\hat{b}_{u+v}^* \hat{a}_w^* \hat{a}_u - h.c.] .
	\end{split} \end{equation} 	
Thus, we find 
 \[
	\begin{split}
	[ \mathcal{K}_\theta, A]
	&=  -\frac{1}{\sqrt{N}} \int_{\vert u \vert \leq \theta} du dv dw \ \widehat{\nu}_H(v;w) \gl(u) \vert u \vert^2 [\hat{b}_{u+v}^* \hat{a}_w^* \hat{a}_u + h.c.] \\
	&\quad + \frac{1}{\sqrt{N}} \int_{\vert u + v \vert \leq \theta} du dv dw \ \widehat{\nu}_H(v;w) \gl(u) \vert u + v \vert^2 [\hat{b}_{u+v}^* \hat{a}_w^* \hat{a}_u + h.c.] \\
	&\quad + \frac{1}{\sqrt{N}} \int_{\vert w \vert \leq \theta} du dv dw \ \widehat{\nu}_H(v;w) \gl(u) \vert w \vert^2 [\hat{b}_{u+v}^* \hat{a}_w^* \hat{a}_u + h.c.] \\
	&=: A_1 + A_2 + A_3.
	\end{split}
	\]
Let $\varepsilon\in (0;\alpha -\beta)$. For $\vert u \vert \leq \ell^{-\beta - \varepsilon}$ and $\vert v \vert \geq \ell^{-\alpha}$ we have (since $\alpha > \beta + \eps$) $|u+v| \geq \ell^{-\alpha} /2$, for $\ell$ small enough. Therefore, by \eqref{eq:whnuHL2bnd}, \eqref{eq:decphhat}, we find	
\begin{equation}\label{eq:A1}\begin{split}
	\vert \langle \xi, A_1 \xi \rangle \vert
	&\leq \frac{C}{\sqrt{N}} \int_{\vert u \vert \leq \theta, \ \vert v \vert \geq \ell^{-\alpha}} du dv dw \ \frac{1}{\vert v \vert^2} \vert \widehat{\ph}_0(v+w) \vert \gl (u)  \vert u \vert^2 \Vert \hat{a}_w \hat{a}_{u+v} \xi \Vert  \Vert \hat{a}_u \xi \Vert \\
	&\leq \frac{C}{\sqrt{N}} \ell^{\alpha} \left[ \int du dv dw  |\hat{\ph}_0 (v+w)|  (u+v)^2 \| \hat{a}_w \hat{a}_{u+v}  \xi \|^2 \right]^{1/2}  \\ &\hspace{2cm} \times \left[ \int_{|u| \leq \theta, |v| \geq \ell^{-\alpha}}  du dv dw \frac{1}{|v|^4} |\hat{\ph}_0 (v+w)|  \gl (u) \, u^4 \| \hat{a}_u \xi \|^2 \right]^{1/2} \\ &\leq C \ell^{3\alpha}\| \cK^{1/2} \xi \|  \| (\cN + 1)^{1/2} \xi \| + C \ell^{(3\alpha-2\beta-2\eps)/2} \| \cK^{1/2} \xi \| \| \cK_\theta^{1/2} \xi \|, 
	\end{split} \end{equation} 
where, in the last step, we separated the $u$-integral in the second parenthesis between $|u| > \ell^{-\beta-\eps}$ (where we control $|u|^4$ with $\gl (u)$) and $|u| < \ell^{-\beta-\eps}$ (where we use $\cK_\theta$). 
	
Now we turn to $A_3$. We bound 
\begin{equation}\label{eq:A2}\begin{split}
	\vert \langle \xi, A_3 \xi \rangle \vert
	&\leq \frac{C}{\sqrt{N}} \theta \ell^{2\alpha} \int_{\vert w \vert \leq \theta, \vert v \vert \geq \ell^{-\alpha}} dudv dw \, |\hat{\ph}_0 (v+w)| \frac{\gl (u)}{|u|}   \vert w \vert  \Vert \hat{a}_w \hat{a}_{u+v} \xi \Vert  |u| \Vert \hat{a}_u \xi \Vert \\
	&\leq \frac{C}{\sqrt{N}}  \theta \ell^{2\alpha} \left[ \int_{|w| \leq \theta} du dv dw \,  \frac{\gl (u)}{|u|^2} |w|^2 \Vert \hat{a}_w \hat{a}_{u+v} \xi \Vert^2 \right]^{1/2} \\ &\hspace{3cm} \times   \left[ \int_{\vert w \vert \leq \theta, \vert v \vert \geq \ell^{-\alpha}}  du dv dw \, |\hat{\ph}_0 (v+w)|^2 \gl (u) |u|^2 \| \hat{a}_u \xi \|^2 \right]^{1/2} \\ 
	&\leq C \ell^{3\alpha} \| \cK_\theta^{1/2} \xi \| \| \cN^{1/2} \xi \| + C \theta^{5/2} \ell^{2\alpha +2\beta + 5\eps/2}  \| \cK^{1/2}_\theta \xi \|^2 \\ &\leq C \ell^{3\alpha} \| \cK_\theta^{1/2} \xi \| \| \cN^{1/2} \xi \| + C  \ell^{(4\beta+ 5\eps -\alpha)/2}  \| \cK^{1/2}_\theta \xi \|^2,
	\end{split} \end{equation} 
where we divided the $u$-integral in the third line an the integral over $|u| > \ell^{-\beta-\eps}$ (here we can control factors of $u$ and extract arbitrary polynomial decay in $\ell$ from $\gl$) and an integral over $|u| \leq \ell^{-\beta -\eps}$ (here we used the fact that $|v+w| \geq |v| - |w| \geq \ell^{-\beta}$ and the bound (\ref{eq:decphhat}) for the decay of $\hat{\ph}_0$). 

We are left with $A_2$. Observing that $|u| > |v| - |u+v| \geq \ell^{-\beta-\eps}$ if $|u+v| \leq \theta \leq \ell^{-\alpha} - \ell^{-\beta-\eps}$ and $|v| \geq \ell^{-\alpha}$, we can extract arbitrary polynomial decay in $\ell$ from $\gl$. Thus, we easily get 
\begin{equation} \label{eq:A3}
|\langle \xi, A_2 \xi \rangle | \leq C \ell^{3\alpha} \| (\cN+1)^{1/2} \xi \|^2.
\end{equation}

From (\ref{eq:A1}), (\ref{eq:A2}), (\ref{eq:A3}) and the assumption $\beta < \alpha \leq 4 \beta$, we conclude that
\[ |\langle \xi , [ \cK_\theta, A] \xi \rangle | \leq C \| \cK_\theta^{1/2} \xi \|^2 + C \ell^{3\alpha-2\beta-2\eps} \| (\cK+\cN+1)^{1/2} \xi \|^2.  \]
Defining $f_\xi(s) = \langle \xi, e^{-sA} \mathcal{K}_\theta e^{sA} \xi \rangle$, we conclude that 
\begin{align*}
\vert \partial_s f_\xi (s) \vert 
&\leq C f_\xi (s) + C \ell^{3\alpha-2\beta-2\eps} \| (\cK + \cN+  1)^{1/2} e^{sA} \xi \|^2. 
\end{align*}
Estimating $\cK \leq \cH_N$ and applying (\ref{rough estimate e -sA HN e sA}) and Lemma \ref{lm:Npowcubic}, Gronwall leads us to 
\eqref{rough estimate e -sA K theta e sA}. 
\end{proof}

Lemma \ref{lm:Ktheta} can be used to improve the estimate (\ref{rough estimate e -sA HN e sA}) on the growth of the Hamiltonian. We begin with the potential energy operator $\cV_N$. 
\begin{lemma} 
Assume \eqref{eq:asmptsVVext}. Let $0< \beta < \alpha \leq 4 \beta$, $\alpha \geq 4$. Then there exists $C>0$ such that for $\ell > 0$  small enough and all $s\in [0;1]$, we have
	\begin{equation} \label{rough estimate e -A VN eA with l}
	e^{-sA} \mathcal{V}_N e^{sA} \leq C (\mathcal{H}_N + \cN +1).
	\end{equation}
\end{lemma}
\begin{proof}
Eq. \eqref{rough estimate e -A VN eA with l} follows from (\ref{eq:commVNA2}), using (\ref{rough bound e-A (N+1) eA with l}), \eqref{rough estimate e -sA K theta e sA} and Gronwall's lemma.
\end{proof}  

We conclude this subsection with an improvement on the growth of the kinetic energy operator $\cK$. 
\begin{lemma}
	Assume \eqref{eq:asmptsVVext}. Let $0< \beta < \alpha \leq 4 \beta$, $\alpha \geq 4$, $0 < \eps < \alpha-\beta$. Then there exists $C>0$ such that for $\ell\in (0;1)$ small enough, we have for all $s\in [0;1]$ 
	\begin{equation} \label{rough estimate e -sA K e sA with l}
	e^{-sA} \mathcal{K} e^{sA} \leq C \ell^{-\frac{(\alpha+\beta + \eps)}{2}} (\mathcal{H}_N  + \cN +1).
	\end{equation}
\end{lemma}
\begin{proof}
From (\ref{eq:commKA}) and (\ref{S1}), we find that $[\cK, A] = S_2 + \delta'_\cK$, with \[ \pm \delta'_\cK \leq C \ell^{\alpha/2} \cK+ C (\cV_N + \cN + 1) \]  and 
\[ S_2 = \frac{1}{\sqrt{N}} \int dx dy dz \ [(N^3 (Vf_\ell)(N\cdot)) * \widecheck{\chi}_{H^c}](x-y) \pn(y) \cgl(x-z) [b_x^* a_y^* a_z + h.c.]. \]
Switching to Fourier space, we have  
\begin{align*}
\vert \langle \xi, S_2 \xi \rangle \vert
	&\leq \frac{C}{\sqrt{N}} \int_{\vert p \vert \leq \ell^{-\alpha}} dp dq  \, \gl (q)  \| \hat{a} (\hat{\ph}_p)  \hat{a}_{q-p} \xi \|  \Vert \hat{a}_q \xi \Vert \\ 
	&\leq C\ell^{3\alpha} \| \cN^{1/2} \xi \|^2 +  \frac{C}{\sqrt{N}} \int_{\vert p \vert \leq \ell^{-\alpha}, |q| < \ell^{-\beta -\eps}} dp dq  \,  \| \hat{a} (\hat{\ph}_p)  \hat{a}_{q-p} \xi \|  \Vert \hat{a}_q \xi \Vert,
\end{align*} 
where we used $\gl$ to extract decay in $\ell$ for the region $|q| > \ell^{-\beta-\eps}$. To control the last integral, we split the region where $|q-p| < \ell^{-\alpha} + \ell^{-\beta-\eps}$ in two parts, the first with $|q-p| < \ell^{-\alpha} - \ell^{-\beta-\eps} =: \theta$ and the second with $\theta \leq |q-p| <  \ell^{-\alpha} + \ell^{-\beta-\eps}$. We obtain 
\[ \begin{split}  \vert \langle \xi, S_2 \xi \rangle \vert \leq \; &C\ell^{3\alpha} \| \cN^{1/2} \xi \|^2  \\ &+\frac{C}{\sqrt{N}} \left[ \int_{|q-p| \leq \theta} \hspace{-.4cm}  dp dq (q-p)^2  \| \hat{a} (\hat{\ph}_p)  \hat{a}_{q-p} \xi \|^2 \right]^{1/2} \left[ \int_{|q-p| \leq \theta}  \hspace{-.4cm}  dp dq (q-p)^{-2} \| \hat{a}_q \xi \|^2 \right]^{1/2}  \\ &+\frac{C}{\sqrt{N}} \left[ \int dp dq (q-p)^2  \| \hat{a} (\hat{\ph}_p)  \hat{a}_{q-p} \xi \|^2 \right]^{1/2}  \\ &\hspace{4cm}  \times \left[ 
 \int_{\theta \leq |q-p| \leq \ell^{-\alpha} + \ell^{-\beta-\eps}}  \hspace{-.4cm}  dp dq (q-p)^{-2} \| \hat{a}_q \xi \|^2 \right]^{1/2} \\
 \leq \; &  C\ell^{3\alpha} \| \cN^{1/2} \xi \|^2 + C \ell^{-\alpha/2} \| \cK_\theta^{1/2} \xi \| \| \cN^{1/2} \xi \| + C \ell^{-(\beta+\eps)/2} \| \cK^{1/2} \xi \| \| \cN^{1/2} \xi \| .\end{split} \] 
We conclude that 
\[ \begin{split} \Big| \frac{d}{ds} \langle \xi, e^{sA} \cK e^{-sA} \xi \rangle \Big| 
\leq \; &C \ell^{\alpha/2}  \langle \xi e^{-sA} \cK e^{sA} \xi \rangle + C \ell^{-\alpha/2} \langle \xi, e^{-sA} (\cK_\theta + \cV_N + \cN + 1) e^{sA} \xi \rangle \\ &+ C \ell^{-(\beta + \eps)/2} \| \cK^{1/2} e^{sA} \xi \| \| (\cN + 1)^{1/2} e^{sA} \xi \|. \end{split} \]
With (\ref{rough bound e-A (N+1) eA with l}), \eqref{rough estimate e -sA K theta e sA}, \eqref{rough estimate e -A VN eA with l}, and with (\ref{rough estimate e -sA HN e sA}) (estimating $\cK \leq \cH_N$), we obtain 
\[  \Big| \frac{d}{ds} \langle \xi, e^{sA} \cK e^{-sA} \xi \rangle \Big| \leq C \ell^{-(\alpha+ \beta +\eps)/2} \langle \xi , (\cH_N + \cN + 1) \xi \rangle. \]	
Integrating over $s$ yields \eqref{rough estimate e -sA K e sA with l}.
\end{proof}

\subsection{Analysis of $e^{-A} \mathcal{D}_N e^A$}

In this section we study the contribution arising from the operator $\mathcal{D}_N$, defined in \eqref{eq:def-Geff}. 
\begin{lemma}\label{lm:F-conj}
There is $C>0$ s.t. for all $F \in L^\infty (\bR^3)$, $\alpha , \beta>0$, $\ell\in (0;1)$, we have
	\begin{equation} \label{int F(u) e -A au* au e A}
	\begin{split} 
	\left\vert \int du \ F(u) \langle \xi_1, \left( e^{-A} a_u^* a_u e^A - a_u^* a_u \right) \xi_2 \rangle \right\vert
	&\leq C \ell^\frac{\alpha}{2} \Vert F \Vert_\infty \Vert (\mathcal{N}+1)^\frac{1}{2} \xi_1 \Vert  \Vert (\mathcal{N}+1)^\frac{1}{2} \xi_2 \Vert, \\
	\left\vert \int du \ F(u) \langle \xi_1, \left( e^{-A} b_u^* b_u e^A - b_u^* b_u \right) \xi_2 \rangle \right\vert
	&\leq C \ell^\frac{\alpha}{2} \Vert F \Vert_\infty \Vert (\mathcal{N}+1)^\frac{1}{2} \xi_1 \Vert \Vert (\mathcal{N}+1)^\frac{1}{2} \xi_2 \Vert.
	\end{split} 
	\end{equation}
\end{lemma}
\begin{proof}
With (\ref{eq:comm-b2}), we find 
\begin{align*}
	&\int du \ F(u) \left( e^{-A} a_u^* a_u e^A - a_u^* a_u\right) \\
	&= \frac{1}{\sqrt{N}} \int_0^1 ds \int dx dy dz \ (F(x)+F(y)-F(z) ) \nu_H(x;y) \cgl(x-z) \left[ e^{-sA}b_x^* a_y^* a_z e^A + h.c.\right].
	\end{align*}
	Using \eqref{Lp norms of widecheck chi L,l} and \eqref{rough bound e-A (N+1) eA with l} we obtain
	\[ \begin{split}
	&\left\vert \frac{1}{\sqrt{N}} \int_0^1 ds \int dx dy dz \ (F(x)+ F(y)-F(z)) \nu_H(x;y) \cgl(x-z) \langle \xi_1, e^{-sA} b_x^* a_y^* a_z e^{sA} \xi_2 \rangle \right\vert \\
	&\leq \frac{3 \Vert F \Vert_{\infty}}{\sqrt{N}} \int_0^1 ds \int dx dy dz \  \vert \nu_H(x;y) \vert \vert \cgl(x-z) \vert  \Vert a_x a_y e^{sA} \xi_1 \Vert  \Vert a_z e^{sA} \xi_2 \Vert \\
	&\leq C \Vert F \Vert_\infty \ell^\frac{\alpha}{2} \Vert (\mathcal{N}+1)^\frac{1}{2} \xi_1 \Vert  \Vert (\mathcal{N}+1)^\frac{1}{2} \xi_2 \Vert.
	\end{split}\]
	Interchanging $\xi_1$ and $\xi_2$ yields the same estimate for the hermitian conjugate and implies the first estimate in \eqref{int F(u) e -A au* au e A}. Proceeding similarly, we obtain also the second bound.
	\end{proof} 
	
Using Lemma \ref{lm:F-conj}, we can easily control the action of $e^A$ on the operator $\cD_N$. The proof of the next lemma follows very closely the proof of \cite[Prop. 8.7]{BBCS4}. 
\begin{lemma}
Let $\alpha, \beta>0$. Then there exists $C>0$ such that for all $\ell\in (0;1)$ holds
\begin{equation} \label{conjugation DN}
e^{-A} \mathcal{D}_N e^{A} = \mathcal{D}_N + \delta_{\mathcal{D}_N},
\end{equation}
where
\[
\pm \delta_{\mathcal{D}_N}  \leq C \ell^\frac{\alpha}{2} (\mathcal{N}+1).
\]
\end{lemma}

\subsection{Analysis of $e^{-A} \mathcal{Q}_N e^A$}

In this section we study the contribution to $\cJ_N$ arising from the operator $\mathcal{Q}_N$, defined in \eqref{eq:def-Geff}. 
\begin{lemma}
Let $0< \beta < \alpha \leq 4 \beta$, $\alpha \geq 4$ and $\varepsilon \in (0;\alpha -\beta)$. Then there exists $C>0$ such that for $\ell\in (0;1)$ small enough and $N$ large enough, we have 
\begin{equation} \label{conjugation QN}
e^{-A} \mathcal{Q}_N e^A = \mathcal{Q}_N + \delta_{\mathcal{Q}_N},
\end{equation}
where
\[
\pm \delta_{\mathcal{Q}_N} \leq C \big[ \ell^{(\alpha - \beta-\eps)/2} + \ell^{(5\alpha - 7\beta-2\eps)/4} \big] \, (\cH_N + \cN +  1). \]
\end{lemma}
\begin{proof}
	The first contribution to $\mathcal{Q}_N$ in (\ref{eq:def-Geff}) can be handled with \eqref{int F(u) e -A au* au e A}; it produces an error term bounded by $C \ell^\frac{\alpha}{2} (\mathcal{N}+1)$. As for the second contribution to $\mathcal{Q}_N$, we compute, with the commutation relations (\ref{eq:comm-b}), (\ref{eq:comm-b2}), 
	\[ \begin{split} 
	&4\pi \frak{a}_0 \int du dv \, \widecheck{\chi}_{H^c}(u-v) \pn(u) \pn(v) [b_u^* b_v^*, A] \\
	&=-\frac{8\pi a_0}{\sqrt{N}} \int dx dy dz du \ \widecheck{\chi}_{H^c}(u-z) \pn(u) \pn(z) \nu_H(x;y)  \cgl(x-z) b_x^* b_y^* b_u^* \\
	&\quad + \frac{8\pi a_0}{\sqrt{N}} \int dx dy dz du \ \widecheck{\chi}_{H^c}(u-y) \pn(u) \pn(y) \nu_H(x;y)  \cgl(x-z) b_z^* b_u^* b_x \\
	&\quad +\frac{8 \pi a_0}{\sqrt{N}} \int dx dy dz du \ \widecheck{\chi}_{H^c}(u-x) \pn(u) \pn(x) \nu_H(x;y)  \cgl(x-z) b_z^* a_u^* a_y \left(1-\frac{\mathcal{N}+1}{N} \right) \\
	&\quad + \frac{8\pi a_0}{\sqrt{N}} \int dx dy dz \ \widecheck{\chi}_{H^c}(x-y) \pn(x) \pn(y) \nu_H(x;y)  \cgl(x-z)b_z^* \left( 1- \frac{\mathcal{N}+1}{N} \right) \\
	&\quad - \frac{8\pi a_0}{N \sqrt{N}} \int dx dy dz du dv \ \widecheck{\chi}_{H^c}(u-v) \pn(u) \pn(v) \nu_H(x;y)  \cgl(x-z) b_z^* a_u^* a_v^* a_x a_y \\
	&\quad - \frac{16 \pi a_0}{N \sqrt{N}} \int dx dy dz du \ \widecheck{\chi}_{H^c}(u-y) \pn(u) \pn(y) \nu_H(x;y)  \cgl(x-z) b_z^* a_u^* a_x \\
	&=: \sum_{j=1}^6 \tau_j.
	\end{split} \]
Switching to momentum space we get, with the notation introduced in (\ref{eq:aph}), 
	\begin{equation} \label{tau1}
	\begin{split}
	\vert \langle \xi, \tau_1 \xi \rangle \vert
	&\leq \frac{C}{\sqrt{N}} \int dp ds dt \ \chi_{H^c}(p) \vert \hat{G}(s) \vert \chi_H(s) \gl(t) \vert\hat{\pn}(p+t) \vert \\
	&\qquad \hspace{3cm} \times \Vert (\mathcal{N}+1)^{-\frac{1}{2}}\hat{a} (\hat{\ph}_p) \hat{a} (\hat{\ph}_{-s}) \hat{a}_{s+t} \xi \Vert  \Vert (\mathcal{N}+1)^\frac{1}{2} \xi \Vert \\
	&\leq C \ell^\frac{3(\alpha-\beta)}{2} \Vert (\mathcal{K} + \cN )^\frac{1}{2} \xi \Vert \Vert (\mathcal{N}+1)^\frac{1}{2} \xi \Vert,
	\end{split}
	\end{equation}
where we used the bound (\ref{eq:aph2}), $\| \gl \| \leq C \ell^{-3\beta/2}$ and $\int_{|s| \geq \ell^{-\alpha}} |s|^{-6} \leq C \ell^{3\alpha}$. Similarly, 
\begin{equation} \label{tau2}
	\vert \langle \xi, \tau_2 \xi \rangle \vert , \vert \langle \xi, \tau_3 \xi \rangle \vert
	\leq C \ell^\frac{3(\alpha-\beta)}{2}  \Vert (\mathcal{K} + \cN )^\frac{1}{2} \xi \Vert 
	\Vert (\mathcal{N}+1)^\frac{1}{2} \xi \Vert.
	\end{equation}
Moreover, we can easily bound 
\begin{equation}\label{tau4} \vert \langle \xi, \tau_4 \xi \rangle \vert, \vert \langle \xi, \tau_6 \xi \rangle \vert \leq C \ell^{3\alpha} \| (\cN + 1)^{1/2} \xi \|^2 \end{equation} 
for $N \in \bN$ large enough (here, we can use the small factors $N^{-1/2}$ and $N^{-3/2}$ to gain arbitrary decay in $\ell$). As for $\tau_5$, switching to momentum space we estimate \begin{equation}\label{tau5} \begin{split} 
\vert \langle \xi, \tau_5 \xi \rangle \vert \leq \; &\frac{C}{N^{3/2}} \int dp ds dt \, \chi_{H^c} (p) |\hat{G} (s)| \chi_H (s) \gl (t)  \, \| \hat{a}_t \hat{a} (\hat{\ph}_p) \hat{a} (\hat{\ph}_{-p}) \xi \| \| \hat{a}_{s+t} \hat{a} (\hat{\ph}_s) \xi \| \\ \leq \; & \frac{C}{N} \left[ \int_{|s| \geq \ell^{-\alpha}} dp ds dt \, \frac{p^2}{|s|^6} \| \hat{a}_t \hat{a} (\hat{\ph}_p) \xi \|^2 \right]^{1/2} \left[ \int_{|p| \leq \ell^{-\alpha}} dp ds dt  \, \frac{s^2}{p^2}  \| \hat{a}_{s+t} \hat{a} (\hat{\ph}_s) \xi \|^2 \right]^{1/2} \\ \leq \; &C \ell^{\alpha}  \Vert (\mathcal{K} + \cN )^\frac{1}{2} \xi \Vert^2,
 \end{split} \end{equation} 
where we used $|\hat{G} (s)| \leq C / s^2$ and the bound (\ref{eq:aph2}). Combining \eqref{tau1}, \eqref{tau2}, \eqref{tau4} and \eqref{tau5} we conclude that
\[ \vert \langle \xi, [\mathcal{Q}_N, A] \xi \rangle \vert \leq C \ell^\alpha \| (\cK + \cN )^{1/2} \xi \|^2 + C \ell^{3(\alpha-\beta)/2} \| (\cK+ \cN)^{1/2} \xi \| \| (\cN+1)^{1/2} \xi \| .\]
With \eqref{rough bound e-A (N+1) eA with l},\eqref{rough estimate e -sA K e sA with l}, we obtain 
\[ \begin{split} \Big| \langle \xi, e^{-A} \mathcal{Q}_N e^A \xi \rangle - \langle \xi, \mathcal{Q}_N \xi \rangle \Big| &\leq \int_0^1 |\langle \xi, e^{-sA} [\mathcal{Q}_N , A] e^{sA} \xi \rangle | \, ds \\ &\leq C \big[ \ell^{(\alpha - \beta-\eps)/2} + \ell^{(5\alpha - 7\beta-2\eps)/4} \big] \langle \xi, (\cH_N + \cN +1) \xi \rangle \, . \end{split} \] 
\end{proof}

\subsection{Contributions from $e^{-A} \mathcal{K} e^A$}

To control the action of $e^A$ on the kinetic energy operator $\cK$, we need better estimates on the commutator $[\cK, A]$ and also on a term arising from the second commutator $[[\cK,A],A]$. 
\begin{lemma} 
Let $0< \beta < \alpha \leq 4 \beta$, $\alpha > 4$, and $\varepsilon \in (0;\alpha -\beta)$. Then there exists a constant $C>0$ such that for all $\ell > 0$ sufficiently small and for $N$ sufficiently large we have
	\begin{equation} \label{[K,A]}
	\begin{split}
	[\mathcal{K}, A ]
	&= -\frac{1}{\sqrt{N}} \int dx dy dz \ N^3 (Vf_\ell )(N(x-y)) \pn(y) \cgl(x-z) [b_x^* a_y^* a_z + h.c.] \\
	&\quad + \frac{1}{\sqrt{N}} \int dx dy dz \ 8 \pi a_0 \widecheck{\chi}_{H^c}(x-y) \cgl(x-z) \pn(y) [b_x^* a_y^* a_z +h.c.]  + \wt{\delta}_\mathcal{K}
	\end{split}
	\end{equation}
	where 
	\[
	\begin{split}
	\vert \langle \xi, \wt{\delta}_\mathcal{K} \xi \rangle \vert
	\leq \; &C \ell^\frac{(\alpha -4)}{2} \Vert (\mathcal{N}+1)^\frac{1}{2} \xi \Vert^2 
	+ C \ell^\frac{\alpha}{2}  \| (\cK+ \cN)^{1/2} \xi \| \| (\cK_{\ell^{-\beta-\eps}} + \cN + 1)^{1/2} \xi \| .
	\end{split}
	\]
	Moreover,  
	\begin{equation} \label{double commutator kinetic energy}
	\begin{split}
\Big\vert \frac{1}{\sqrt{N}} \int dx dy dz \ 8 \pi a_0 &\widecheck{\chi}_{H^c}(x-y) \cgl(x-z) \pn(y) \langle \xi, [b_x^* a_y^* a_z, A] \xi \rangle \Big\vert \\
	\leq \; &C \ell^\frac{3(\alpha - \beta)}{2}  \Vert (\mathcal{K} + \cN)^\frac{1}{2} \xi \Vert 
 \Vert (\mathcal{N}+1)^\frac{1}{2} \xi \Vert 	+ C\ell^{\alpha} \Vert (\mathcal{K} + \cN)^\frac{1}{2} \xi \Vert^2 
 \\ &+ C \ell^\frac{(\alpha-\beta)}{2}  \Vert \mathcal{K}_{\ell^{-\beta - \varepsilon}}^\frac{1}{2} \xi \Vert \Vert (\mathcal{N}+1)^\frac{1}{2} \xi \Vert. \end{split}
	\end{equation}
\end{lemma}
\begin{proof}
We start with the decomposition $[\cK, A] = T_2 + S_1 + S_2 + \delta_\cK$ from (\ref{eq:commKA}). 
The term $S_1$ corresponds exactly to the first term on the r.h.s. of (\ref{[K,A]}). To show (\ref{[K,A]}), we prove that $T_2$ is small and that $S_2$ corresponds to the second term on the r.h.s. of (\ref{[K,A]}), up to small corrections. Switching to momentum space, we get 
\[
\begin{split}
T_2 &= - \frac{2}{\sqrt{N}} \int dp dq  \, p \cdot q \,  \hat{G}(p) \chi_H(p)  \gl(q) [\hat{b}_{p+q}^* \hat{a} (\hat{\ph}_{-p}) \hat{a}_{q} + h.c.].
	\end{split}
	\]
We estimate  
\[ \begin{split}
&\vert \langle \xi, T_2 \xi \rangle \vert 
\leq \frac{C}{\sqrt{N}} \int dp dq \ \vert p \vert  \vert q \vert  \vert \hat{G}(p) \vert \chi_H(p) \gl(q) \Vert \hat{a}_{p+q} \hat{a} (\hat{\ph}_{-p}) \xi \Vert  \Vert \hat{a}_q \xi \Vert \\
	&\leq \frac{C}{\sqrt{N}} \left[ \int dp dq \ \vert p \vert^2 \Vert \hat{a}_{p+q} \hat{a} (\hat{\ph}_{-p}) \xi \Vert^2 \right]^\frac{1}{2} 
	\left[ \int dp dq \,  \vert \hat{G}(p) \vert^2 \chi_H(p) \gl(q)^2 \vert q \vert^2 \Vert \hat{a}_q \xi \Vert^2 \right]^\frac{1}{2} \\
&\leq C\ell^\frac{\alpha}{2} \| (\mathcal{K}+\cN)^\frac{1}{2} \xi \|  \left[ \Vert \mathcal{K}_{\ell^{-\beta -\varepsilon}}^\frac{1}{2} \xi \Vert + \ell^{3\alpha} \Vert \mathcal{N}^\frac{1}{2} \xi \Vert \right],
	\end{split} \]
where in the second integral we separated $|q| < \ell^{-\beta-\eps}$ (where we can use $\cK_{\ell^{-\beta-\eps}}$) and $|q| > \ell^{-\beta-\eps}$ (where $\gl$ is effective). Now we consider $S_2$, as defined in (\ref{eq:T2S12}). Switching to Fourier space and subtracting the second term on the r.h.s. of (\ref{[K,A]}), we can bound
\[\begin{split}  \Big| \langle \xi, S_2 \xi &\rangle - \frac{8\pi\frak{a}_0}{\sqrt{N}} \int dx dy dz \, \check{\chi}_{H^c} (x-y) \cgl (x-z) \ph_0 (y) \, \langle \xi, \big[ b_x^* a_y^* a_z + \text{h.c.} \big] \xi \rangle \Big| \\ \leq \; &\frac{1}{\sqrt{N}} \int dp dq \, \int dp dq \, | \widehat{V f_\ell} (q/N) - 8\pi \frak{a}_0 | \chi_{H^c} (q) \gl (p) \| \hat{a}_{p-q} \hat{a} (\hat{\ph}_q) \xi \| \| \hat{a}_p \xi \| \\
\leq \; &\frac{C}{N^{3/2}} \int dp dq\, (\ell^{-1}+ |q|) \chi_{H^c} (q) \gl (p) \| \hat{a}_{p-q} \hat{a} (\hat{\ph}_q) \xi \| \| \hat{a}_p \xi \| \leq C \ell^{3\alpha} \| (\cN+1)^{1/2} \xi \|^2 \end{split} \]
where we used \eqref{N3 Vfl - 8pi a0} and we chose $N$ large enough to obtain the desired decay in $\ell$. This concludes the proof of (\ref{[K,A]}). 

To show (\ref{double commutator kinetic energy}), we write, in momentum space, $A$ as in (\ref{eq:Amom}) and 
\[ \begin{split} \frac{8\pi \frak{a}_0}{\sqrt{N}} \int dx dy dz \, \check{\chi}_{H^c} (x-y) &\cgl (x-z) \ph_0 (y) b_x^* a_y^* a_z \\ & = \frac{8\pi \frak{a}_0}{\sqrt{N}} \int dq dp \, \chi_{H^c} (q) \gl (p) \, \hat{b}^*_{q+p} \hat{a}^* (\hat{\ph}_{-q}) \hat{a}_p. \end{split} \]
Using the relations (\ref{eq:comm-b}), (\ref{eq:comm-b2}) (translated to momentum space), we compute 
\[
\frac{8\pi \frak{a}_0}{\sqrt{N}} \int dx dy dz \, \check{\chi}_{H^c} (x-y) \cgl (x-z) \ph_0 (y) \big[ b_x^* a_y^* a_z , A \big] = \sum_{j=1}^{14} Y_j \]
with the operators $\{ Y_1, \dots , Y_{14} \}$ as given below. We bound each term separately. 

We start with 
\[ Y_1 = -\frac{8\pi \frak{a}_0}{N}  \int dq dp ds \, \chi_{H^c} (q) \gl (p) \hat{G} (s) \chi_H (s) \gl (p+q) \hat{b}^*_{s+p+q} \hat{b}^* (\hat{\ph}_{-s}) \hat{a}^* (\hat{\ph}_{-q}) \hat{a}_p. \]
Using (\ref{eq:aph1}), (\ref{eq:aph2}) and $|\hat{G} (s)| \leq C |s|^{-2}$, we find
\begin{equation*} \begin{split} |\langle \xi, Y_1 \xi \rangle | \leq \; &\frac{C}{N} \left[ \int dq dp ds \, s^2 \| \hat{a}_{s+p+q} \hat{a} (\hat{\ph}_{-s}) \hat{a} (\hat{\ph}_{-q}) \xi \|^2 \right]^{\frac{1}{2}}  \\ &\hspace{5cm} \times \left[ \int_{|s| \geq \ell^{-\alpha}} dq dp ds \, |s|^{-6} \gl (q)^2  \| \hat{a}_p \xi \|^2  \right]^{\frac{1}{2}}  \\ \leq \; &C \ell^{3(\alpha-\beta)/2} \| (\cK + \cN)^{1/2} \xi \| \| \cN^{1/2} \xi \|. \end{split} \end{equation*} 
We continue with 
\[ \begin{split} Y_2 &=  \frac{8\pi \frak{a}_0}{N}  \int dq dp ds dt \, \chi_{H^c} (q) \gl (p) \hat{G} (s) \chi_H (s) \hat{\ph}_0 (s-p-q) \gl (t)\\ &\hspace{7cm} \times  (1-\cN/N) \hat{a}^*_{-t} \hat{a}^* (\hat{\ph}_{-q}) \hat{a}_{t-s} \hat{a}_p \end{split} \]
which can be bounded by
\[ \begin{split} |\langle \xi, Y_2 \xi \rangle | \leq \; & C \ell^{3\alpha} \| (\cN+1)^{1/2} \xi \|^2 \\ &+ \frac{C}{N} \left[ \int_{|s|> \ell^{-\alpha}, |t|< \ell^{-\beta-\eps}}  dq dp ds dt \, |s|^{-4} |\hat{\ph}_0 (s-p-q)| \, t^2 \| \hat{a}_{-t} \hat{a} (\hat{\ph}_{-q}) \xi \|^2 \right]^{1/2} \\ &\hspace{3cm} \times  \left[ \int  dq dp ds dt \, \frac{ \gl(t)^2}{t^{2}} |\hat{\ph}_0 (s-p-q)| \, \| \hat{a}_{t-s} \hat{a}_p \xi \|^2  \right]^{1/2} \\
\leq \; &C \ell^{3\alpha} \| (\cN+1)^{1/2} \xi \|^2 + C \ell^{(\alpha-\beta)/2} \| \cK_{\ell^{-\beta-\eps}}^{1/2} \xi \| \| (\cN+1)^{1/2} \xi \|\end{split} \]
where we used the decay of $\gl$ to handle the region $|t| > \ell^{-\beta-\eps}$. 
Next, we consider
\[ \begin{split} Y_3 =  & \frac{8\pi \frak{a}_0}{N}  \int dp ds dt \, \chi_{H^c} (p+s+t) \cgl (p) \hat{G} (s) \chi_H (s)  \gl (t)  \\ &\hspace{6cm} \times (1-\cN/N) \hat{a}^*_{-t} \hat{a}^* (\hat{\ph}_{p+s+t}) \hat{a} (\hat{\ph}_{s}) \hat{a}_{p} \end{split} \]
whose contribution is estimated similarly as the one of $Y_2$ by
\[ \begin{split} |\langle \xi, Y_3 \xi \rangle | \leq \; & C \ell^{3\alpha} \| (\cN+1)^{1/2} \xi \|^2 \\ &+ \frac{C}{N} \left[ \int_{|s|> \ell^{-\alpha}, |t|< \ell^{-\beta-\eps}}  dp ds dt \, |s|^{-4} \, t^2 \| \hat{a}_{-t} \hat{a} (\hat{\ph}_{p+s+t}) \xi \|^2 \right]^{1/2} \\ &\hspace{3cm} \times  \left[ \int  dp ds dt \, \frac{\gl (t)^2}{t^{2}} \, \| \hat{a} (\hat{\ph}_s)  \hat{a}_p \xi \|^2  \right]^{1/2} \\
\leq \; &C \ell^{3\alpha} \| (\cN+1)^{1/2} \xi \|^2 + C \ell^{(\alpha-\beta)/2} \| \cK_{\ell^{-\beta-\eps}}^{1/2} \xi \| \| (\cN+1)^{1/2} \xi \|.\end{split} \]
The terms 
\[ \begin{split} Y_4 &= \frac{8\pi \frak{a}_0}{N}  \int dq dp ds dt \, \chi_{H^c} (q) \gl (p) \hat{G} (s) \chi_H (s)  \gl (t)  \hat{\ph}_0 (q+s+t) \hat{\ph}_0 (q-p)  \\ &\hspace{10cm} \times (1-\cN/N) \hat{a}^*_{-t} \hat{a}_p, \\
Y_5 &= \frac{8\pi \frak{a}_0}{N}  \int dp ds dt \, \chi_{H^c} (p+s+t) \gl (p) \hat{G} (s) \chi_H (s)  \gl (t)  \widehat{\ph^2_0} (p+t)  (1-\cN/N) \hat{a}^*_{-t} \hat{a}_p, \\
Y_6 &= - \frac{8\pi \frak{a}_0}{N^2}  \int dq dp ds dt \, \chi_{H^c} (q) \gl (p) \hat{G} (s) \chi_H (s)  \gl (t)  \hat{\ph}_0 (q-s-t) \hat{a}^*_{-t}  \hat{a}^*_{q+p}  \hat{a} (\hat{\ph}_{s}) \hat{a}_{p}, \\
Y_7 &= - \frac{8\pi \frak{a}_0}{N^2}  \int dq dp ds dt \, \chi_{H^c} (q) \gl (p) \hat{G} (s) \chi_H (s)  \gl (t)  \widehat{\ph^2_0} (q+s) \hat{a}^*_{-t}  \hat{a}^*_{q+p}   \hat{a}_{-s-t}  \hat{a}_{p} 
\end{split} \]
can be bounded by 
\[ \pm Y_4, \pm Y_5, \pm Y_6, \pm Y_7 \leq C \ell^{3\alpha} (\cN + 1) \]
for $N$ large enough (the additional factor $1/N$ can be used to gain arbitrary decay in $\ell$). As for 
\[ Y_8 = - \frac{8\pi \frak{a}_0}{N^2}  \int dq dp ds dt \, \chi_{H^c} (q) \gl (p) \hat{G} (s) \chi_H (s)  \gl (t) \hat{a}^*_{-t} \hat{a}^* (\hat{\ph}_{-q}) \hat{a}^*_{q+p}   \hat{a}_{-s-t} \hat{a} (\hat{\ph}_{s}) \hat{a}_{p}, \]
we estimate, again with (\ref{eq:aph1}), (\ref{eq:aph2}) and $|\hat{G} (s)| \leq C |s|^{-2}$,  
\[ \begin{split} 
|\langle \xi, Y_8 \xi \rangle | \leq \; &\frac{C}{N^2} \left[ \int_{|s|> \ell^{-\alpha}}  dq dp ds dt \, |s|^{-6} \, q^2 \| \hat{a}_{-t} \hat{a} (\hat{\ph}_{-q}) \hat{a}_{q+p} \xi \|^2 \right]^{1/2} \\ &\hspace{3cm} \times  \left[ \int_{|q| \leq \ell^{-\alpha}}  dq dp ds dt  \, |q|^{-2} \, s^2 \| \hat{a}_{-s-t} \hat{a} (\hat{\ph}_s)  \hat{a}_p \xi \|^2  \right]^{1/2} \\
\leq \; &C \ell^{\alpha} \| (\cK+ \cN)^{1/2} \xi \|^2. \end{split} \]
To bound 
\[ Y_9 =  \frac{8\pi \frak{a}_0}{N}  \int dq dp ds dt \, \chi_{H^c} (q) \gl (p) \hat{G} (s) \chi_H (s)  \gl (t)  \hat{\ph}_0 (p+s) \hat{b}^* (\hat{\ph}_{-q} )  \hat{b}^*_{q+p}  \hat{a}^*_{s+t}  
  \hat{a}_{t}, \]
we first apply Cauchy-Schwarz:  
\[ \begin{split} |\langle \xi, Y_9 \xi \rangle | \leq &\frac{C}{N} \int dq dp ds dt \, \chi_{H^c} (q) \gl (p) | \hat{G} (s)|  \chi_H (s)  \gl (t)  |\hat{\ph}_0 (p+s)| \\ &\hspace{6cm} \times \| \hat{a} (\hat{\ph}_{q+p}) \hat{a}_{s+t} \xi \| \|  \hat{a}^* (\hat{\ph}_{-q}) \hat{a}_t \xi \|. \end{split} \]
With (\ref{eq:ccr}), we find $\| \hat{a}^* (\hat{\ph}_{-q}) \hat{a}_t \xi \|^2 
\leq \| \hat{a} (\hat{\ph}_{-q}) \hat{a}_t \xi \|^2+ \| \hat{a}_t \xi \|^2$. For the contribution from  $\| \hat{a}_t \xi \|^2$, we can extract arbitrary decay in $\ell$ from the factor $N^{-1}$. For $|t| > \ell^{-\beta-\eps}$, we can gain decay in $\ell$ from the Gaussian $\gl$. We find
\[ \begin{split} |\langle \xi, Y_9 \xi \rangle | \leq &C \ell^{3\alpha} \| (\cN+1)^{1/2} \xi \|^2 \\&+ \frac{C}{N} \left[ \int dq dp ds dt  \, \frac{\gl (t)}{t^2} \, |\hat{\ph}_0 (p+s)| \|  \hat{a} (\hat{\ph}_{q+p}) \hat{a}_{s+t} \xi \|^2 \right]^{1/2} \\ &\hspace{1cm} \times \left[ \int_{|s| > \ell^{-\alpha}, |t| < \ell^{-\beta-\eps}}  dq dp ds dt  \, \frac{1}{|s|^4}  |\hat{\ph}_0 (p+s)| t^2 \| \hat{a} (\hat{\ph}_{-q}) \hat{a}_t \xi \|^2 \right]^{1/2} \\ \leq &C \ell^{3\alpha} \| (\cN+1)^{1/2} \xi \|^2 + C \ell^{(\alpha-\beta)/2} \| \cK_{\ell^{-\beta-\eps}}^{1/2} \xi \| \| (\cN + 1)^{1/2} \xi \| .
\end{split} \]
Next, consider 
 \[ Y_{10} = \frac{8\pi \frak{a}_0}{N}  \int dq ds dt \, \chi_{H^c} (q) \gl (s+t) \hat{G} (s) \chi_H (s)  \gl (t)  \hat{b}^*_{q+s+t}   \hat{a}^* (\hat{\ph}_{-s}) \hat{b}^* (\hat{\ph}_{-q} ) \hat{a}_t .\]
 We proceed as for $Y_9$, with $\| \hat{a}^* (\hat{\ph}_{-q} ) \hat{a}_t \xi \|^2 \leq \| \hat{a} (\hat{\ph}_{-q} ) \hat{a}_t \xi \|^2 + \| \hat{a}_t \xi \|^2$. Again, to bound the contribution arising from $\| \hat{a}_t \xi \|^2$ we can extract decay in $\ell$ from $N^{-1}$. We find 
 \[ \begin{split} 
|\langle \xi, Y_{10} \xi \rangle | \leq \; &C \ell^{3\alpha} \| (\cN+1)^{1/2} \xi \|^2 \\ &+ \frac{C}{N} \left[ \int_{|q| \leq \ell^{-\alpha}}  dq ds dt  \, |q|^{-2} \, s^2 \| \hat{a}_{q+s+t} \hat{a} (\hat{\ph}_{-s}) \xi \|^2  \right]^{1/2}\\ &\hspace{3cm} \times 
\left[ \int_{|q| \leq \ell^{-\alpha} , |s|> \ell^{-\alpha}}  dq ds dt \, |s|^{-6} \, q^{2} \|  \hat{a} (\hat{\ph}_{-q})  \hat{a}_{t} \xi \|^2 \right]^{1/2}   \\
\leq \; &C \ell^{3\alpha} \| (\cN+1)^{1/2} \xi \|^2 +  C \ell^{\alpha} \| (\cK+ \cN)^{1/2} \xi \|^2. \end{split} \]
As for 
 \[ Y_{11} = - \frac{8\pi \frak{a}_0}{N}  \int dq dp ds dt \, \chi_{H^c} (q) \gl (p) \hat{G} (s) \chi_H (s)  \gl (t) \hat{\ph}_0 (q+t) \hat{b}^*_{q+p} \hat{a}^*_{s+t}  \hat{a}^* (\hat{\ph}_{-s}) \hat{a}_{p}, \]
we have  
\[ \begin{split} 
|\langle \xi, Y_{11} \xi \rangle | \leq \; &\frac{C}{N} \left[ \int  dq dp ds dt  \, |\hat{\ph}_0 (q+t)| \, s^2 \| \hat{a} (\hat{\ph}_{-s}) \hat{a}_{s+t}  \hat{a}_{q+p} \xi \|^2 \right]^{1/2} \\ &\hspace{3cm} \times  \left[ \int_{|s|> \ell^{-\alpha}}   dq dp ds dt  \, |s|^{-6}   |\hat{\ph}_0 (q+t)| \, \gl (t)^2 \, \| \hat{a}_p \xi \|^2  \right]^{1/2} \\
\leq \; &C \ell^{3(\alpha-\beta)/2} \| (\cK+\cN)^{1/2} \xi \| \| (\cN+1)^{1/2} \xi \|.  \end{split} \]
On the other hand,
\[ Y_{12} =  \frac{8\pi \frak{a}_0}{N}  \int dq dp ds dt \, \chi_{H^c} (q) \gl (p) \hat{G} (s) \chi_H (s)  \gl (t) \widehat{\ph^2_0} (q+s) \hat{b}^*_{q+p} \hat{a}^*_{-t}  \hat{b}_{-s-t}  \hat{a}_{p} \]
can be bounded (controlling as usual the region with $|t| > \ell^{-\beta-\eps}$ with $\chi_{L<\ell}$) by
\[ \begin{split} 
|\langle \xi, Y_{12} \xi \rangle | \leq \; &C \ell^{3\alpha} \| (\cN+1)^{1/2} \xi \|^2 \\ &+ \frac{C}{N} \left[ \int_{|s| > \ell^{-\alpha}, |t|< \ell^{-\beta-\eps}}  dq dp ds dt  \, |s|^{-4} \, |\widehat{\ph^2_0} (q+s) | \, t^2  \|  \hat{a}_{-t}  \hat{a}_{q+p} \xi \|^2 \right]^{1/2} \\ &\hspace{3cm} \times  \left[ \int  dq dp ds dt  \,  |\widehat{\ph^2_0} (q+s) | \, \frac{\gl (t)^2}{t^2} \, \|\hat{a}_{-s-t} \hat{a}_p \xi \|^2  \right]^{1/2} \\
\leq \; &C \ell^{3\alpha} \| (\cN+1)^{1/2} \xi \|^2 + C \ell^{(\alpha-\beta)/2} \| \cK_{\ell^{-\beta-\eps}}^{1/2} \xi \| \| (\cN+1)^{1/2} \xi \| .  \end{split} \]
Also the term 
\[ Y_{13} = \frac{8\pi \frak{a}_0}{N}  \int dq dp ds dt \, \chi_{H^c} (q) \gl (p) \hat{G} (s) \chi_H (s)  \gl (t) \hat{\ph}_0 (q-s-t) \hat{b}^*_{q+p}  \hat{a}^*_{-t} \hat{b} (\hat{\ph}_s) \hat{a}_{p} \]
can be bounded similarly: 
\[ \begin{split} 
|\langle \xi, Y_{13} \xi \rangle | \leq \; &C \ell^{3\alpha} \| (\cN+1)^{1/2} \xi \|^2 \\ &+ \frac{C}{N} \left[ \int_{|s| > \ell^{-\alpha}, |t|< \ell^{-\beta-\eps}}  dq dp ds dt  \, |s|^{-4} \, |\hat{\ph}_0 (q-s-t) | \, t^2  \|  \hat{a}_{-t}  \hat{a}_{q+p} \xi \|^2 \right]^{1/2} \\ &\hspace{3cm} \times  \left[ \int  dq dp ds dt  \,  |\hat{\ph}_0 (q-s-t) | \, \frac{\gl (t)^2}{t^2} \, \|\hat{a} (\hat{\ph}_s) \hat{a}_p \xi \|^2  \right]^{1/2} \\
\leq \; &C \ell^{3\alpha} \| (\cN+1)^{1/2} \xi \|^2 + C \ell^{(\alpha-\beta)/2} \| \cK_{\ell^{-\beta-\eps}}^{1/2} \xi \| \| (\cN+1)^{1/2} \xi \|.   \end{split} \]
Finally, we bound
 \[ Y_{14} = -\frac{8\pi \frak{a}_0}{N}  \int dq ds dt \, \chi_{H^c} (q) \gl (t)^2 \hat{G} (s) \chi_H (s)  \hat{b}^*_{q-t} \hat{a}^* (\hat{\ph}_{-q})   \hat{b}_{-s-t}  \hat{a} (\hat{\ph}_s) \]
 as follows: 
 \[\begin{split} 
|\langle \xi, Y_{14} \xi \rangle | \leq \; &\frac{C}{N} \left[ \int_{|s| > \ell^{-\alpha}}  dq ds dt  \, |s|^{-6} \, q^2  \|  \hat{a}_{q-t}  \hat{a} (\hat{\ph}_{-q}) \xi \|^2 \right]^{1/2} \\ &\hspace{3cm} \times  \left[ \int_{|q| \leq \ell^{-\alpha}}  dq ds dt   \, |q|^{-2}  s^2 \|\hat{a} (\hat{\ph}_s) \hat{a}_{-s-t} \xi \|^2  \right]^{1/2} \\
\leq \; &C \ell^{\alpha} \| (\cK+\cN)^{1/2} \xi \|^2.   \end{split} \]
Combining the estimates for $\{ Y_1, \dots , Y_{14}\}$, we obtain (\ref{double commutator kinetic energy}).
\end{proof}

\subsection{Contributions from $e^{-A} \mathcal{C}_N e^A$}

To control the action of $e^A$ on the cubic term $\mathcal{C}_N$ in (\ref{eq:def-Geff}), we need precise estimates for the commutator $[\mathcal{C}_N, A]$. 
\begin{lemma}
Let $0<\beta < \alpha$ and $\varepsilon\in (0;\alpha-\beta)$. Then there exists a constant $C>0$ such that for all $\ell\in (0;1)$ sufficiently small and for $N$ sufficiently large we have
	\begin{equation}
	\begin{split} \label{[CN, A]}
	[\mathcal{C}_N, A] &= \frac{1}{N} \int dx dy dz \ N^3 V(N(x-y)) \nu_H(x;y) \pn(y) \cgl(x-z) [a_z^* a_y + h.c.] \left(1- \frac{\mathcal{N}}{N}\right) \\
	& + \frac{1}{N} \int dx dy dz \ N^3 V(N(x-y)) \nu_H(x;y) \pn(x) \cgl(x-z) [a_z^* a_x + h.c.] \left(1- \frac{\mathcal{N}}{N}\right) \\
	& + \delta_{C_N},
	\end{split}
	\end{equation}
	where 
	\[
	\begin{split}
	\vert \langle \xi, \delta_{C_N} \xi \rangle \vert 
	&\leq C \ell^\frac{3(\alpha-\beta)}{2}  \Vert (\mathcal{K} + \cN)^\frac{1}{2} \xi \Vert  \Vert \mathcal{N}^\frac{1}{2} \xi \Vert + C \ell^\frac{(\alpha-\beta)}{2} \Vert (\mathcal{K}_{\ell^{-\beta-\varepsilon}} + \cV_N + \cN)^\frac{1}{2} \xi \Vert^2. 	\end{split}
	\]
\end{lemma}
\begin{proof}
A long but straightforward computation using (\ref{eq:comm-b}), (\ref{eq:comm-b2}) shows that 
\[ \begin{split}  [\mathcal{C}_N ,  A] &= \frac{1}{N} \int dx dy dz \ N^3 V(N(x-y)) \nu_H(x;y) \pn(y) \cgl(x-z) [a_z^* a_y + h.c.] \left(1- \frac{\mathcal{N}}{N}\right) \\
	& + \frac{1}{N} \int dx dy dz \ N^3 V(N(x-y)) \nu_H(x;y) \pn(x) \cgl(x-z) [a_z^* a_x + h.c.] \left(1- \frac{\mathcal{N}}{N}\right) \\
 &+ \sum_{j=1}^{12} C_j, \end{split} \]
with the error terms $\{C_1, \dots , C_{12}\}$ are listed and bounded below. 

We begin with
\[ C_1= -\frac{1}{N} \int dx dy dz dv \ N^3 V(N(z-v)) \pn(z) \nu_H(x;y) \cgl(x-z) b_x^* b_y^* a_v^* a_z  \]
which can be bounded, switching to momentum space, by 
\[ \begin{split} 
| \langle \xi, C_1 \xi \rangle | \leq \; &\frac{1}{N} \int dp ds dt \, |\hat{V} (p/N)||\hat{G} (s)| \chi_H (s) \gl (t) \| \hat{a}_{s+t} \hat{a} (\hat{\ph}_{-s}) \hat{a}_{-p} \xi \| \| \hat{a} (\hat{\ph}_{t-p}) \xi \| \\
\leq \; &\frac{C}{N} \left[ \int dp ds dt \,  s^2 \| \hat{a}_{s+t} \hat{a} (\hat{\ph}_{-s}) \hat{a}_{-p} \xi \|^2 \right]^{1/2} \\ &\hspace{3cm} \times  \left[ \int_{|s| > \ell^{-\alpha}}  dp ds dt \, |s|^{-6} \gl (t)^2 \, \| \hat{a} (\hat{\ph}_{t-p}) \xi \|^2 \right]^{1/2} \\ \leq \; & C \ell^{3(\alpha-\beta)/2}  \| (\cK+\cN)^{1/2} \xi \|  \| \cN^{1/2} \xi \|. \end{split}  \]
Also for the term
\[ C_2 =  \frac{1}{N} \int dx dy dz dv \ N^3 V(N(y-v)) \pn(y) \nu_H(x;y) \cgl(x-z) \left(1-\frac{\mathcal{N}}{N}\right) a_z^* a_v^* a_x a_y \]
we switch to momentum space. We find 
\[ \begin{split} 
| \langle \xi, C_2 \xi \rangle | \leq \; &\frac{1}{N} \int dp ds dt \, |\hat{V} (p/N)| |\hat{G} (s)| \chi_H (s) \gl (t) \| \hat{a}_{-t} \hat{a}_{-p} \xi \| \| \hat{a} (\widehat{\ph^2}_{s-p}) \hat{a}_{-s-t} \xi \| \\
\leq \; &C \ell^{3\alpha} \| (\cN+1)^{1/2} \xi \|^2 \\ &+ \frac{C}{N} \left[ \int_{|s| > \ell^{-\alpha}, |t| < \ell^{-\beta-\eps}} dp ds dt \,  |s|^{-4} \, t^2 \| \hat{a}_{-t} \hat{a}_{-p} \xi \|^2 \right]^{1/2} \\ &\hspace{3cm} \times  \left[ \int  dp ds dt \, \frac{\gl (t)^2}{t^2} \, \| \hat{a} (\widehat{\ph^2}_{s-p}) \hat{a}_{-s-t} \xi \|^2 \right]^{1/2} \\ \leq \; & C \ell^{3\alpha} \| \cN^{1/2} \xi \|^2 + C \ell^{(\alpha- \beta)/2} \| \cK_{\ell^{-\beta-\eps}}^{1/2} \xi \| \| \cN^{1/2} \xi \|, \end{split}  \]
where the first term on the r.h.s. arises from $|t| > \ell^{-\beta-\eps}$, where we can use $\gl (t)$ to extract (arbitrary polynomial) decay in $\ell$. The term 
\[ C_3 = \frac{1}{N} \int dx dy dz dv \ N^3 V(N(x-v)) \pn(x) \nu_H(x;y) \cgl(x-z) \left(1- \frac{\mathcal{N}}{N}\right) a_z^* a_v^* a_y a_x \]
can be handled similarly to $C_2$. We find 
\[ | \langle \xi, C_3 \xi \rangle | \leq  C \ell^{3\alpha} \| \cN^{1/2} \xi \|^2 + C \ell^{(\alpha- \beta)/2} \| \cK_{\ell^{-\beta-\eps}}^{1/2} \xi \| \| \cN^{1/2} \xi \| .\]
On the other hand, to estimate 
\[ C_4 = - \frac{1}{N^2} \int dx dy dz du dv \ N^3 V(N(u-v)) \pn(u) \nu_H(x;y) \cgl(x-z) a_u^* a_z^* a_v^* a_x a_y a_u, \]
we switch only partially to Fourier space (keeping $V$ in position space). We obtain  
\[ \begin{split} 
| \langle \xi, C_4 \xi \rangle | \leq \; &\int du dv ds dt \, N V(N(u-v)) |\hat{G} (s)| \chi_H (s) \gl (t) \| \hat{a}_{-t} a_u a_v \xi \| \| \hat{a} (\hat{\ph}_{s}) \hat{a}_{-s-t} a_u \xi \| \\
\leq \; &\frac{C}{N^{3/2}} \left[ \int_{|s| > \ell^{-\alpha}} du dv ds dt \,  N^2 V(N(u-v)) \, |s|^{-4} \, \| \hat{a}_{-t} a_u a_v \xi \|^2 \right]^{1/2} \\ &\hspace{3cm} \times \left[ \int du dv ds dt \, N^3 V(N(u-v)) \| \hat{a} (\hat{\ph}_{s}) \hat{a}_{-s-t} a_u \xi \|^2 \right]^{1/2}  \\
\leq \; &C \ell^{\alpha/2} \| \cV_N^{1/2} \xi \| \| \cN^{1/2} \xi \|.  \end{split}  \]
The terms 
\[ \begin{split} C_5 &=  - \frac{1}{N^2} \int dx dy dz du \ N^3 V(N(u-x)) \pn(u) \nu_H(x;y) \cgl(x-z) a_u^* a_z^* a_y a_u \\
C_6 &=  - \frac{1}{N^2} \int dx dy dz du \ N^3 V(N(u-y)) \pn(u) \nu_H(x;y) \cgl(x-z) a_u^* a_z^* a_x a_u \end{split}  \]
can be bounded using the factor $N^{-2}$ to gain arbitrary decay in $\ell$. We easily find
\[ \pm C_5, \pm C_6 \leq C \ell^{3\alpha} (\cN+ 1). \]
As for 
\[ C_7 = \frac{1}{N} \int dx dy dz dv \, N^3 V(N(y-v)) \pn(y) \nu_H(x;y) \cgl(x-z) b_y^* b_v^* a_x^* a_z, \] 
we stay in position space and estimate 
\[ \begin{split} 
| \langle \xi, C_7 \xi \rangle | \leq \; &\int dx dy dz dv \,  N^2 V(N(y-v)) |\pn(y)| |\nu_H(x;y)| \cgl(x-z) \| a_y a_v a_x \xi \| \| a_z \xi \| \\ \leq  \; &C  \left[ \int dx dy dz dv \,  N^2 V(N(y-v)) \cgl(x-z) \| a_y a_v a_x  \xi \|^2 \right]^{1/2} \\ &\hspace{1.5cm} \times \left[ \int dx dy dz dv \,  N^2 V(N(y-v)) |\nu_H(x;y)|^2 \cgl(x-z)  \| a_z \xi \|^2 \right]^{1/2}  \\ \leq \; &C \ell^{\alpha/2} \| \cV_N^{1/2} \xi \|  \| \cN^{1/2} \xi \|,
\end{split} \]
using that (\ref{Lp norms of widecheck chi L,l}) and (\ref{eq:nuHL2bnd}). We can proceed very similarly to bound
\[ 
C_8 =  \frac{1}{N} \int dx dy dz dv \, N^3 V(N(x-v)) \pn(x) \nu_H(x;y) \cgl(x-z) b_x^* b_v^* a_y^* a_z .\]
We find 
\[ |\langle \xi, C_8 \xi \rangle |  \leq C \ell^{\alpha/2} \| \cV_N^{1/2} \xi \| \| \cN^{1/2} \xi \| \]
Also 
\[ C_9 =  -\frac{1}{N} \int dx dy dz dv \, N^3 V(N(z-v)) \pn(z) \nu_H(x;y) \cgl(x-z) b_z^* a_v^* b_x b_y \]
can be handled analogously, estimating 
\[ \begin{split} 
|\langle \xi, C_9 \xi \rangle | \leq \; &\frac{C}{\sqrt{N}}  \left[ \int dx dy dz dv \, N^2 V(N(z-v)) |\nu_H (x;y)|^2 \cgl(x-z) \| a_z a_v \xi \|^2 \right]^{1/2} \\ &\hspace{2.5cm} \times \left[  \int dx dy dz dv \, N^3 V(N(z-v)) \cgl (x-z) \| a_x a_y \xi \|^2 \right]^{1/2} \\ \leq \; &C \ell^{\alpha/2} \| \cV_N^{1/2} \xi \| \| \cN^{1/2} \xi \|. \end{split} \]
To bound 
\[ C_{10} = - \frac{1}{N} \int dx dy dz du \ N^3 V(N(u-z)) \pn(u) \nu_H(x;y) \cgl(x-z) b_u^* b_x^* a_y^* a_u, \]
we switch partially to momentum space: 
\[ \begin{split} 
| \langle \xi, C_{10} \xi \rangle | \leq \; &\int du dz ds dt \, N^2 V(N(u-z)) |\hat{G} (s)| \chi_H (s) \gl (t) \| a_u \hat{a} (\hat{\ph}_{-s}) \hat{a}_{s+t} \xi \| \| a_u \xi \| \\
\leq \; &\frac{C}{N} \left[ \int du ds dt \,  s^2 \| a_u \hat{a} (\hat{\ph}_{-s}) \hat{a}_{s+t} \xi \|^2 \right]^{\frac{1}{2}} \left[ \int_{|s| > \ell^{-\alpha}} du ds dt \, |s|^{-6}  \gl (t)^2 \, \| a_u \xi \|^2 \right]^{\frac{1}{2}}
\\ \leq \;&C \ell^{3(\alpha-\beta)/2} \| (\cK + \cN)^{1/2} \xi \|  \| \cN^{1/2} \xi \|. \end{split} \]
Also to estimate 
\[ C_{11} =  \frac{1}{N} \int dx dy dz du \ N^3 V(N(u-y)) \pn(u) \nu_H(x;y) \cgl(x-z) b_u^* a_z^* b_x a_u  \]
we switch partially to momentum space. We find
\[ \begin{split} 
| \langle \xi, C_{11} \xi \rangle | \leq \; &\int du dy ds dt \, N^2 V(N(u-y)) |\hat{G} (s)| \chi_H (s) \gl (t) \| a_u \hat{a}_{-t}  \xi \| \| a_u \hat{a}_{-s-t} \xi \| \\
\leq \; &C \ell^{3\alpha} \| \cN^{1/2} \xi \|^2 + \frac{C}{N} \left[ \int_{|t| < \ell^{-\beta-\eps}} du ds dt \,  |s|^{-4} t^2 \| a_u \hat{a}_t \xi \|^2 \right]^{\frac{1}{2}} \\ &\hspace{3.5cm} \times  \left[ \int_{|s| > \ell^{-\alpha}} du ds dt \,  \frac{\gl (t)^2}{t^2} \, \| a_u \hat{a}_{-s-t} \xi \|^2 \right]^{\frac{1}{2}}
\\ \leq \;&C \ell^{3\alpha} \| \cN^{1/2} \xi \|^2 + C \ell^{(\alpha-\beta)/2} \| \cK_{\ell^{-\beta-\eps}}^{1/2} \xi \| \| \cN^{1/2} \xi \|, \end{split} \]
where the first term on the r.h.s. arises from the region $|t| > \ell^{-\beta-\eps}$, where we can use $\gl$ to extract decay in $\ell$. As for the term
\[ C_{12} =  \frac{1}{N} \int dx dy dz du \ N^3 V(N(u-x)) \pn(u) \nu_H(x;y) \cgl(x-z) b_u^* a_z^* b_y a_u, \]
it can be bounded similarly as $C_{11}$. We find
\[ |\langle \xi, C_{12} \xi \rangle | \leq C \ell^{3\alpha} \| \cN^{1/2} \xi \|^2 + C \ell^{(\alpha-\beta)/2} \| \cK_{\ell^{-\beta-\eps}}^{1/2} \xi \| \| \cN^{1/2} \xi \|. \]
\end{proof}

\subsection{Proof of Proposition \ref{prop:JNell}}

Recalling (\ref{eq:def-Geff0}) and applying (\ref{rough bound e-A Vext eA}), (\ref{conjugation DN}) and (\ref{conjugation QN}), we obtain 
\begin{equation}\label{eq:JN1}
\cJ_N = \mathcal{D}_N + \mathcal{Q}_N + \mathcal{C}_N + \mathcal{H}_N + \int_0^1 ds \ e^{-sA} [\mathcal{C}_N + \mathcal{K}+ \mathcal{V}_N, A] e^{sA}
 +\delta_1,
\end{equation} 
where \[ \pm \delta_1 \leq C \left[ \ell^{(\alpha-\beta-\eps)/2}  + \ell^{(5\alpha-7\beta-2\eps)/4} \right] (\cH_N + \cN + 1) \, .\]  

Combining  (\ref{commutator [VN, A]}) with the first term on the r.h.s. of (\ref{[K,A]}) (and recalling the definition of $\mathcal{C}_N$ in (\ref{eq:def-Geff})), we obtain 
\begin{equation} \label{eq:overviewcommutator}
\begin{split}
&[\mathcal{K} + \mathcal{V}_N, A] \\
&= -\mathcal{C}_N + \frac{1}{\sqrt{N}} \int dx dy dz \ N^3 V(N(x-y)) \pn(y) \widecheck{(1-\gl)}(x-z) [b_x^* a_y^* a_z + h.c.] \\
&\quad -\frac{1}{\sqrt{N}} \int dx dy dz \ (G* \widecheck{\chi}_{H^c})(x-y) \pn(y) \cgl(x-z) N^2 V(N(x-y)) [b_x^* a_y^* a_z + h.c.] \\
&\quad + \frac{8 \pi a_0}{\sqrt{N}} \int dx dy dz \  \widecheck{\chi}_{H^c}(x-y) \cgl(x-z) \pn(y) [b_x^* a_y^* a_z +h.c.] + \delta_2,
\end{split}
\end{equation}
where 
\[ \begin{split} 
|\langle \xi, \delta_2 \xi \rangle | \leq \; &C \left[ \ell^{(\alpha-4)/2} +  \ell^{(\alpha-\beta)/2} \right] \, \| (\cK_{\ell^{-\beta-\eps}} + \cN + \cV_N + 1)^{1/2} \xi \|^2 \\ &+ C \ell^{\alpha /2} \, \| \cK^{1/2} \xi \| \| (\cK_{\ell^{-\beta-\eps}} + \cN+1)^{1/2} \xi \|.   \end{split} \]

Next, we observe that the second term on the r.h.s. of (\ref{eq:overviewcommutator}) can be bounded, using that $0 \leq 1- \gl (p) \leq \min(1, \ell^{2\beta} p^2)$, by 
\[
\begin{split}
\Big\vert \frac{1}{\sqrt{N}} \int dx &dy dz \ N^3 V(N(x-y)) \pn(y) \widecheck{(1-\gl)}(x-z) \langle \xi, b_x^* a_y^* a_z \xi \rangle \Big\vert \\
&\leq \Vert \mathcal{V}_N^\frac{1}{2} \xi \Vert \left[ \int dx \ \Vert a ( \widecheck{(1-\gl)}_x \xi  \Vert^2 \right]^\frac{1}{2} \\
&= \Vert \mathcal{V}_N^\frac{1}{2} \xi \Vert \left[ \int dp |1- \gl (p)|^2 \| \hat{a}_p \xi \|^2 \right]^{1/2} \leq C \ell^\beta \Vert \mathcal{V}_N^\frac{1}{2} \xi \Vert \| \cK^{1/2} \xi \|. \end{split} \]
As for the third term on the r.h.s. of (\ref{eq:overviewcommutator}), we can first extract $\| G * \check{\chi}_{H^c} \|_\infty \leq \| G \| \| \chi_{H^c} \| \leq C \ell^{-\alpha}$ and then use the factor $N^{-1/2}$ to gain arbitrary decay in $\ell$. We obtain 
\[  \begin{split}  \Big| \frac{1}{\sqrt{N}} \int dx dy dz \ (G* \widecheck{\chi}_{H^c})(x-y) \pn(y) \cgl(x-z) &N^2 V(N(x-y)) \langle \xi,  [b_x^* a_y^* a_z + h.c.] \xi \rangle \Big| \\ &\leq C \ell^{3\alpha} \| \cV_N^{1/2} \xi \| \| (\cN+1)^{1/2} \xi \|. \end{split} \]
Finally, the fourth term on the r.h.s. of (\ref{eq:overviewcommutator}) can be bounded by 
\[ \begin{split} \Big| &\frac{8\pi\frak{a}_0}{\sqrt{N}} \int dx dy dz \check{\chi}_{H^c} (x-y) \cgl (x-z) \ph_0 (y) \langle \xi, [ b_x^* a_y^* a_z + h.c. ] \xi \rangle \Big| \\ &\leq \frac{C}{\sqrt{N}} \left[\int dx dy dz \, \cgl (x-z) \| a_x a_y \xi \|^2 \right]^{1/2} \left[ \int dx dy dz |\check{\chi}_{H^c} (x-y)|^2 \cgl (x-z) \| a_z \xi \|^2 \right]^{1/2} \\
&\leq \frac{C \ell^{-3\alpha/2}}{\sqrt{N}} \| \cN \xi \| \| \cN^{1/2} \xi \|.  \end{split} \]
Thus
\[ [\mathcal{K} + \mathcal{V}_N, A] = - \mathcal{C}_N + \delta_3 \]
with 
\[ \begin{split} | \langle \xi , \delta_3 \xi \rangle | \leq \; &C \left[ \ell^{(\alpha-4)/2} +  \ell^{(\alpha-\beta)/2} \right] \, \| (\cK_{\ell^{-\beta-\eps}} + \cN + \cV_N + 1)^{1/2} \xi \|^2 \\ &+ C \left[ \ell^{\alpha /2} + \ell^\beta \right] \, \| \cK^{1/2} \xi \| \| (\cK_{\ell^{-\beta-\eps}} +\cV_N +  \cN+1)^{1/2} \xi \| +C \ell^{-4\alpha} \| \cN \xi \|^2 / N. \end{split} \]

Inserting into (\ref{eq:JN1}) and using  (\ref{rough bound e-A (N+1) eA with l}), \eqref{rough estimate e -sA K theta e sA}, (\ref{rough estimate e -A VN eA with l}), (\ref{rough estimate e -sA K e sA with l}), we obtain  
\begin{equation}\label{eq:JN2} \begin{split}  \cJ_N = \; &\mathcal{D}_N + \mathcal{Q}_N + \mathcal{C}_N + \mathcal{H}_N - \int_0^1 ds \ e^{-sA} \mathcal{C}_N e^{sA} + \int_0^1 ds e^{-sA} [\mathcal{C}_N , A] e^{-sA} + \delta_4 \\ 
= \; &\mathcal{D}_N + \mathcal{Q}_N + \mathcal{H}_N + \int_0^1 ds \, s \, e^{-sA} [\mathcal{C}_N, A ] e^{sA} + \delta_4 \end{split} \end{equation} 
with
\[ \pm \delta_4 \leq  C \left[ \ell^{(\alpha-4)/2} + \ell^{(\alpha-\beta-\eps)/4} + \ell^{(5\alpha - 7\beta -2\eps)/4} + \ell^{(3\beta-\alpha-\eps)/4}   \right] (\cH_N +  \cN + 1) + C \ell^{-4\alpha} \cN^2 / N. \]
We computed the commutator $[\mathcal{C}_N , A]$ in \eqref{[CN, A]}. To deal with the two main contributions on the r.h.s. of \eqref{[CN, A]}, we switch to momentum space. For the first term, we find 
\begin{equation}\label{eq:last1} \begin{split} 
&\frac{1}{N} \int dx dy dz \ N^3 V(N(x-y)) \nu_H(x;y) \pn(y) \cgl(x-z) [a_z^* a_y + h.c.] \left(1- \frac{\mathcal{N}}{N}\right) \\
&= \int dp \left[ \frac{1}{N} \int_{|q|> \ell^{-\alpha}} dq \ \hat{V} ((p-q)/N) \hat{G} (q) \right] \gl(p) \left[\hat{a}_p^* \hat{a} (\widehat{\ph^2_p})  \left(1-\frac{\mathcal{N}}{N}\right) + h.c. \right].
\end{split} \end{equation} 
With \eqref{eq:Vfa0} , we can estimate  
\[ \Big| \frac{1}{N} \int_{|q| > \ell^{-\alpha}} dq \, \hat{V} ((p-q)/N) \hat{G} (q)  - (8\pi \frak{a}_0 - \hat{V} (0)) \Big| \leq C (|p| + \ell^{-2\alpha})/N. \]
Since moreover (using $0 \leq 1- \gl (p) \leq \max(1, \ell^{2\beta} p^2)$) 
\[ \int dp \, |1- \gl(p)| |\langle \xi, \hat{a}_p^* \hat{a} (\widehat{\pn^2}_p) (1-\cN/N) \xi \rangle | \leq \int dp \, |p| \ell^{\beta} \| \hat{a}_p \xi \| \| \hat{a} (\widehat{\ph^2_p}) \xi \| \leq \ell^{\beta} \| \cK^{1/2} \xi \|  \| \cN^{1/2} \xi \| , \] 
we conclude from (\ref{eq:last1}), switching back to position space, that 
\[ \begin{split} &\frac{1}{N} \int dx dy dz \ N^3 V(N(x-y)) \nu_H(x;y) \pn(y) \cgl(x-z) [a_z^* a_y + h.c.] \left(1- \frac{\mathcal{N}}{N}\right) \\ &= (8\pi \frak{a}_0 - \hat{V} (0)) \int dp \, \big[ \hat{a}_p \hat{a}^* (\widehat{\ph^2}_p) + h.c.] + \delta_5 = 2 (8\pi \frak{a}_0 - \hat{V} (0)) \int dx\,  \ph_0^2 (x) a_x^* a_x + \delta_5, \end{split} \]
where 
\[ |\langle \xi , \delta_5 \xi \rangle | \leq C \ell^{\beta} \| \cK^{1/2} \xi \| \| \cN^{1/2} \xi \| + C \| \cN \xi \|^2 / N. \]
Similarly, we can also handle the second term on the r.h.s. of \eqref{[CN, A]}. 
We conclude that
\[ [ \mathcal{C}_N , A] = 4 (8 \pi \frak{a}_0 - \hat{V} (0) ) \int dx \, \ph_0^2 (x) a_x^* a_x + \delta_6 \]
with 
\[ \begin{split} |\langle \xi,  \delta_6 \xi \rangle | \leq \; &C (\ell^
\beta + \ell^{3(\alpha-\beta)/2}) \| (\cK+\cN)^{1/2} \xi \| \| \cN^{1/2} \xi \|  \\ &+ C \ell^{(\alpha-\beta)/2} \| (\cK_{\ell^{-\beta-\eps}} + \cV_N + \cN)^{1/2} \xi \|^2 + C \| \cN \xi \|^2 / N . \end{split} \] 
Inserting in (\ref{eq:JN2}), we find, with \eqref{int F(u) e -A au* au e A}, 
(\ref{rough bound e-A (N+1) eA with l}), \eqref{rough estimate e -sA K theta e sA}, (\ref{rough estimate e -A VN eA with l}), (\ref{rough estimate e -sA K e sA with l}),
\[ \begin{split} 
\cJ_N = \; &\mathcal{D}_N + \mathcal{Q}_N + \mathcal{H}_N +2 (8\pi  \frak{a}_0 - \hat{V} (0) ) \int dx \, \ph_0^2 (x) a_x^* a_x + \delta_7 \end{split} \]
with
\[ \pm \delta_7 \leq  C \left[ \ell^{(\alpha-4)/2} + \ell^{(\alpha-\beta-\eps)/4} + \ell^{(5\alpha - 7\beta -2\eps)/4} + \ell^{(3\beta-\alpha-\eps)/4}   \right] (\cH_N + \cN + 1) + C \ell^{-4\alpha} \cN^2 / N. \]

Under the assumption $\alpha > 4$, $7\beta/5 < \alpha < 3 \beta$, we can find $\eps > 0$ small enough so that $\kappa = \min ((\alpha-4)/2, (\alpha-\beta-\eps)/4, (5\alpha-7\beta-2\eps)/4 , (3\beta-\alpha-\eps)/4) > 0$. Inserting $\mathcal{D}_N , \mathcal{Q}_N$ as in (\ref{eq:def-Geff}), we arrive therefore at
\begin{equation}\label{eq:JNell-f1} 
\begin{split} 
\cJ_N &\geq N \mathcal{E}^{\text{GP}}(\pn)  - \varepsilon_{GP} \mathcal{N} 
+4\pi \frak{a}_0 \int dx dy \ \widecheck{\chi}_{H^c}(x-y) \pn(x) \pn(y) [b_x b_y + b_x^* b_y^*] \\
&\quad + 16\pi \frak{a}_0 \int dx \ \pn(x)^2 a_x^* a_x 
+ (1- C\ell^{\kappa}) \mathcal{H}_N - C\ell^\kappa (\cN + 1) - C \ell^{-4\alpha} (\mathcal{N}+1)^2 / N \, .
\end{split}
\end{equation} 
Next, we observe that 
\[ 0 \leq 4\pi \frak{a}_0 \int dx dy \check{\chi}_{H^c} (x-y) \ph_0 (x) \ph_0 (y) [b_x + b_x^* ] [ b_y^* + b_y]. \]
This implies that 
\[ \begin{split} &4\pi \frak{a}_0 \int dx dy \check{\chi}_{H^c} (x-y)  \ph_0 (x) \ph_0 (y) [b_x b_y + b_x^* b_y^*] \\ &\geq - 8\pi \frak{a}_0 \int dx dy \check{\chi}_{H^c} (x-y) \ph_0 (x) \ph_0 (y) b_x^* b_y - 4 \pi \frak{a}_0 \int dx dy \check{\chi}_{H^c} (x-y) \ph_0 (x) \ph_0 (y) [ b_x, b_y^* ] \\ &\geq - 8\pi \frak{a}_0 \int dx dy \check{\chi}_{H^c} (x-y) \ph_0 (x) \ph_0 (y) a_x^* a_y - C \ell^{-3\alpha} - C \cN^2 / N, \end{split} \]
where in the last step we used the commutation relations (\ref{eq:comm-b}) (and we replaced $b_x^*, b_y$ by $a_x^*, a_y$). Since, switching to momentum space,  
\[ \begin{split}  \int dx dy \check{\chi}_{H} (x-y) &\ph_0 (x) \ph_0 (y) \langle \xi, a_x^* a_y \xi \rangle  \\ &= \int dp \chi_{H} (p) \| \hat{a} (\hat{\ph}_p) \xi \|^2 \leq \ell^{2\alpha} \int dp \, p^2  \| \hat{a} (\hat{\ph}_p) \xi \|^2 \leq \ell^{2\alpha} \| (\cK+\cN)^{1/2} \xi \|^2,  \end{split} \] 
with (\ref{eq:aph2}), we conclude that 
\[ \begin{split} 4\pi \frak{a}_0 \int dx dy & \check{\chi}_{H^c} (x-y)  \ph_0 (x) \ph_0 (y) [b_x b_y + b_x^* b_y^*] \\ & \geq - 8\pi \frak{a}_0 \int dx \, \ph_0 (x)^2 a_x^* a_x - C \ell^{-3\alpha} - C \cN^2 / N - C \ell^{2\alpha} (\cK+ \cN).  \end{split} \]
Inserting in (\ref{eq:JNell-f1}), we arrive at 
\[ \begin{split}  \cJ_N \geq \; &N \cE^{\text{GP}} (\ph_0) + (1 - C \ell^\kappa) d\Gamma (-\Delta + V_\text{ext} +8\pi \frak{a}_0 |\ph_0 (x)|^2 - \eps_{GP} ) \\ &- C \ell^\kappa \cN - C \ell^{-3\alpha} - C \ell^{-4\alpha} \cN^2 / N, \end{split} \]
which implies (\ref{eq:propJNell}), if $\ell > 0$ is small enough.

\appendix
\section{Properties of the Gross-Pitaevskii Functional}\label{apx:gpfunctional}

In this appendix we collect several well-known results about the Gross-Pitaevskii functional $\cE_{GP}$, defined in equation \eqref{eq:defGPfunctional}. Let us recall that $\cE_{GP}:\cD_{GP}\to \mathbb{R}$ is given by
		\[
		 \mathcal{E}_{GP}(\varphi) = \int_{\bR^3} \left( \vert \nabla \pn(x) \vert^2 + V_{ext}(x) \vert \pn(x)\vert^2 + 4\pi a_0 \vert \pn(x)\vert^4 \right) dx 
		\]
with domain
		 \[ \mathcal{D}_{GP} =\big\{ \varphi \in H^1(\bR^3)\cap L^4(\bR^3): \ V_{ext}\vert \varphi\vert^2\in L^1(\bR^3)\big\}. \]
Recall, moreover, assumption $(2)$ in Eq. \eqref{eq:asmptsVVext} on the external potential $V_{ext}$. The following was proved in \cite[Theorems 2.1, 2.5 \& Lemma A.6]{LSY}.
\begin{theorem} \label{thm:gpmin1}
	There exists a minimizer $\pn\in \cD_{GP}$ with $ \|\pn\|_2=1$ such that
			\[ \inf_{\psi\in \mathcal{D}_{GP} \ : \|\psi\|_2 =1} \mathcal{E}_{GP}(\psi) = \mathcal{E}_{GP}(\pn).  \]
	The minimizer $\pn$ is unique up to a complex phase, which can be chosen so that $\pn$ is strictly positive. Furthermore, the minimizer $\pn$ solves the Gross-Pitaevskii equation
	
	\begin{equation} \label{eq:gpeq}
	-\Delta \pn + V_{ext} \pn + 8\pi a_0 \vert \pn \vert^2 \pn = \eps_{GP}\pn,
	\end{equation}
with $\mu$ given by
	$$ \eps_{GP}= \mathcal{E}_{GP}(\pn) + 4\pi a_0 \Vert \pn \Vert_4^4.$$
	
Moreover, $\pn \in L^\infty(\bR^3)\cap C^1(\bR^3)$ and for every $\nu>0$ there exists $C_\nu$ (which only depends on $\nu$ and $\mathfrak{a}_0$) such that for all $x\in \bR^3$ it holds true that
	\begin{equation} \label{eq:expdecaypn}
	\vert \pn(x)\vert \leq C_\nu e^{-\nu \vert x \vert}.
	\end{equation} 
\end{theorem}
We denote by $ \pn$ in the following the unique, strictly positive minimizer of $\cE_{GP}$, subject to the contraint $ \|\pn\|_2=1$. In addition to Theorem \ref{thm:gpmin1}, we need to collect a few additional facts about the regularity of $\pn$. Before we do so, notice that the assumption
		\[ V_{ext} (x+y) \leq C (V_{ext}(x) + C) (V_{ext}(y)+C)\] 
implies that $V_{ext}$ has at most exponential growth, as $|x|\to \infty$. Indeed, by \eqref{eq:asmptsVVext}, we find $R>0$ such that $V_{ext}(x) >0$ for all $\vert x \vert \geq R$. Let $\tilde{C}$ be the maximum of $V_{ext}$ in the ball of radius $2R$ around the origin. For $\vert x \vert \geq R$, we pick $n\in \mathbb{N}$ such that $nR \leq \vert x \vert < (n+1)R$ and obtain 
\begin{align*}
\vert V_{ext}(x) \vert \leq C^n (V_{ext}(x/n)+C)^n 
\leq (C(\tilde{C}+C))^n 
\leq  (C(\tilde{C}+C))^{\vert x \vert/R}.
\end{align*}
Hence, $V_{ext}$ grows at most exponentially. In particular, by \eqref{eq:expdecaypn}, this implies that 
\begin{equation} \label{eq:Vextbnd}
\Vert V_{ext} \pn \Vert_\infty  \leq C.
\end{equation}

\begin{lemma}\label{lem:gpmin2}
Let $ V_{ext}$ satisfy the assumptions in \eqref{eq:asmptsVVext}. Then $\pn\in H^2(\mathbb{R}^3)\cap C^2(\mathbb{R}^3)$ and for every $\nu>0$ there exists $C_\nu>0$ such that for every $x\in \mathbb{R}^3$ we have
	\begin{equation} \label{exponential decay}
	\vert \nabla \pn (x) \vert, \ \vert \Delta \pn(x) \vert 
	\leq C_\nu e^{-\nu \vert x \vert}.
	\end{equation}
Moreover, if $\whpn$ denotes the Fourier transform of $\pn$, we have for all $ p\in\bR^3$ that
		\begin{equation} \label{eq:decphhat}
		| \whpn (p)|     \leq  \frac{C}{(1+ | p|)^4}.
		\end{equation}
\end{lemma}
\begin{proof}
By the previous Theorem \ref{thm:gpmin1}, the Gross-Pitaevskii equation  \eqref{eq:gpeq}, Eq. \eqref{eq:expdecaypn} and the fact that $V_{ext}, \nabla V_{ext}$ grow at most exponentially (by the assumptions \eqref{eq:asmptsVVext} and the previous remark), we obtain the exponential decay of $\Delta \pn$. Moreover, since $\pn\in L^\infty(\bR^3)$ and local Hölder continuity of $V_{ext}$, we get the local Hölder continuity of $\Delta \pn$. Elliptic regularity then implies that $\pn\in C^2(\mathbb{R}^3, \mathbb{R})$. 

Next, by \cite[Theorem 3.9]{GT}, if $u\in C^2(B_2(y))$ solves $ \Delta u=f$, then there exists a constant $C>0$, independent of $y\in\bR^3$, such that
	\begin{align*}
	\Vert \nabla u  \Vert_{L^\infty (B_1(y) )} \leq C \big(\Vert u \Vert_{L^\infty (B_2(y) )} + \Vert f \Vert_{L^\infty (B_2(y))}\big).
	\end{align*}
Here, $ B_1(y) $ and $B_2(y)$ denote the open balls of radius one and two, respectively, centered at $y\in\bR^3$. Applying this last bound to $\pn$ and using the exponential decay of $\pn, \Delta \pn$ implies that also $\nabla \pn$ has exponential decay. 

Finally, let us prove the decay estimate \eqref{eq:decphhat}. By Theorem \eqref{thm:gpmin1}, Eq. \eqref{eq:gpeq} and elliptic regularity theory, we conclude $\pn \in C^4(\mathbb{R}^3)$ and that
		\begin{align*}
		\Delta^2 \pn = - \Delta( \pn V_{ext}) + \eps_{GP}\Delta \pn + 8\pi a_0 (3 \pn^2 \Delta \pn + 6 \pn |\nabla \pn|^2).
		\end{align*}
Since, on the one hand, $\pn, \nabla \pn$ and $\Delta \pn$ all have exponential decay with arbitrary rate while $V_{ext}, \nabla V_{ext}, \Delta V_{ext}$ grow at most exponentially by assumption \eqref{eq:asmptsVVext}, we conclude that $\Delta^2\pn\in L^1(\mathbb{R}^3)$. This implies the estimate \eqref{eq:decphhat} by switching to Fourier space.
\end{proof}



\end{document}